\pdfoutput=1  

\documentclass[11pt]{article}
\usepackage[margin=.8in,left=.8in]{geometry}
\usepackage{amsmath}
\usepackage{amsfonts}
\usepackage{amssymb}
\usepackage{amsthm}
\usepackage{mathtools}
\usepackage{subcaption}
\usepackage{graphicx}
\usepackage{color}
\usepackage{xcolor}
\usepackage{xfrac}
\usepackage{stmaryrd}

\usepackage{hyperref}
\hypersetup{
  colorlinks = true,
  allcolors=.           
}

\usepackage[all]{xy}

\usepackage[titletoc]{appendix}

\usepackage{tikz}
\usetikzlibrary{cd}

\usepackage{wasysym}

\usepackage[safe]{tipa} 

\usepackage{lipsum}
\usepackage{adjustbox}

\definecolor{darkblue}{rgb}{0.05,0.25,0.65}
\definecolor{greenii}{RGB}{20,140,10}
\definecolor{lightgray}{rgb}{0.9,0.9,0.9}
\definecolor{orangeii}{RGB}{200,100,5}

\usepackage{multirow}

\usepackage{stmaryrd}    

\usepackage{enumerate} 

\usepackage{helvet}   

\usepackage{mathptmx}
\usepackage{amsmath}
\usepackage{graphicx}

\usepackage{floatflt}  
\usepackage{array}     
\newcolumntype{L}[1]{>{\raggedright\let\newline\\\arraybackslash\hspace{0pt}}m{#1}}
\newcolumntype{C}[1]{>{\centering\let\newline\\\arraybackslash\hspace{0pt}}m{#1}}
\newcolumntype{R}[1]{>{\raggedleft\let\newline\\\arraybackslash\hspace{0pt}}m{#1}}

\newcommand{\mapsup}{\rotatebox[origin=c]{90}{$\mapsto$}}

\usepackage[new]{old-arrows}   

\newdir{> }{{}*!/10pt/@{>}}

\usepackage{amssymb}

\def\acts{\raisebox{1pt}{\;\rotatebox[origin=c]{90}{$\curvearrowright$}}}

\DeclareRobustCommand{\rchi}{{\mathpalette\irchi\relax}}
\newcommand{\irchi}[2]{\raisebox{\depth}{$#1\chi$}} 

\makeatletter
\newif\if@sup
\newtoks\@sups
\def\append@sup#1{\edef\act{\noexpand\@sups={\the\@sups #1}}\act}%
\def\reset@sup{\@supfalse\@sups={}}%
\def\mk@scripts#1#2{\if #2/ \if@sup ^{\the\@sups}\fi \else%
  \ifx #1_ \if@sup ^{\the\@sups}\reset@sup \fi {}_{#2}%
  \else \append@sup#2 \@suptrue \fi%
  \expandafter\mk@scripts\fi}
\def\tensor#1#2{\reset@sup#1\mk@scripts#2_/}
\def\multiscripts#1#2#3{\reset@sup{}\mk@scripts#1_/#2%
  \reset@sup\mk@scripts#3_/}
\makeatother

\makeatletter
\newbox\slashbox \setbox\slashbox=\hbox{$/$}
\def\itex@pslash#1{\setbox\@tempboxa=\hbox{$#1$}
  \@tempdima=0.5\wd\slashbox \advance\@tempdima 0.5\wd\@tempboxa
  \copy\slashbox \kern-\@tempdima \box\@tempboxa}
\def\slash{\protect\itex@pslash}
\makeatother

\def\clap#1{\hbox to 0pt{\hss#1\hss}}
\def\mathllap{\mathpalette\mathllapinternal}
\def\mathrlap{\mathpalette\mathrlapinternal}
\def\mathclap{\mathpalette\mathclapinternal}
\def\mathllapinternal#1#2{\llap{$\mathsurround=0pt#1{#2}$}}
\def\mathrlapinternal#1#2{\rlap{$\mathsurround=0pt#1{#2}$}}
\def\mathclapinternal#1#2{\clap{$\mathsurround=0pt#1{#2}$}}

\let\oldroot\root
\def\root#1#2{\oldroot #1 \of{#2}}
\renewcommand{\sqrt}[2][]{\oldroot #1 \of{#2}}

\DeclareSymbolFont{symbolsC}{U}{txsyc}{m}{n}
\SetSymbolFont{symbolsC}{bold}{U}{txsyc}{bx}{n}
\DeclareFontSubstitution{U}{txsyc}{m}{n}

\DeclareSymbolFont{stmry}{U}{stmry}{m}{n}
\SetSymbolFont{stmry}{bold}{U}{stmry}{b}{n}

\DeclareFontFamily{OMX}{MnSymbolE}{}
\DeclareSymbolFont{mnomx}{OMX}{MnSymbolE}{m}{n}
\SetSymbolFont{mnomx}{bold}{OMX}{MnSymbolE}{b}{n}
\DeclareFontShape{OMX}{MnSymbolE}{m}{n}{
    <-6>  MnSymbolE5
   <6-7>  MnSymbolE6
   <7-8>  MnSymbolE7
   <8-9>  MnSymbolE8
   <9-10> MnSymbolE9
  <10-12> MnSymbolE10
  <12->   MnSymbolE12}{}

\makeatletter
\def\Decl@Mn@Delim#1#2#3#4{%
  \if\relax\noexpand#1%
    \let#1\undefined
  \fi
  \DeclareMathDelimiter{#1}{#2}{#3}{#4}{#3}{#4}}
\def\Decl@Mn@Open#1#2#3{\Decl@Mn@Delim{#1}{\mathopen}{#2}{#3}}
\def\Decl@Mn@Close#1#2#3{\Decl@Mn@Delim{#1}{\mathclose}{#2}{#3}}
\Decl@Mn@Open{\llangle}{mnomx}{'164}
\Decl@Mn@Close{\rrangle}{mnomx}{'171}
\Decl@Mn@Open{\lmoustache}{mnomx}{'245}
\Decl@Mn@Close{\rmoustache}{mnomx}{'244}
\makeatother

\makeatletter
\DeclareRobustCommand\widecheck[1]{{\mathpalette\@widecheck{#1}}}
\def\@widecheck#1#2{%
    \setbox\z@\hbox{\m@th$#1#2$}%
    \setbox\tw@\hbox{\m@th$#1%
       \widehat{%
          \vrule\@width\z@\@height\ht\z@
          \vrule\@height\z@\@width\wd\z@}$}%
    \dp\tw@-\ht\z@
    \@tempdima\ht\z@ \advance\@tempdima2\ht\tw@ \divide\@tempdima\thr@@
    \setbox\tw@\hbox{%
       \raise\@tempdima\hbox{\scalebox{1}[-1]{\lower\@tempdima\box
\tw@}}}%
    {\ooalign{\box\tw@ \cr \box\z@}}}
\makeatother

\makeatletter
\def\udots{\mathinner{\mkern2mu\raise\p@\hbox{.}
\mkern2mu\raise4\p@\hbox{.}\mkern1mu
\raise7\p@\vbox{\kern7\p@\hbox{.}}\mkern1mu}}
\makeatother

\def\1{{\bf 1}}

\def\<{\langle}
\def\>{\rangle}

\renewcommand{\(}{\begin{equation}}
\renewcommand{\)}{\end{equation}}
\newcommand{\bea}{\begin{eqnarray*}}
\newcommand{\eea}{\end{eqnarray*}}

\usepackage{cleveref}

\crefformat{section}{\S#2#1#3} 
\crefformat{subsection}{\S#2#1#3}
\crefformat{subsubsection}{\S#2#1#3}

\theoremstyle{italics}
\newtheorem{theorem}{Theorem}[section]
\newtheorem{lemma}[theorem]{Lemma}
\newtheorem{prop}[theorem]{Proposition}
\newtheorem{cor}[theorem]{Corollary}

\theoremstyle{definition}
\newtheorem{defn}[theorem]{Definition}

\newtheorem{example}[theorem]{Example}

\newtheorem{remark}[theorem]{Remark}
\newtheorem{note[theorem]}{Note}

\usepackage{amsfonts}
\usepackage{colortbl}

\renewcommand{\emph}{\textit}

\begin{document}

\title{Twistorial Cohomotopy implies Green-Schwarz anomaly cancellation}

\author{
  Domenico Fiorenza,
  \thanks{
    Domenico Fiorenza, {\it Dipartimento di Matematica,
    Sapienza Universit{\`a} di Roma, Piazzale Aldo Moro 2, 00185 Rome, Italy.}
    {\tt fiorenza@mat.uniroma1.it}
  },
  \;\;\;
  Hisham Sati
  \thanks{
   Hisham Sati, {\it Mathematics, Division of Science,
   and Center for Quantum and Topological Systems (CQTS), NYUAD Research Institute,
   New York University Abu Dhabi, UAE},
   {\tt hsati@nyu.edu}
  },
  \;\;\;
  Urs Schreiber
  \thanks{
     Urs Schreiber, {\it Mathematics, Division of Science,
     and Center for Quantum and Topological Systems (CQTS), NYUAD Research Institute,
     New York University Abu Dhabi, UAE;
     on leave from Czech Academy of Science, Prague.}
     {\tt us13@nyu.edu}
  }
}

\maketitle

\begin{abstract}

\noindent We characterize the integral cohomology and
the rational homotopy type of the
maximal Borel-equivariantization of the combined Hopf/twistor fibration,
and find that subtle relations satisfied by the
cohomology generators are
just those that govern Ho{\v r}ava-Witten's
proposal for the extension of the Green-Schwarz mechanism
from heterotic string theory to heterotic M-theory.
We discuss how this squares with the {\it Hypothesis H}
that the elusive mathematical foundation of M-theory is based
on charge quantization in tangentially twisted unstable Cohomotopy theory.

\end{abstract}

\medskip

\tableofcontents

\newpage

\section{Introduction and overview}

  \noindent
  {\bf The Green-Schwarz mechanism in heterotic M-theory.}
  At the heart of M-theory
  (the conjectural non-perturbative completion of type IIA string theory,
  see \cite{Duff96}\cite{Duff98}\cite{Duff99})
  is the proposal \cite[(1.2)]{Witten97}\cite[(1.2)]{Witten97b}
  that the \emph{difference} of the classes of:
\begin{itemize}
\vspace{-3mm}
  \item[{\bf (i)}] the flux density $G_4$ of
  the higher gauge field of M-theory (the \emph{C-field},
  or \emph{3-index A-field} \cite{CremmerJuliaScherk78}),
\vspace{-3mm}
  \item[{\bf (ii)}] $\sfrac{1}{4}$th
  of the Pontrjagin form of the spin connection $\omega$ on
  spacetime $Y^{11}$ (e.g. \cite[\S XII.4]{KobayashiNomizu63}\cite[p. 10]{GS-KO}),
  \end{itemize}

 \vspace{-3mm}
  \noindent
   lifts to an integral class:
  \vspace{-.4cm}
  \begin{equation}
    \label{BareShiftedFluxQuantization}
    \overset{
      \mathclap{
      \raisebox{3pt}{
        \tiny
        \color{darkblue}
        \bf
        \begin{tabular}{c}
          C-field
          \\
          4-flux density
        \end{tabular}
      }
      }
    }{
     [G_4]
    }
    \;\;\;-\;\;\;
    \overset{
      \mathclap{
      \raisebox{3pt}{
        \tiny
        \color{darkblue}
        \bf
        \begin{tabular}{c}
          $\sfrac{1}{2}$ gravitational
          \\
          instanton density
        \end{tabular}
      }
      }
    }{
      \big[\tfrac{1}{4}p_1(\omega)\big]
    }
    \;\;\in
    \xymatrix{
      \overset{
        \raisebox{3pt}{
          \tiny
          \color{darkblue}
          \bf
          \begin{tabular}{c}
            integral cohomology of
            \\
            11-d spacetime
          \end{tabular}
        }
      }{
      H^4
      \big(
        Y^{11};\, \mathbb{Z}
      \big)
      }
      \ar[rr]^-{
        \mbox{
          \tiny
          \color{greenii}
          \bf
          rationalization
        }
      }
      &&
      \overset{
        \raisebox{3pt}{
          \tiny
          \color{darkblue}
          \bf
          \begin{tabular}{c}
            real cohomology of
            \\
            11-d spacetime
          \end{tabular}
        }
      }{
        H^4
        \big(
          Y^{11};\, \mathbb{R}
        \big).
      }
    }
  \end{equation}

  \vspace{-2mm}
  One motivation for this proposal
  \cite[\S 2.1]{Witten97}
  comes from heterotic M-theory
  (Ho{\v r}ava-Witten theory
  \cite{HoravaWitten95} \cite{Witten96}\cite{HoravaWitten96b}\cite{DOPW99}\cite{DOPW00}\cite{Ovrut02},
  the conjectural non-perturbative
  completion of heterotic string theory
  \cite{GHMR85}\cite{GHMR86}\cite{AGLP12}).
  Here the celebrated
  (``first superstring revolution'', see \cite{Schwarz07})
  {\it Green-Schwarz anomaly cancellation mechanism}
  \cite{GreenSchwarz84}\cite{CHSW85}
  (review in \cite[\S 2.2]{Witten99}\cite{Freed00})
  in heterotic string theory,
  which in itself is understood clearly,
  is argued to imply, upon lift to heterotic M-theory,
  that the restriction of the 4-flux $G_4$
  to an MO9-plane $X^{10}$ inside 11-dimensional spacetime $Y^{11}$ satisfies
  this relation \cite[(1.13)]{HoravaWitten96b}:\footnote{
    Our normalization convention for $G_4$ absorbs a factor of
    $\sfrac{-1}{2\pi}$ compared to \cite{Witten96}.
  }
  \vspace{-2mm}
  \begin{equation}
    \label{FluxQuantizationOnMO9}
    \overset{
      \mathclap{
      \raisebox{3pt}{
        \tiny
        \color{darkblue}
        \bf
        \begin{tabular}{c}
          C-field 4-flux density
          \\
          restricted to $X^{10}$
        \end{tabular}
      }
      }
    }{
      \big[
        G_4
      \big]\lvert_{X^{10}}
    }
    \;\;\;=\;
    \overset{
      \mathclap{
      \raisebox{3pt}{
        \tiny
        \color{darkblue}
        \bf
        \begin{tabular}{c}
          $\sfrac{1}{2}$ gravitational
          \\
          instanton density
        \end{tabular}
      }
      }
    }
    {
      \big[\tfrac{1}{4}p_1(\omega)\big]\vert_{X^{10}}
    }
    \;-\;
    \overset{
      \mathclap{
      \raisebox{4pt}{
        \tiny
        \color{darkblue}
        \bf
        \begin{tabular}{c}
          gauge instanton
          \\
          density on $X^{10}$
        \end{tabular}
      }
      }
    }{
      [c_2(A)]
    }
    \;\;
    \in
    \;
    \xymatrix{
      \overset{
        \mathclap{
        \raisebox{3pt}{
          \tiny
          \color{darkblue}
          \bf
          \begin{tabular}{c}
            integral cohomology
            \\
            of 10d MO9-plane
          \end{tabular}
        }
        }
      }{
        H^4
        \big(
          X^{10};\,
          \mathbb{Z}
        \big)
      }
      \ar[rr]^-{
        \mbox{
          \tiny
          \color{darkblue}
          \bf
          rationalization
        }
      }
      &&
      \overset{
        \mathclap{
        \raisebox{3pt}{
          \tiny
          \color{darkblue}
          \bf
          \begin{tabular}{c}
            real cohomology
            \\
            of 10d MO9-plane
          \end{tabular}
        }
        }
      }{
      H^4
      \big(
        X^{10};\,
        \mathbb{R}
      \big)
      }\,,
    }
  \end{equation}

    \vspace{-1mm}
\noindent  where $A$ (the \emph{gauge field})
  is a connection on an $E_8$-principal bundle over $X^{10}$,
  and $c_2(A)$ is its second Chern-form.
  But the summand $[c_2(A)]$ \emph{is} integral by itself: it is
  the real image of the second Chern class
  of the $E_8$-bundle.
  Therefore, \eqref{FluxQuantizationOnMO9} implies that
  \eqref{BareShiftedFluxQuantization} holds at least
  upon restriction to MO9-planes; and it
  suggests \cite{DMW00}\cite{EvslinSati}\cite{DFM03} \cite{Sati06}\cite{FSS14a} that the integral
  4-class in \eqref{BareShiftedFluxQuantization} is to be thought of
  as the second Chern class of an extension $\widetilde A$ of the $E_8$-gauge field
  from $X^{10}$ to all of $Y^{11}$:

  \vspace{-2mm}
  \begin{equation}
    \label{HeteroticMTheoryShiftedFluxQuantization}
    \overset{
      \mathclap{
      \raisebox{3pt}{
        \tiny
        \color{darkblue}
        \bf
        \begin{tabular}{c}
          C-field
          \\
          4-flux density
        \end{tabular}
      }
      }
    }{
     [G_4]
    }
    \;\;\;-\;
    \overset{
      \mathclap{
      \raisebox{3pt}{
        \tiny
        \color{darkblue}
        \bf
        \begin{tabular}{c}
          $\sfrac{1}{2}$ gravitational
          \\
          instanton density
        \end{tabular}
      }
      }
    }{
      \big[\tfrac{1}{4}p_1(\omega)\big]
    }
    \;=\;
    \overset{
      \raisebox{3pt}{
        \tiny
        \color{orangeii}
        \bf
        \begin{tabular}{c}
          gauge instanton density of
          \\
          auxilary gauge field on $X^{11}$
        \end{tabular}
      }
    }{
      \underset{ =: a}{
        -
        \underbrace{
          \big[c_2\big(\,\widetilde A \, \big)\big]
        }
      }
    }
    \;\;\in\;
    \xymatrix{
      \overset{
        \raisebox{3pt}{
          \tiny
          \color{darkblue}
          \bf
          \begin{tabular}{c}
            integral cohomology of
            \\
            11-d spacetime
          \end{tabular}
        }
      }{
      H^4
      \big(
        Y^{11};\, \mathbb{Z}
      \big)
      }
      \ar[rr]^-{
        \mbox{
          \tiny
          \color{darkblue}
          \bf
          rationalization
        }
      }
      &&
      \overset{
        \raisebox{3pt}{
          \tiny
          \color{darkblue}
          \bf
          \begin{tabular}{c}
            real cohomology of
            \\
            11-d spacetime
          \end{tabular}
        }
      }{
        H^4
        \big(
          Y^{11};\, \mathbb{R}
        \big).
      }
    }
  \end{equation}

  \noindent {\bf Open problem.}
  Despite the tight web of hints and consistency checks like these,
  actually formulating M-theory remains an open problem
  (e.g., \cite[6]{Duff96}\cite[p. 2]{HoweLambertWest97}\cite[p. 6]{Duff98}\cite[p. 2]{NicolaiHelling98}\cite[p. 330]{Duff99}\cite[12]{Moore14}\cite[p. 2]{ChesterPerlmutter18}\cite[$@$21:15]{Witten19}\cite[$@$17:14]{Duff19}).
  In particular:

  \begin{itemize}
  \vspace{-3.5mm}
  \item[{\bf (i)}] The conditions \eqref{BareShiftedFluxQuantization} and
  \eqref{FluxQuantizationOnMO9} had
  not actually been derived from any theory
  (see the comments around \cite[(1.13)]{HoravaWitten96b}
  and \cite[(2.1)]{Witten97}).

   \vspace{-3mm}
  \item[{\bf (ii)}] The ontology of the gauge field
  on $Y^{11}$ in \eqref{HeteroticMTheoryShiftedFluxQuantization}
  has remained mysterious, as no such gauge field is seen in
  11-dimensional supergravity
  \cite{CremmerJuliaScherk78}\cite{DAuriaFre82}\cite{CDF91},
  which, however, is famously argued to be the low-energy limit
  of M-theory \cite{Witten95}.
  \end{itemize}

  \noindent {\bf By Charge quantization in generalized cohomology?}
  On the other hand,
  the Green-Schwarz mechanism in perturbative string theory
  is well understood
  as an index-theoretic phenomenon, resulting from charge quantization in
  a generalized cohomology theory, namely in K-theory
  \cite{Freed00}\cite{Clingher05}\cite{Bunke11}.
  This mathematical understanding has been most
  fruitful, spawning understanding of
  elliptic genera (e.g. \cite{Han}\cite{ChenHanZhang11}),
  twisted higher bundles/gerbes (e.g. \cite{SSS09}\cite{Waldorf13}),
  Hermitian and generalized geometry (e.g. \cite{Garcia}), and more.

  There have been various proposals
  for lifting this situation to
  heterotic M-theory, understanding also
  the shifted integrality condition \eqref{HeteroticMTheoryShiftedFluxQuantization}
  as an effect of charge quantization
  in \emph{some} generalized cohomology theory
  \cite{DFM03}\cite{HopkinsSinger05} \cite{Sati05a}\cite{Sati05b}\cite{Sati05c}\cite{tcu}\cite{FSS14a}\cite{FSS14b},
  which then would control M-theory
  in generalization of how K-theory controls string theory
  (for the latter, see \cite{GS-RR} and references therein).
  However, while advancements in understanding have certainly been made,
  the situation had remained inconclusive.

\medskip

\noindent {\bf Non-abelian characters.}
Our strategy is to invoke the further generalization of generalized cohomology
to \emph{non-abelian cohomology}
\cite{Toen02}\cite{RobertsStevenson12}\cite{SSS12}\cite{NSS12a}\cite{NSS12b}\cite{FSS19b}\cite{SS20b}.
In \cref{NonAbelianDoldChernCharacter} below we discuss
how the \emph{Chern-Dold character} in generalized cohomology
(\cite{Buchstaber70}, see \cite[\S 2.1]{LindSatiWesterland16})
extends to non-abelian cohmology theories, where it is given by
passage to {\it rational homotopy theory}, here over the real numbers \cite[\S 3.2]{FSS20c}\footnote{
For general introduction and review of
rational homotopy theory \cite{Quillen69} and its Sullivan models \cite{Sullivan77}
see for instance \cite{FHT00}\cite{Hess07},
for brief discussion in our context see \cite[\S A]{FSS16a}\cite[\S 2.2]{FSS17}
and for extensive details see the companion article \cite{FSS20c}.
As in all applications to differential geometry and physics, we consider
rational homotopy theory over the real numbers \cite[Rem. 3.64]{FSS20c},
as in \cite{DGMS75}\cite{GriffithMorgan13}.
Notice that in the supergravity literature
these real Sullivan models are known as
``FDA''s (following \cite{vanNieuwenhuizen82}\cite{DAuriaFre82}); for
details and translation see
\cite{FSS13}\cite{FSS16a}\cite{FSS16b}\cite{HSS18}\cite{BMSS19}
with review in \cite{FSS19a}.
}:
\vspace{-.4cm}
\begin{equation}
  \label{ChernCharacterInGeneralizedCohomology}
  \xymatrix@R=-12pt{
    \overset{
      \mathclap{
      \raisebox{3pt}{
        \tiny
        \color{darkblue}
        \bf
        \begin{tabular}{c}
          $E$-cohomology
        \end{tabular}
      }
      }
    }{
      E^\bullet(X)
    }
    \ar@{}[d]|-{
      \rotatebox[origin=c]{-90}{
        $\simeq$
      }
    }
    \ar[rrrr]^-{
      \mbox{
        \tiny
        \color{greenii}
        \bf
        \begin{tabular}{c}
          Non-abelian
          character map
        \end{tabular}
      }
    }
    &&&&
    \overset{
      \mathclap{
      \raisebox{3pt}{
        \tiny
        \color{darkblue}
        \bf
        \begin{tabular}{c}
          Non-abelian
          real cohomology
        \end{tabular}
      }
      }
    }{
    H^\bullet
    \big(
      X;\,
      L_{\mathbb{R}} E
    \big)
    }
    \ar@{}[d]|-{
      \rotatebox[origin=c]{-90}{
        $\simeq$
      }
    }
    \\
    \pi_{-\bullet}
    \,
    \mathrm{Maps}
    \big(
      X,
      \,
      \underset{
        \mathclap{
        \mbox{
          \tiny
          \color{darkblue}
          \bf
          classifying space
        }
        }
      }{
        \underbrace{
          E
        }
      }
    \big)
    \ar[rrrr]_-{
      \mbox{
        \tiny
        \color{greenii}
        \bf
        $\mathbb{R}$-rationalization
      }
    }
    &&&&
    \pi_{-\bullet}
    \,
    \mathrm{Maps}
    \big(
      X
      ,\,
      \underset{
        \mathclap{
        \mbox{
          \tiny
          \color{darkblue}
          \bf
          \begin{tabular}{c}
            $\mathbb{R}$-rationalized
            classifying space
          \end{tabular}
        }
        }
      }{
        \underbrace{
          L_{\mathbb{R}} E
        }
      }
    \,\big)
  }
\end{equation}
\vspace{-.4cm}

For $X$ a smooth manifold, the right hand side of
\eqref{ChernCharacterInGeneralizedCohomology} is
a subquotient of the set of differential forms on $X$
(Prop. \ref{NonAbelianDeRhamTheorem});
thus the image of the nonabelian character identifies
equivalence classes of differential forms satisfying certain conditions.
If these forms are interpreted as flux densities, then
these are {\bf non-abelian charge-quantization} conditions.
For example, in the abelian case $E = \mathrm{KU}, \mathrm{KO}$, the
Chern-Dold character \eqref{ChernCharacterInGeneralizedCohomology}
reduces to the traditional Chern character in K-theory
\cite[Ex. 4.13, 4.14]{FSS20c}, and the
corresponding charge quantization condition is that thought to
hold for RR-fields/D-brane charges in type II/I string theory
(for more discussion see \cite{GS-KO}\cite{GS-RR}).

\medskip

\noindent {\bf Charge quantization in J-twisted Cohomotopy.}
The most fundamental non-abelian cohomology theory is
\emph{Cohomotopy theory}
\cite{Borsuk36}\cite{Spanier49}\cite{Peterson56}\cite{Taylor09}\cite{KMT12}
whose classifying spaces are the $n$-spheres
$S^n \simeq B(\Omega S^n)$.
Accordingly, twisted Cohomotopy is classified by
spherical fibrations, and we say {\it J-twisted Cohomotopy}
\cite{FSS19b} \cite{FSS19c}
for twisting by the unit sphere fibration in the tangent bundle.
The main theorem of \cite[3.9]{FSS19b} says that the
Chern-Dold
character \eqref{ChernCharacterInGeneralizedCohomology}
in J-twisted 4-Cohomotopy encodes
Witten's shifted C-field flux quantization condition \eqref{BareShiftedFluxQuantization}:

\vspace{-1.0cm}
\begin{equation}
  \label{ShiftedFluxQuantizationFromChernCharacterInJTwistedCohomotopy}
  \xymatrix@C=34pt{
    \overset{
      \raisebox{3pt}{
        \tiny
        \color{darkblue}
        \bf
        \begin{tabular}{c}
          J-twisted
          \\
          4-Cohomotopy
        \end{tabular}
      }
    }{
    \pi^{\tau}
      \underset{
        \mathclap{
        \raisebox{-3pt}{
          \tiny
          \color{darkblue}
          \bf
          \begin{tabular}{c}
            manifold
            \\
            with
            \\tangential
            $\mathrm{Sp}(2)$-structure $\tau$
          \end{tabular}
        }
        }
      }{
       \big(
         X
       \big)
      }
    }
    \ar[rr]^-{
      \mbox{
        \tiny
        \color{darkblue}
        \bf
        \begin{tabular}{c}
          Non-abelian
          \\
          character map
        \end{tabular}
      }
    }_-{\mathrm{ch}}
    &&
    \left\{
      \!\!\!
      {\begin{array}{c}
        G_4,
        \\
        G_7
      \end{array}}
      \!\!\!
      \in \Omega^\bullet(X)
    \,\left\vert\,
      {\begin{aligned}
        d\, \phantom{2}
        \overset{
          \mathclap{
          \raisebox{3pt}{
            \tiny
            \color{darkblue}
            \bf
            \begin{tabular}{c}
              C-field
              \\
              4-flux
            \end{tabular}
          }
          }
        }{
          G_4
        }
        & = 0
        \,,
        \;\;
        \overset{
          \raisebox{3pt}{
            \tiny
            \color{darkblue}
            \bf
            shifted flux quantization \eqref{BareShiftedFluxQuantization}
          }
        }{
          [G_4] - [\tfrac{1}{4}p_1(\omega)]
          \;\in\;
          H^4
          \big(
            X;\, \mathbb{Z}
          \big)
        }
        \\
        d\,
        2G_7
        & =
        -
        \big(
          G_4 - \tfrac{1}{4}p_1(\omega)
        \big)
        \wedge
        \big(
          G_4 + \tfrac{1}{4}p_1(\omega)
        \big)
        \\
        \mathclap{
        \!\!\!\!\!\!\!\!\!\!
        \raisebox{3pt}{
          \tiny
          \color{darkblue}
          \bf
          \begin{tabular}{c}
            dual
            \\
            7-flux
          \end{tabular}
        }
        }
        & \phantom{=}\;
        -
        \tfrac{1}{2}
        \big(
          p_2(\omega)
          -
          \tfrac{1}{4}
          \big(p_1(\omega)\big)^2
        \big)
      \end{aligned}}
    \right.
    \right\}
  }
\end{equation}
\vspace{-.2cm}

In fact, further constraints are implied, matching a whole list
of further topological conditions expected in M-theory
(see \cite[Table 1]{FSS19b}). This motivates,
following \cite[\S 2.5]{Sati13}, {\bf Hypothesis H}:
\emph{Charge quantization in M-theory happens in
J-twisted Cohomotopy theory}
\cite{FSS19b}\cite{FSS19c}\cite{SS19a}\cite{BSS19}\cite{SS19b}\cite{SS20b}.
While J-twisted Cohomotopy in degree 4 alone \eqref{ShiftedFluxQuantizationFromChernCharacterInJTwistedCohomotopy}
does not reflect the
appearance of a heterotic gauge field as in \eqref{FluxQuantizationOnMO9},
its lifting to 7-Cohomotopy, through the quaternionic Hopf fibration
$h_{\mathbb{H}}$,
does encode structure expected on heterotic M5-branes
\cite{FSS19c}\cite{FSS19d}\cite{FSS20a}\cite{FSS20b}.
Therefore we turn attention to an intermediate lift:

\medskip

\noindent {\bf Charge quantization in Twistorial Cohomotopy.}
In between J-twisted Cohomotopy in degrees 4 and 7,
when lifting along the quaternionic Hopf fibration,
we find \emph{Twistorial Cohomotopy} (\cref{TwistorialCohomotopyTheory}):
The twisted non-abelian cohomology theory
classified by the Borel-equivariantized \emph{twistor space} $\mathbb{C}P^3$
(\cref{ConstructionOfBorelEquivariantHopfTwistorFibration}).
Our main Theorem \ref{SullivanModelOfParametrizedHopfTwistorFibrations},
illustrated in \hyperlink{FigureC}{\it Figure I}, implies
(Corollary \ref{ShiftedIntegralityUnderTTheoreticChernCharacter}) that
charge quantization in Twistorial Cohomotopy
\\
\noindent {\bf (i)} makes the class of an
$\mathrm{S}\big(\mathrm{U}(1)^2\big)$-gauge field
$\widetilde A$ appear,
hence a ``heterotic line bundle'' \cite{AGLP12}\cite{BBL17}, and
\\
\noindent {\bf (ii)} enforces on this gauge field
Ho{\v r}ava-Witten's
Green-Schwarz mechanism \eqref{HeteroticMTheoryShiftedFluxQuantization}
in heterotic M-theory:

\vspace{-.4cm}
\begin{equation}
  \label{HeteroticFluxQuantizationFromTwistorialCohomotopy}
  \xymatrix@C=34pt{
    \overset{
      \raisebox{3pt}{
        \tiny
        \color{orangeii}
        \bf
        \begin{tabular}{c}
          Twistorial
          \\
          Cohomotopy
        \end{tabular}
      }
    }{
    \mathcal{T}^{\tau}
      \underset{
        \mathclap{
        \raisebox{-3pt}{
          \tiny
          \color{darkblue}
          \bf
          \begin{tabular}{c}
            manifold
            \\
            with
            \\tangential
            $\mathrm{Sp}(2)$-structure $\tau$
          \end{tabular}
        }
        }
      }{
       \big(
         X
       \big)
      }
    }
    \ar[rr]^-{
      \mbox{
        \tiny
        \color{greenii}
        \bf
        \begin{tabular}{c}
          Twisted Non-abelian
          \\
          character map
        \end{tabular}
      }
    }_-{\mathrm{ch}}
    &&
    \left\{
      \!\!\!
      {\begin{array}{c}
        F_2,
        \\
        G_4,
        \\
        G_7
      \end{array}}
      \!\!\!
      \in \Omega^\bullet(X)
    \,\left\vert\,
      {\begin{aligned}
        d\, \phantom{2}
        \overset{
          \mathclap{
          \;\;\;\;\;\;
          \raisebox{6pt}{
            \tiny
            \color{orangeii}
            \bf
            \begin{tabular}{c}
              1st Chern form of
              \\
              heterotic line bundle
            \end{tabular}
          }
          }
        }{
          F_2
        }
        &
        =
        0
        \,,
        \;\;\;\;
        \overset{
          \raisebox{3pt}{
            \tiny
            \color{orangeii}
            \bf
            \begin{tabular}{c}
              2nd Chern class of corresponding
              \\
              $S\big( \mathrm{U}(1)^{2} \big)$-gauge field
              $\widetilde A$
            \end{tabular}
          }
        }{
          -[F_2 \wedge F_2]
          \;\in\;
          H^4\big(X;\, \mathbb{Z}\big)
        }
        \\
        d\, \phantom{2}
        \overset{
          \mathclap{
          \raisebox{3pt}{
            \tiny
            \color{darkblue}
            \bf
            \begin{tabular}{c}
              C-field
              4-flux
            \end{tabular}
          }
          }
        }{
          G_4
        }
        & = 0
        \,,
        \;\;
        \overset{
          \raisebox{1.5pt}{
            \tiny
            \color{orangeii}
            \bf
            Ho{\v r}ava-Witten's Green-Schwarz mechanism \eqref{HeteroticMTheoryShiftedFluxQuantization}
          }
        }{
          [G_4] - [\tfrac{1}{4}p_1(\omega)]
          \;=\;
          [F_2 \wedge F_2]
          \;\in\;
          H^4
          \big(
            X;\, \mathbb{Z}
          \big)
        }
        \\
        d\,
        2G_7
        & =
        -
        \big(
          G_4 - \tfrac{1}{4}p_1(\omega)
        \big)
        \wedge
        \big(
          G_4 + \tfrac{1}{4}p_1(\omega)
        \big)
        \\
        \mathclap{
        \!\!\!\!\!\!\!\!\!\!
        \raisebox{3pt}{
          \tiny
          \color{darkblue}
          \bf
          \begin{tabular}{c}
            dual
            \\
            7-flux
          \end{tabular}
        }
        }
        & \phantom{=}\;
        -
        \tfrac{1}{2}
        \big(
          p_2(\omega)
          -
          \tfrac{1}{4}
          \big(p_1(\omega)\big)^2
        \big)
      \end{aligned}}
    \right.
    \right\}
  }
\end{equation}

\newpage

\noindent
Here $X$ \eqref{HeteroticFluxQuantizationFromTwistorialCohomotopy}
denotes a spacetime manifold with
$\mathrm{Sp}(2) \hookrightarrow \mathrm{Spin}(8)$-structure
\eqref{The8Manifold},
as befits backgrounds expected in M-theory compactified on 8-manifolds
(see \cite[Rem. 3.1]{FSS19b} for pointers).
This reduction is a requirement/prediction of {\it Hypothesis H}
by Prop. \ref{EquivarianceOfCombinedHopfTwistorFibration} below.
Besides the GS-anomaly cancellation presented here
and in \cite{SS20c},
this turns out to imply several M5-brane consistency conditions
\cite{FSS19c}\cite{FSS20a}\cite{SS20a}.

\medskip

\noindent
The crux of the proof of \eqref{HeteroticFluxQuantizationFromTwistorialCohomotopy}
is this cohomological analysis of the $\mathrm{Sp}(2)$-equivariantized Hopf/twistor fibration
(\cref{BorelEquivariantHopfTwistorFibration}):

\vspace{-1cm}

  \hypertarget{FigureC}{}
  \begin{equation}
    \label{TheComponents}
    \raisebox{226pt}{
    \xymatrix@R=10pt@C=0pt{
      S^7 \!\sslash\! \mathrm{Sp}(2)
      \ar[dddd]^-{
        \mathllap{
          \mbox{
            \tiny \bf
             \begin{tabular}{c}
             \color{greenii}          Borel-equivariant
              \\
              \color{greenii}         complex
              \\
              \color{greenii}          Hopf fibration
              \\
               $h_{\mathbb{C}} \sslash \mathrm{Sp}(2)$
            \end{tabular}
          }
        }
      }
      \ar@/^2.8pc/[dddddddd]|<{
 \;\;\;\;\;\;\;
                \overset{
          \mathrlap{
           \mbox{
            \tiny \bf
                        \begin{tabular}{c}
                        \phantom{A}
                        \\
                         \phantom{A}
                        \\
                         \phantom{A}
                        \\
                         \phantom{A}
                        \\
                          \phantom{A}
                        \\
                         \phantom{A}
                        \\
                         \phantom{A}
                        \\
       \color{greenii}       Borel-equivariant
              \\
       \color{greenii}       quaternionic
              \\
     \color{greenii}         Hopf fibration
              \\
              $h_{\mathbb{H}}\sslash \mathrm{Sp}(2)$
            \end{tabular}
          }
          }
        }{
        }
      }
      &
      \ar@{<->}[rr]_{\simeq}^-{
        \mathclap{
        \mbox{
          \tiny \bf
          \color{greenii}
          \begin{tabular}{c}
            coset space
            \\
            realization
          \end{tabular}
        }
        }
      }
      &&
      &
      B \mathrm{Sp}(1)
      \ar[dddd]
      &
      \ar@{<->}[rr]_-{\simeq}^-{
        \mathclap{
        \mbox{
          \tiny \bf
          \color{greenii}
          \begin{tabular}{c}
            exceptional
            \\
            isomorphism
          \end{tabular}
        }
        }
      }
      &&
      &
      B \mathrm{SU}(2)_L
      \ar@{}[r]|-{\mbox{$\times$}}
      \ar@{=}[dddd]
      &
      \;\;\;\;\;\ast\;\;\;\;\;
      \ar[dddd]^-{ B e }
      \\
      \\
      \\
      \\
      \mathbb{C}P^3 \!\sslash\! \mathrm{Sp}(2)
      \ar[dddd]^-{
        \mathllap{
          \mbox{
            \tiny \bf
                     \begin{tabular}{c}
            \color{greenii}     Borel-equivariant
              \\
       \color{greenii}          twistor fibration
              \\
              $   t_{\mathbb{H}}
        \sslash
        \mathrm{Sp}(2)$
            \end{tabular}
          }
        }
      }
      &
      \ar@{<->}[rr]^{\simeq}_-{
        \mathclap{
        \mbox{
          \tiny \bf
          \color{greenii}
          \begin{tabular}{c}
            coset space
            \\
            realization
          \end{tabular}
        }
        }
      }
      &&
      &
      B
      \big(
        \mathrm{Sp}(1) \times \mathrm{U}(1)
      \big)
      \ar[dddd]
      &
      \ar@{<->}[rr]_-{\simeq}^-{
        \mathclap{
        \mbox{
          \tiny \bf
          \color{greenii}
          \begin{tabular}{c}
            exceptional
            \\
            isomorphism
          \end{tabular}
        }
        }
      }
      &&
      &
      B \mathrm{SU}(2)_L
      \ar@{}[r]|-{\mbox{$\times$}}
      \ar@{=}[dddd]
      &
      B \mathrm{U}(1)_R
      \ar@/^2.5pc/[dddd]^-{\!\!
        B( c\, \mapsto \mathrm{diag}(c,\overline{c}))
      }
      \ar[dd]^-{
        \rotatebox[origin=c]{-90}{
          \scalebox{.7}{$
            \simeq
          $}
        }
      }
      \\
      \\
      &&&&&&&&&
       \color{orangeii}{B
      \mathrm{S}
      \big(
        \mathrm{U}(1)^2
      \big)}
      \ar@{^{(}->}[dd]
      \\
      \\
      S^4 \!\sslash\! \mathrm{Spin}(5)
      &
      \ar@{<->}[rr]_-{\simeq}^-{
        \mathclap{
        \mbox{
          \tiny \bf
          \color{greenii}
          \begin{tabular}{c}
            coset space
            \\
            realization
          \end{tabular}
        }
        }
      }
      &&
      &
      B \mathrm{Spin}(4)
      &
      \ar@{<->}[rr]_-{\simeq}^-{
        \mathclap{
        \mbox{
          \tiny \bf
          \color{greenii}
          \begin{tabular}{c}
            exceptional
            \\
            isomorphism
          \end{tabular}
        }
        }
      }
      &&
      &
      B \mathrm{SU}(2)_L
      \ar@{}[r]|-{\mbox{$\times$}}
      &
      B \mathrm{SU}(2)_R
      \\
      {\phantom{A}}
      \\
      H^4\big(
        \mathbb{C}P^3 \!\sslash\! \mathrm{Sp}(2); \mathbb{Z}
      \big)
      \ar@{<-}[dddddddd]^-{
        \mathllap{
          \mbox{
            \tiny \bf
                       \begin{tabular}{c}
            \color{greenii}   pullback in
            \\
       \color{greenii}     cohomology along
              \\
          \color{greenii}     Borel-equivariant
              \\
          \color{greenii}     twistor fibration
              \\
              \\
              $\big(
          t_{\mathbb{H}}
          \sslash
          \mathrm{Sp}(2)
        \big)^\ast$
            \end{tabular}
          }
        }
      }
      &
      \ar@{<->}[rr]^{\simeq}
      &&
      &
      H^4\big(
        B \mathrm{Sp}(1) \times \mathrm{U}(1); \mathbb{Z}
      \big)
      &
      \ar@{<->}[rr]^{\simeq}
      &&
      &
      H^4(B \mathrm{SU}(2); \mathbb{Z})
      \ar@{}[r]|-{\mbox{$\oplus$}}
      &\
      H^4(B \mathrm{U}(1); \mathbb{Z})
      \\
      &
      &
      \overset{
        \mathclap{
        \mbox{
          \tiny \bf
          \color{orangeii}
          \begin{tabular}{c}
            universal
            gauge instanton density
            \\
            in heterotic M-theory
          \end{tabular}
        }
        }
      }{
        \mathllap{-}
        a
      }
      \ar@{<-|}[dddd]|-{
        \mbox{
          \tiny \bf
          \color{orangeii}
          \begin{tabular}{c}
            Ho{\v r}ava-Witten's
            Green-Schwarz mechanism
            \\
            in heterotic M-theory
          \end{tabular}
        }
      }
      &
      =
      &
      c_1
      \cup
      c_1
      &&
      =
      &&
      &
      c^{\scalebox{.55}{$R$}}_1
      \cup
      c^{\scalebox{.55}{$R$}}_1
      \mathrlap{
        \mbox{
          \tiny \bf
          \color{orangeii}
          \begin{tabular}{c}
            2nd Chern class of
            \\
            $\mathrm{S}\big( \mathrm{U}(1)^2 \big)$-field
          \end{tabular}
        }
      }
      \ar@{<-|}[dddd]^-{\!\!
        \big(\!
          B \left(c \mapsto \mathrm{diag}(c,\bar c) \right)
        \! \big)^{\!\!\ast}
      }
      &&
      \\
      \\
      \\
      \\
      &
      &
      \overset{
        \mathclap{
        \mbox{
          \tiny \bf
          \color{darkblue}
          \begin{tabular}{c}
            shifted integral C-field flux
            \\
            relative to background flux
          \end{tabular}
        }
        }
      }{
        \widetilde \Gamma_4
        -
        \widetilde \Gamma_4^{\mathrm{vac}}
      }
      &
      =
      &
      \tfrac{1}{2}\rchi_4 - \tfrac{1}{4}p_1
      &&\mathclap{=}&&
      &
      \mathllap{-}
      c^{\scalebox{.55}{$R$}}_2
      \mathrlap{
        \mbox{
          \tiny \bf
          \color{darkblue}
          \begin{tabular}{c}
            right
            \\
            2nd Chern class
          \end{tabular}
        }
      }
      &&
      \\
      &
      &
      \overset{
        \mathclap{
        \mbox{
          \tiny \bf
          \color{darkblue}
          \begin{tabular}{c}
            universal
            \\
            shifted integral
            C-field flux
          \end{tabular}
        }
        }
      }{
        \widetilde \Gamma_4
      }
      &
      =
      &
      \tfrac{1}{2}\rchi_4 + \tfrac{1}{4}p_1
      &&\mathclap{=}&&
      c^{\scalebox{.55}{$L$}}_2
      \mathrlap{
        \mbox{
          \tiny \bf
          \color{darkblue}
          \begin{tabular}{c}
            left 2nd Chern class
          \end{tabular}
        }
      }
      &
      &
      \mathrlap{
        \!\!\!\!\!\!\!\!\!\!\!\!\!\!\!\!\!\!\!\!\!
        \scalebox{.8}{
            \cite[3.9]{FSS19b}
        }
      }
      \\
      &
      &
      \overset{
        \mbox   {
          \tiny \bf
          \color{darkblue}
          C-field
          background flux
        }
      }{
        \widetilde \Gamma_4^{\mathrm{vac}}
      }
      &
      =
      &
      \tfrac{1}{2}p_1
      &&\mathclap{=}&&
      c^{\scalebox{.55}{$L$}}_2
      \ar@{}[r]|-{\mbox{$+$}}
      &
      c^{\scalebox{.55}{$R$}}_2
      \mathrlap{
        \;\;\;
        \mbox{
          \tiny  \bf
          \color{darkblue}
          \begin{tabular}{c}
          total
          \\
          2nd Chern class
       \end{tabular}
        }
      }
      \\
      H^4\big(
        S^4 \!\sslash\! \mathrm{Spin}(5); \mathbb{Z}
      \big)
      &
      \ar@{<->}[rr]^{\simeq}
      &&
      &
      H^4\big(
        B \mathrm{Spin}(4); \mathbb{Z}
      \big)
      &
      \ar@{<->}[rr]^{\simeq}
      &&
      &
      H^4(B \mathrm{Sp}(1); \mathbb{Z})\,
      \ar@{}[r]|-{\mbox{$\oplus$}}
      &
      \,H^4(B \mathrm{Sp}(1); \mathbb{Z})
    }
    }
  \end{equation}

  \vspace{.2cm}
  {\footnotesize
\noindent {\bf Figure I.}
{\bf Integral cohomology of Borel-equivariant Hopf/twistor fibration and its
interpretation under {Hypothesis H}.}

\vspace{.2cm}
\hspace{-.8cm}
\begin{minipage}[left]{18cm}
\begin{itemize}
\vspace{-.12cm}
\item[\bf (i)] The top part shows the equivalent incarnations of the
Borel-equivariant Hopf/twistor fibration (Def. \ref{ParametrizedHopfTwistorFibration})
according to Prop. \ref{BorelEquivariantTwistorFibrationAsMapsOfClassifyingSpaces}.

\vspace{-.12cm}
\item[\bf (ii)] The bottom part shows the
corresponding identifications of the integral cohomology
generators \eqref{IntegralCohomologyOfS4ModSp2} according to
\cite[3.9]{FSS19b}.

\vspace{-.12cm}
 \item[\bf (iii)] This makes manifest, shown in the middle of the diagram,
 how these generators pull back along the
 fibration (Theorem \ref{CohomologyOfBorelEquivariantTwistorFibration}),
 which

\vspace{-.12cm}
 \item[\bf (iv)] allows to normalize the generators in the Sullivan model
 for the rational homotopy type of the fibrations, below in
 Theorem \ref{SullivanModelOfParametrizedHopfTwistorFibrations}.

\vspace{-.12cm}
 \item[\bf (v)] Blue labels indicate the interpretation
 of the bottom generators
 as universal fluxes in M-theory,
 according to \cite{FSS19b}\cite{FSS20a};

\vspace{-.12cm}
 \item[\bf (vi)]
 while orange label indicate the
 interpretation of the new top generators as universal fluxes
 in heterotic M-theory, discussed in \cref{GreenSchwarzMechanismFromJTwistedCohomotopy}.
 \end{itemize}
 \end{minipage}
 }

\newpage

\noindent {\bf Conclusion and Outlook.}

\medskip

\noindent {\bf The heterotic gauge field.}
It is remarkable that the class of a
gauge field $\widetilde A$,
which had remained mysterious in \eqref{HeteroticMTheoryShiftedFluxQuantization},
does appear from charge-quantization in Twistorial Cohomotopy,
according to \eqref{HeteroticFluxQuantizationFromTwistorialCohomotopy}.

\vspace{2mm}
\noindent {\bf (i) Missing generality?}
Of course, the gauge field in \eqref{HeteroticFluxQuantizationFromTwistorialCohomotopy}
has the abelian structure group $G = S\big( \mathrm{U}(1)^2\big)$ instead of
the non-abelian structure group $G = E_8$ that could be expected
to apply to \eqref{HeteroticMTheoryShiftedFluxQuantization}.
In terms of characteristic classes, this means that
charge quantization in Twistorial Cohomotopy constrains the
class $a = [\mathrm{c}_2(\widetilde A)] \;\in\; H^4\big( X^{11}; \mathbb{Z} \big)$,
which for $G = E_8$ may be \emph{any} element in degree four integral cohomology
(since $\tau_{{}_{11}} B E_8  \simeq_{\mathrm{whe}} \tau_{{}_{11}} K(\mathbb{Z}; 4)$,
e.g. \cite[3.2]{DFM03}),
to factor as minus a cup square of an element in degree two integral cohomology.
This might indicate that Twistorial Cohomotopy as presented does not capture full
heterotic M-theory; or that one should look for other
factorizations of the quaternionic Hopf fibration, or for variants
of the construction presented here.

\vspace{2mm}
\noindent {\bf {\bf (ii)} Or predictive constraint?}
On the other hand, it is curious to notice that
in heterotic string \emph{phenomenology} the reduction of the
heterotic gauge bundle along the inclusions

\vspace{-.4cm}
\begin{equation}
  \label{ReductionToHeteroticLineBundles}
  \big(
    \mathrm{U}(1)
  \big)^{n-1}
  \,\simeq \;
  S
  \big(
    \mathrm{U}(1)^n
  \big)
  \;\subset\;
  \mathrm{SU}(n)
  \;\subset\;
  E_8
  \,,
  \phantom{AAA}
  \mbox{for $2 \leq n \leq 5$}
  \,,
\end{equation}
\vspace{-.5cm}

\noindent
has led to a little revolution in string phenomenology \cite{AGLP11}\cite{AGLP12}.
These \emph{heterotic line bundle models} turn out
to be an abundant source of low energy theories with the \emph{exact}
field content
of the (minimally supersymmetric) standard model of particle physics
(up to decoupled and ultra-heavy fields), amenable to
effective computational classification
\cite{ACGLP13}\cite{HLLS13}\cite{BBL17}\cite{GrayWang19}
(used for $n = 4, 5$ in the observable sector,
while  our $n = 2$ is used in the hidden sector
\cite[\S 4.2]{ADO20a}\cite[\S 2.2]{ADO20b}\cite{DumitruOvrut21}\cite{DumitruOvrut22}).
Before considering the
reduction \eqref{ReductionToHeteroticLineBundles}, only a small
handful of hand-crafted semi-realistic models were known.

Notice that this works because the structure group of the
heterotic gauge bundle is part of what \emph{breaks} $E_8$ down to the
low-energy gauge group: the latter is
within the commutant of the former in
$E_8$. Therefore, realistic phenomenology does not
require $\widetilde A$ in \eqref{HeteroticMTheoryShiftedFluxQuantization}
to be in a non-abelian GUT-group -- in fact it must instead be complementary to
(be in the commutant within $E_8$ of) the low-energy gauge group
($\widetilde A$ is a background field/vev, not the dynamical gauge field
fluctuating about it);
and analysis of heterotic line bundle models
indicates that restricting $\widetilde A$ to be reduced along \eqref{ReductionToHeteroticLineBundles}
narrows in heterotic M-theory onto its phenomenologically realistic sector.

This might indicate that {\it Hypothesis H} captures not only
mathematical but also phenomenological constraints of M-theory.

\medskip
\noindent {\bf The degree-8 polynomial.}
Beyond encoding the shifted heterotic flux quantization
in the first two lines of \eqref{HeteroticFluxQuantizationFromTwistorialCohomotopy},
the third line there shows that charge-quantization in
Twistorial Cohomotopy also enforces (Corollary \ref{CohomologicalRelationInTwistorialCohomotopy})
the trivialization of this 8-class:
\begin{equation}
  \label{The8Class}
  \begin{aligned}
  I_8^H
  :=
  &
  \big(
    [G_4] - \tfrac{1}{4}p_1
  \big)
  \cup
  \big(
    [G_4] + \tfrac{1}{4}p_1
  \big)
  +
  \tfrac{1}{2}
  \big(
    p_2  - \tfrac{1}{4} p_1 \cup p_1
  \big)
  \\
  =
  \;
  &
  \big(
    [F_2 \wedge F_2]
    +
    \tfrac{1}{2}p_1
  \big)
  \cup
  [F_2 \wedge F_2]
  \;+\;
  \tfrac{1}{2}
  \big(
    p_2  - \tfrac{1}{4} p_1 \cup p_1
  \big)
  \;\;
  &
  \in \;\;
  H^8
  (
    X, \mathbb{R}
  )\;.
  \end{aligned}
\end{equation}
In the form of the first line of \eqref{The8Class},
the condition $I_8^H = 0$ is inherited from charge-quantization
in J-twisted Cohomotopy \eqref{ShiftedFluxQuantizationFromChernCharacterInJTwistedCohomotopy}.
We had shown previously that this condition guarantees
the vanishing \emph{under general conditions} \footnote{
  Both cancellations had previously been discussed only
  subject to tacit assumptions on the C-field;
  see \cite[p. 2]{FSS19b} and \cite[(6)]{SS20a}.
}
of
\\
\noindent {\bf (a)} the anomaly in the Hopf-WZ term on the M5-brane \cite{FSS19c},
\\
{\bf (b)} the total remaining anomaly of the M5-brane \cite{SS20a}.

In the form of the second line in \eqref{The8Class} --
now equivalently re-expressed in terms of the emergent heterotic gauge flux
instead of the $G_4$-flux -- this class is seen to be closely related to
the 8-class denoted $\widehat {I_8}$ in \cite[(1.10)]{HoravaWitten96b}:
Up to global and relative rescaling of $\widehat {I_8}$
(as on the bottom of \cite[p. 15]{HoravaWitten96b}) both are related by
\begin{equation}
  \label{ShiftedI8}
  I_8^H
  \;=\;
  \widehat{I_8}
  -
  \tfrac{1}{4}p_1 \cup p_1\;.
\end{equation}
It might be interesting to understand the potential significance of this relation.
Notice that

\noindent {\bf (i)} in \cite{HoravaWitten96b} there is no condition that
$\widehat{I_8}$ should vanish;

\noindent {\bf (ii)} the shift in \eqref{ShiftedI8} is what makes $I_8^H$
an integral class (using \eqref{chi8RelationOnBSp2} and \eqref{ShiftedQuantizationConditionOnG4InJTwistedTTheory});

\noindent{\bf (iii)} it is expected \cite{Moss08} that $\widehat{I_8}$
is just approximate: it receives an infinite
but unknown tower of corrections;

\noindent {\bf (iv)} while {\it Hypothesis H} suggests that
$I_8^H = 0$ is a statement about fully-fledged M-theory.

\newpage

\section{Borel-equivariant Hopf/twistor fibration}
\label{BorelEquivariantHopfTwistorFibration}

The \emph{twistor fibration}
\vspace{-.4cm}
$$
\xymatrix@C=15pt@R=9pt{ S^2 \ar[r] & \mathbb{C}P^3 \ar[d] \\ & S^4}
$$
(due to \cite[\S III.1]{Atiyah79},
see also, e.g., \cite[\S 1]{Bryant82}\cite{ArmstrongSalamon13}\cite[\S 6]{ABS19})
or \emph{Penrose fibration} (as in \cite{EellsSalamon85}),
is the
canonical map $\mathbb{C}P^3 \to \mathbb{H}P^1$
(under the identification $\mathbb{H}P^1 \simeq S^4$, recalled below as
Prop. \ref{EquivariantIdentificationOfS4WithHP1})
that sends complex lines to the quaternionic lines which they span
\cite[\S III (1.1)]{Atiyah79}.
While, as the name suggests, this is traditionally motivated
from the role of $\mathbb{C}P^3$ as a \emph{twistor space}
(a division-algebraic account is in \cite{BengtssonCederwall88}),
our interest in the twistor fibration here comes from
its appearance as an intermediate stage of the
\emph{quaternionic Hopf fibration} \cite[\S 6]{GWY83}.
We observe that it
is given by the following iterative quotienting by multiplicative groups
in the four real normed division algebras
(reals $\mathbb{R}$, complex numbers $\mathbb{C}$, quaternions $\mathbb{H}$,
octonions $\mathbb{O}$):
\vspace{-3mm}
\begin{equation}
  \label{HopfTwistorFactorization}
  \raisebox{80pt}{
  \xymatrix@R=5pt@C=20pt{
    S^1
    \ar@{}[r]|<<<<<<<{\simeq}
    \ar[drr]
    &
    \;\;\;\;\;\mathbb{C}^\times/\mathbb{R}^\times_+
    \ar[drr]
    &
    &
    &&
    \hspace{1cm}
    &
    {\begin{array}{c}
      \mbox{for}
      \\
      v \mathrlap{\,\neq\, 0}
    \end{array}}
    \ar@{}[d]|-{
      \rotatebox[origin=c]{-90}{
        $\in$
      }
    }
    \\
    &&
    S^7
    \ar[dd]|-{
      h_{\mathbb{C}}
      \phantom{\mathclap{\vert}}
    }
    \ar@{}[r]|-{\simeq}
    &
    \big(
      \mathbb{R}^8 \!\setminus\! \{0\}
     \big)\!/ \mathbb{R}^\times_+
    \ar[dd]
    \ar@{}[r]|-{\ni}
    &
    \big\{
      v \cdot t
      \left\vert
        \,
        t \in \mathbb{R}^\times_+
      \right.
      \!
    \big\}
    \ar@{|->}[dd]
    &&
    \mathbb{R}^8
    \ar@{}[dd]|-{
      \rotatebox[origin=c]{-90}{
        $\simeq_{{}_{\mathbb{R}}}$
      }
    }
    \\
    S^2
    \ar@{}[r]|<<<<<<<<{\simeq}
    \ar[drr]
    &
    \;\;\;\;\;\mathbb{H}^\times/\mathbb{C}^\times
    \ar[drr]|-{\phantom{AA}}
    &
    \ar@{}[r]|-{
      \mbox{
        \tiny
        \color{darkblue}
        \bf
        \begin{tabular}{c}
          7d complex
          \\
          Hopf fibration
        \end{tabular}
      }
    }
    &
    \\
    &&
    \mathbb{C}P^3
    \ar@{}[r]|-{\simeq}
    \ar[dd]|-{
      t_{\mathbb{H}}
      \phantom{\mathclap{\vert}}
    }
    &
    \big(
      \mathbb{C}^4 \!\setminus\! \{0\}
    \big)\!/ \mathbb{C}^\times
    \ar[dd]
    \ar@{}[r]|-{\ni}
    &
    \big\{
      v \cdot z
      \left\vert
        \,
        z \in \mathbb{C}^\times
      \right.
      \!\!
    \big\}
    \ar@{|->}[dd]
    \ar@{}[r]|<<{
      \left.
        \begin{array}{c}
          \\
          \\
          \\
          \\
          \\
          \\
          \\
          \\
        \end{array}
      \right\}
      \rotatebox[origin=c]{-90}{
        \tiny
        {\color{darkblue}
        \bf
        quaternionic
        Hopf fibration
        }
        $\mathrlap{h_{\mathbb{H}}}$
      }
    }
    &
    &
    \mathbb{C}^4
    \ar@{}[dd]|-{
      \rotatebox[origin=c]{-90}{
        $\simeq_{{}_{\mathbb{R}}}$
      }
    }
    \\
    S^4
    \ar@{}[r]|<<<<<<<<{\simeq}
    \ar[drr]|-{
      \rotatebox[origin=c]{-15}{
        \scalebox{.6}{
          $\simeq$
        }
      }
    }
    &
    \;\;\;\;\;\mathbb{O}^\times/\mathbb{H}^\times
    \ar[drr]|-{\phantom{AA}}
    &
    \ar@{}[r]|-{
      \mbox{
        \tiny
        \color{darkblue}
        \bf
        \begin{tabular}{c}
          twistor fibration
        \end{tabular}
      }
    }
    &
    \\
    &&
    \mathbb{H}P^1
    \ar@{}[r]|-{\simeq}
    \ar[dd]
    &
    \big(
      \mathbb{H}^2 \!\setminus\! \{0\}
    \big)\!/ \mathbb{\mathbb{H}}^\times
    \ar[dd]
    \ar@{}[r]|-{\ni}
    &
    \big\{
      v \cdot q
      \left\vert
        \,
        q \in \mathbb{H}^\times
      \right.
      \!\!
    \big\}
    \ar@{|->}[dd]
    &&
    \mathbb{H}^2
    \ar@{}[dd]|-{
      \rotatebox[origin=c]{-90}{
        $\simeq_{{}_{\mathbb{R}}}$
      }
    }
    \\
    {\phantom{S^4}}
    &
    {\phantom{\;\;\;\;\;\mathbb{O}^\times/\mathbb{H}^\times}}
    \\
    &&
    \ast
    \ar@{}[r]|-{\simeq}
    &
    \big(
      \mathbb{\mathbb{O}}^1 \!\setminus\! \{0\}
    \big)\!/ \mathbb{\mathbb{O}}^\times
    \ar@{}[r]|-{\ni}
    &
    \big\{
      v \cdot o
      \left\vert
        \,
        o \in \mathbb{O}^\times
      \right.
      \!\!
    \big\}
    &&
    \mathbb{O}
  }
  }
\end{equation}

\noindent Alternatively, in its coset-space realization
\vspace{-.2cm}
$$
\hspace{1cm}
\xymatrix@R=8pt@C=12pt{
  \mathrm{SU}(4)/\mathrm{U}(2)
  \ar[r]
  &
  \mathrm{SO}(5)/\mathrm{U}(2) \ar[d] \\ & \mathrm{SO}(5)/\mathrm{SO}(4) }
$$
the twistor fibration is also called \emph{Calabi-Penrose fibration}
(following \cite[\S 3]{Lawson85}, see also \cite{Loo89} and see
\cite[2.31]{Nordstrom08} for a review of Calabi's construction
\cite{Calabi67}\cite{Calabi68}).
We observe that the $\mathrm{Sp}(2)$-coset realization
(\cite[Table 1]{Onishchik60}, see \cite[Table 3]{GorbatsevichOnishchik93})
of the Hopf/twistor fibrations is given as follows
(see also \cite[(73)]{FSS19b}):
\begin{equation}
  \label{HopfTwistorFibrationCosetRealization}
  \raisebox{40pt}{
  \xymatrix{
    S^7
    \ar@{}[r]|-{ \simeq }
    \ar[d]|-{
      h_{\mathbb{C}}
      \mathrlap{
        \;\;
        \mbox{
        \tiny
        \color{darkblue}
        \bf
        \begin{tabular}{c}
          7d complex
          \\
          Hopf fibration
        \end{tabular}
        }
      }
    }
    \ar@/_1.3pc/[dd]_-{
      \mathllap{
        \mbox{
          \tiny
          \color{darkblue}
          \bf
          \begin{tabular}{c}
            quaternionic
            \\
            Hopf fibration
          \end{tabular}
        }
        \;
      }
      h_{\mathbb{H}}
    }
    &
    \mathrm{Sp}(2)/{\phantom{(}}\mathrm{Sp}(1)_L {\phantom{\times \mathrm{Sp}(1)_R}}
    \\
    \mathbb{C}P^3
    \ar@{}[r]|-{ \simeq }
    \ar[d]|-{
      t_{\mathbb{H}}
      \mathrlap{
        \;\;
        \mbox{
        \tiny
        \color{darkblue}
        \bf
        \begin{tabular}{c}
          Calabi-Penrose
          \\
          twistor fibration
        \end{tabular}
        }
      }
    }
    &
    \mathrm{Sp}(2)/\big( \mathrm{Sp}(1)_L \times \mathrm{U}(1)_R \big)
    \\
    S^4
    \ar@{}[r]|-{\simeq}
    &
    \mathrm{Sp}(2)/
    \big(
      \mathrm{Sp}(1)_L \times \mathrm{Sp}(1)_R
    \big)
    }
  }
\end{equation}
where the maps are induced by the canonical subgroup inclusions
(recalled in Example \ref{SubgroupsOfQuaternionLinearGroups}).

\medskip

 We discuss in this paper the enhancement
(Prop. \ref{EquivarianceOfCombinedHopfTwistorFibration} below)
of these classical fibrations \eqref{HopfTwistorFactorization}
\eqref{HopfTwistorFibrationCosetRealization}
to Borel-equivariant
\emph{parametrized fibrations}
(Def. \ref{ParametrizedHopfTwistorFibration} below)
over the classifying space
of the group $\mathrm{Sp}(2)$ (recalled as Def. \ref{QuaternionicGroups} below),
generalizing the analogous discussion for just the quaternionic Hopf
fibration in \cite{FSS19b}\cite{FSS19c}.
The main results are
Theorem \ref{CohomologyOfBorelEquivariantTwistorFibration}
and
Theorem \ref{SullivanModelOfParametrizedHopfTwistorFibrations} below,
which characterize the integral cohomology and the rational homotopy type
of the Borel-equivariant Hopf/twistor fibrations
\eqref{ParametrizedHopfAndTwistorFibration}
(generalizing the result for just the quaternionic
Hopf fibration from \cite[3.19]{FSS19b}).

\medskip

\subsection{Construction}
 \label{ConstructionOfBorelEquivariantHopfTwistorFibration}

Here we determine the maximal symmetry group of the
joint Hopf/twistor fibrations (Prop. \ref{EquivarianceOfCombinedHopfTwistorFibration},
Remark \ref{TwistorSpacesBreaksEquivarianceToSp2}),
construct the corresponding Borel-equivariantization
(Def. \ref{ParametrizedHopfTwistorFibration}) and characterize its
integral cohomology (Theorem \ref{CohomologyOfBorelEquivariantTwistorFibration}).

\medskip

The following isomorphisms (Prop. \ref{EquivariantIdentificationOfS4WithHP1})
are classical,
but since the third of these is rarely made explicit in the literature,
we spell out a proof:
\begin{prop}[Equivariant identification of 4-sphere with quaternionic projective space]
  \label{EquivariantIdentificationOfS4WithHP1}
  There are canonical isomorphisms

  \noindent {\bf (i)} of topological spaces:

  \vspace{-.9cm}

  \begin{equation}
    \xymatrix{
      S^4 \ar[r]^-{ \alpha }_-{\simeq}  & \mathbb{H}P^1
    };
  \end{equation}

  \noindent {\bf (ii)} of topological groups:

  \vspace{-.9cm}

  \begin{equation}
    \xymatrix{
      \mathrm{Spin}(5) \ar[r]^{\gamma}_-{\simeq} & \mathrm{Sp}(2)
    };
  \end{equation}

  \noindent {\bf (iii)} of canonical topological group actions:

  \vspace{-4mm}

  \begin{equation}
    \xymatrix{
      \mathrm{Spin}(5)
      \acts \;
      S^4
      \ar[rr]^{ (\gamma,\alpha) }_-{\simeq}
      &&
      \mathrm{Sp}(2)
      \acts \;
      \mathbb{H}P^1
    }.
  \end{equation}
\end{prop}
\begin{proof}
Quaternionic 2-component spinor formalism
provides an isomorphism of real quadratic vector spaces
(\cite{KugoTownsend82},
streamlined review in \cite{BH09}\cite{VenanciaBatista20}\cite[\S 3.2]{FSS20b})

\vspace{-2mm}
\begin{equation}
  \label{6dMinkowskiSpacetimeAsQuaternionicMatrices}
  \xymatrix@R=-4pt{
    \big(
      \mathbb{R}^{ 5,1 },
      \eta
    \big)
    \ar[rr]^-{ \simeq }
    &&
    \Big(
      \mathrm{Mat}^{\mathrm{herm}}
      \big(
        2 \times 2,
        \mathbb{H}
      \big)
      ,
      - \mathrm{det}
    \Big)
    \\
\scalebox{0.7}{$    \left[
    \!\!
    {\begin{array}{c}
      x^0,
      \\
      x^1,
      \\
      x^2,
      \\
      \vdots
      \\
      x^5
    \end{array}}
    \!\!
    \right]
    $}
    \ar@{}[rr]|-{\longmapsto}
    &&
 \scalebox{0.7}{$    \left(\!\!\!
    {\begin{array}{cc}
      x^0 - x^1
      &
      x^2 + \mathrm{i} x^3 + \mathrm{j} x^4 + \mathrm{k} x^5
      \\
      x^2 - \mathrm{i} x^3 - \mathrm{j} x^4 - \mathrm{k} x^5
      &
      x^0 + x^1
    \end{array}}
    \!\!\! \right)
    $}
  }
\end{equation}

\vspace{-.2cm}

\noindent from
{\bf (a)} 6d Minkowski spacetime $\mathbb{R}^{5,1}$ with metric
$\eta := \mathrm{diag}(-1,+1, \cdots, +1)$
to {\bf (b)}
the vector space of 2-by-2 quaternionic matrices which are hermitian,
$A^\dagger = A$, with quadratic form the negative of the determinant
operation.
Under this identification, the canonical action of
$\mathrm{Spin}(5,1)$ on $\mathbb{R}^{5,1}$ (through that of
$\mathrm{SO}(5,1)$) translates to the conjugation action of
$\mathrm{SL}(2,\mathbb{H})$ \eqref{DefSL2H}:
\vspace{-2mm}
\begin{equation}
  \label{Spin51ByQuaternionicMatrices}
  \xymatrix@R=0pt{
    \mathrm{Spin}(5,1)
    \acts  \;
    \mathbb{R}^{5,1}
    \ar[rr]^-{ \simeq }
    &&
    \mathrm{SL}(2,\mathbb{H})
    \acts \;
    \mathrm{Mat}^{\mathrm{herm}}
    \big(
      2 \times 2,
      \mathbb{H}
    \big)
    \\
    &&
    A \mapsto G \cdot A \cdot G^\dagger
  }
\end{equation}

\vspace{-2mm}
\noindent Now consider the restriction of this situation to
the Euclidean spatial slice
$\mathbb{R}^5 \hookrightarrow \mathbb{R}^{5,1}$
determined by $x^0 = 0$.
Under the isomorphism \eqref{6dMinkowskiSpacetimeAsQuaternionicMatrices},
this clearly corresponds to restriction to the \emph{traceless} hermitian
matrices:

\vspace{-.3cm}

\begin{equation}
  \label{5dEuclideanSpaceAsQuaternionicMatrices}
  \xymatrix@R=5pt{
    \big(
      \mathbb{R}^{5},
      g
    \big)
    \ar[rr]^-{ \simeq }
    &&
    \Big(
      \mathrm{Mat}^{\mathrm{herm}}_{\mathrm{trless}}
      \big(
        2 \times 2,
        \mathbb{H}
      \big)
      ,
      - \mathrm{det}
    \Big)
  }
\end{equation}
\vspace{-.2cm}
\noindent Notice here, from direct inspection (see also \cite[Prop. 5]{BH09}), that
\begin{equation}
  \label{DeterminantOfHermitianTracelessMatrices}
  A
  \;\in\;
  \mathrm{Mat}^{\mathrm{herm}}_{\mathrm{trless}}
  \big(
    2 \times 2, \mathbb{H}
  \big)
  \phantom{AAAA}
  \Rightarrow
  \phantom{AAAA}
  A \cdot A
  \;=\;
  - \mathrm{det}(A) \cdot \mathrm{I}
  \,.
\end{equation}

\vspace{-2mm}
\noindent
Morever, the subgroup $\mathrm{Spin}(5) \subset \mathrm{Spin}(5,1)$
which fixes this subspace
corresponds under \eqref{Spin51ByQuaternionicMatrices}
to that subgroup of $\mathrm{SL}(2,\mathbb{H})$
whose conjugation operation preserves traceless matrices.
Since this means, equivalently, to act trivially on their orthogonal complement,
given by the pure trace matrices,
i.e. the real multiples of the 2-by-2 identity matrix $\mathrm{I} := \mathrm{Id}_{\mathbb{H}^2}$:
$$
  G \cdot \mathrm{I} \cdot G^\dagger
  \;=\;
  \mathrm{I}
  \phantom{AAA}
  \Leftrightarrow
  \phantom{AAA}
  G \cdot G^\dagger
  \;=\;
  \mathrm{I}
  \,,
$$
we see, using \eqref{Sp2InSL2H},
that this is the quaternionic unitary group
$\mathrm{Sp}(2) := \mathrm{U}(2,\mathbb{H})$ \eqref{Spn},
hence that \eqref{Spin51ByQuaternionicMatrices} restricts as follows:

\vspace{-4mm}

\begin{equation}
  \label{Spin5ByQuaternionicMatrices}
  \xymatrix@R=5pt{
    \mathrm{Spin}(5)
    \acts \;
    \mathbb{R}^{5}
    \ar[rr]^-{ \simeq }
    &&
    \;
    \mathrm{Sp}(2)
    \acts \;
    \mathrm{Mat}^{\mathrm{herm}}_{\mathrm{trless}}
    \big(
      2 \times 2,
      \mathbb{H}
    \big)
    \,.
  }
\end{equation}
Analogously, the further restriction to the unit sphere in
$\mathbb{R}^5$ corresponds, under \eqref{6dMinkowskiSpacetimeAsQuaternionicMatrices}
and in view of \eqref{DeterminantOfHermitianTracelessMatrices},
to those matrices which are all of:
\vspace{-2mm}
$$
  \underset{
    \mathclap{
    \raisebox{-3pt}{
      \tiny
      \color{darkblue}
      \bf
      6d Minkowski spacetime
    }
    }
  }{
    \mbox{{\bf (a)} hermitian:}
    \;\;
    A^\dagger = A
  }
  \,,
  \phantom{AAAA}
  \underset{
    \mathclap{
    \raisebox{-3pt}{
      \tiny
      \color{darkblue}
      \bf
      5d Euclidean spacetime
    }
    }
  }{
    \mbox{ {\bf (b)} traceless:}\;\;
    \mathrm{tr}(A) = 0
  }
  \,,
  \phantom{AAAA}
  \underset{
    \mathclap{
    \raisebox{-3pt}{
      \tiny
      \color{darkblue}
      \bf
      4d sphere
    }
    }
  }{
    \mbox{ {\bf (c)} unitary:}\;\;
    A \cdot A
    = 1\;.
  }
$$

\vspace{-1mm}
\noindent From (c) it follows that the matrix
$$
  P_A
    \;:=\;
  \tfrac{1}{2}
  \big(
    \mathrm{I} - A
  \big)
$$
is a projector,
$P_A \cdot P_A = P_A$; and from (b)
it follows that this projector has unit rank:
$$
  \mathrm{tr}(P_A)
  \;=\;
  \tfrac{1}{2}
  \big(\;
    \underset{
      \mathclap{= 2}
    }{
      \underbrace{
        \mathrm{tr}(\mathrm{I})
      }
    }
    -
    \underset{
      \mathclap{= 0}
    }{
      \underbrace{
        \mathrm{tr}(A)
      }
    }
 \; \big)
  \;=\;
  1
  \,.
$$

\vspace{-2mm}
\noindent
Here a unit-rank projector is one for which there exists
$v_A \in \mathbb{H}^2 \setminus \{0\}$
such that
\vspace{-1mm}
$$
  P_A
   \;=\;
   \tfrac{1}{\left\Vert v_A\right\Vert^2}
   \,
   v_A \cdot v_A^\dagger
  \,.
$$
Noticing that $P_A$,
and hence $A \,=\, \mathrm{I} - 2 P_A  $,
thus depends on $v_A$
exactly only through the quaternionic line that
it spans, we have thus found the following  identification
of the 4-sphere with quaternionic projective space:
\vspace{-1mm}
\begin{equation}
  \label{IdentificationOfThe4SphereWithQuaternionicProjectiveSpace}
  \xymatrix@R=-4pt{
    S^4
    \;\simeq\;
    S(\mathbb{R}^5)
    \ar[rr]^-{ \simeq }
    &&
    \mathrm{Mat}^{\mathrm{herm}}_{\mathrm{trless}}
    \big(
      2 \times 2,
      \mathbb{H}
    \big)
    \cap
    \mathrm{U}(2,\mathbb{H})
    \ar@{<-}[rr]^-{ \simeq }
    &&
    \mathbb{H}P^1
    \\
    &&
    \mathrm{I}
    -
    \tfrac{2}{\left\Vert v\right\Vert}
    \,
    v \cdot v^\dagger
    \ar@{}[rr]|-{\longmapsfrom}
    &&
    [v]
  }
\end{equation}

\vspace{-2mm}
\noindent Finally, under the isomorphism on the right of
\eqref{IdentificationOfThe4SphereWithQuaternionicProjectiveSpace}
the canonical $\mathrm{Sp}(2)$-action on $\mathbb{H}P^1$
\vspace{-2mm}
$$
  \xymatrix@R=-2pt{
    \mathrm{Sp}(2)
    \times
    \mathbb{H}P^1
    \ar[rr]
    &&
    \mathbb{H}P^1
    \\
    \big(
      A, [v]
    \big)
    \ar@{}[rr]|-{\longmapsto}
    &&
    \big[
      A \cdot v
    \big]
  }
$$

\vspace{-2mm}
\noindent is manifestly identified with the
conjugation action \eqref{Spin51ByQuaternionicMatrices}.
This implies the claim {(iii)}, by \eqref{Spin5ByQuaternionicMatrices}.
\hfill \end{proof}

\noindent Summarizing
\eqref{6dMinkowskiSpacetimeAsQuaternionicMatrices},
\eqref{Spin5ByQuaternionicMatrices} \&
\eqref{IdentificationOfThe4SphereWithQuaternionicProjectiveSpace}:
\vspace{-2mm}
$$
  \hspace{-8cm}
  \xymatrix@R=6pt{
    \mathllap{
      \mbox{
        \tiny
        \color{darkblue}
        \bf
        \begin{tabular}{c}
          6d Minkowski
          \\
          spacetime
        \end{tabular}
      }
      \;\;
    }
    \mathrm{Spin}(5,1)
    \;
    \underset{\mathrm{can}}{\acts}
    \;\;
    \mathbb{R}^{5,1}
    \ar@{}[r]|-{\simeq}
    &
    \mathrlap{
    \mathrm{SL}(2,\mathbb{H})
    \;
    \underset{\mathrm{adj}}{\acts}
    \;\;
    \overset{
      \mathclap{
      \raisebox{3pt}{
        \tiny
        \color{darkblue}
        \bf
        relativistic quaternionic Pauli matrices
      }
      }
    }{
    \big\{
      A \in \mathrm{Mat}_{{}_{2 \times 2}}(\mathbb{H})
      \,\big\vert\,
      A^\dagger = A
    \big\}
    }
    }
    \\
    \mathllap{
      \mbox{
        \tiny
        \color{darkblue}
        \bf
        \begin{tabular}{c}
          5d Euclidean
          \\
          space
        \end{tabular}
      }
      \;\;\;\;\,\;
    }
    \mathrm{Spin}(5)
   \;
   \underset{\mathrm{can}}{\acts}
   \;\;
    \mathbb{R}^5
    \ar@<-21pt>@{^{(}->}[u]
    \ar@{}[r]|-{\simeq}
    &
    {\phantom{\mathrm{S}}}\,
    \mathrlap{
    \underset{
      \raisebox{-14pt}{
        \scalebox{.8}{$
          = \mathrm{Sp}(2)
        $}
      }
    }{
      \underbrace{
        \mathrm{U}(2,\mathbb{H})
      }
    }
    \;
    \underset{\mathrm{adj}}{\acts}
    \;\;
    \overset{
      \raisebox{3pt}{
        \tiny
        \color{darkblue}
        \bf
        quaternionic Pauli matrices
      }
    }{
    \left\{
      A \in \mathrm{Mat}_{{}_{2 \times 2}}(\mathbb{H})
      \,\big\vert\,
        A^\dagger = A
        \,,\,
        \mathrm{tr}(A) = 0
    \right\}
    }
    }
    \\
    \mathllap{
      \mbox{
        \tiny
        \color{darkblue}
        \bf
        4-sphere
      }
      \;\;\;\;\;\;\;\;\;\;\;
    }
    \mathrm{Spin}(5)
    \;
    \underset{\mathrm{can}}{\acts}
    \;\;
    S^4
    \ar@<-21pt>@{^{(}->}[u]
    \ar@{}[r]|-{\simeq}
   &
    {\phantom{\mathrm{S}}}\,
    \mathrlap{
    \overbrace{
      \mathrm{U}(2,\mathbb{H})
    }
    \;
    \underset{\mathrm{adj}}{\acts}
    \;\;
    \overset{
      \raisebox{3pt}{
        \tiny
        {
        \color{darkblue}
        \bf
        quaternionic rank-1 projectors
        }
        $P = \tfrac{1}{2}\big(\mathrm{I} - A \big)$
      }
    }{
    \underset{
      \simeq \,
      \underset{
        \mathclap{
        \raisebox{-3pt}{
          \tiny
          \color{darkblue}
          \bf
          quaternionic projective line
        }
        }
      }{
        \mathbb{H}P^1
      }
    }{
    \underbrace{
    \left\{
      A \in \mathrm{Mat}_{{}_{2 \times 2}}(\mathbb{H})
      \,\big\vert\,
        A^\dagger = A
        \,,\,
        \mathrm{tr}(A) = 0
        \,,\,
        A^\dagger \cdot A = \mathrm{I}
    \right\}
    }
    }
    }
    }
  }
$$
Of course, there is the analogous situation over the complex numbers:
\vspace{-1mm}
$$
  \hspace{-8cm}
  \xymatrix@R=6pt{
    \mathllap{
      \mbox{
        \tiny
        \color{darkblue}
        \bf
        \begin{tabular}{c}
          4d Minkowski
          \\
          spacetime
        \end{tabular}
      }
      \;\;
    }
    \mathrm{Spin}(3,1)
    \;
    \underset{\mathrm{can}}{\acts}
    \;\;
    \mathbb{R}^{3,1}
    \ar@{}[r]|-{\simeq}
    &
    \mathrlap{
    \mathrm{SL}(2,\mathbb{C})
    \;
    \underset{\mathrm{adj}}{\acts}
    \;\;
    \overset{
      \mathclap{
      \raisebox{3pt}{
        \tiny
        \color{darkblue}
        \bf
        relativistic complex Pauli matrices
      }
      }
    }{
    \big\{
      A \in \mathrm{Mat}_{{}_{2 \times 2}}(\mathbb{C})
      \,\big\vert\,
      A^\dagger = A
    \big\}
    }
    }
    \\
    \mathllap{
      \mbox{
        \tiny
        \color{darkblue}
        \bf
        \begin{tabular}{c}
          3d Euclidean
          \\
          space
        \end{tabular}
      }
      \;\;\;\;\,\;
    }
    \mathrm{Spin}(3)
    \;
    \underset{\mathrm{can}}{\acts}
    \;\;
    \mathbb{R}^3
    \ar@<-21pt>@{^{(}->}[u]
    \ar@{}[r]|-{\simeq}
    &
    \mathrlap{
    \underset{
      \raisebox{-12pt}{
        \scalebox{.8}{$
        $}
      }
    }{
        \mathrm{SU}(2,\mathbb{C})
    }
    \;
    \underset{\mathrm{adj}}{\acts}
    \;\;
    \overset{
      \raisebox{3pt}{
        \tiny
        \color{darkblue}
        \bf
        complex Pauli matrices
      }
    }{
    \left\{
      A \in \mathrm{Mat}_{{}_{2 \times 2}}(\mathbb{C})
      \,\big\vert\,
        A^\dagger = A
        \,,\,
        \mathrm{tr}(A) = 0
    \right\}
    }
    }
    \\
    \mathllap{
      \mbox{
        \tiny
        \color{darkblue}
        \bf
        3-sphere
      }
      \;\;\;\;\;\;\;\;\;\;\;
    }
    \mathrm{Spin}(3)
    \;
    \underset{\mathrm{can}}{\acts}
    \;\;
    S^3
    \ar@<-21pt>@{^{(}->}[u]
    \ar@{}[r]|-{\simeq}
   &
    \mathrlap{
      \mathrm{SU}(2,\mathbb{C})
      \;
    \underset{\mathrm{adj}}{\acts}
    \;\;
    \overset{
      \raisebox{3pt}{
        \tiny
        {
        \color{darkblue}
        \bf
        complex rank-1 projectors
        }
        $P = \tfrac{1}{2}\big(\mathrm{I} - A \big)$
      }
    }{
    \underset{
      \simeq \,
      \underset{
        \mathclap{
        \raisebox{-3pt}{
          \tiny
          \color{darkblue}
          \bf
          complex projective line
        }
        }
      }{
        \mathbb{C}P^1
      }
    }{
    \underbrace{
    \left\{
      A \in \mathrm{Mat}_{{}_{2 \times 2}}(\mathbb{C})
      \,\big\vert\,
        A^\dagger = A
        \,,\,
        \mathrm{tr}(A) = 0
        \,,\,
        A^\dagger \cdot A = \mathrm{I}
    \right\}
    }
    }
    }
    }
  }
$$

\begin{prop}[Equivariance of combined Hopf/twistor fibrations]
  \label{EquivarianceOfCombinedHopfTwistorFibration}
  $\,$

  \vspace{-.5cm}
  \begin{itemize}
  \item[\bf (i)]
  The quaternionic Hopf fibration
  $\xymatrix@C=12pt{S^7 \ar[r]^{ h_{\mathbb{H}} } & S^4 }$
 (Diagram  \eqref{HopfTwistorFactorization})
  is equivariant
  with respect to the action of
  $\mathrm{Sp}(2) \cdot \mathrm{Sp}(1)$ (Def. \ref{QuaternionicGroups})

  \begin{itemize}
  \vspace{-3mm}
  \item[{\bf (a)}] on $S^7$, by
  \vspace{-2mm}
  \begin{equation}
    \label{ActionOnS7}
    \xymatrix@R=-2pt{
      \mathrm{Sp}(2)\cdot \mathrm{Sp}(1) \times S^7
      \ar[rr]
      &&
      S^7
      \\
      \big(
        [A, q'],
        \{v \cdot t \,\left\vert\, t \in \mathbb{R}^\times_+ \right.\}
      \big)
      \ar@{}[rr]|-{\longmapsto}
      &&
      \{
        A \cdot v  \cdot t \cdot q'
        \,\left\vert\, t \in \mathbb{R}^\times_+\right.
      \}
    }
  \end{equation}
  \vspace{-3mm}
  \item[{\bf (b)}] on $S^4$, by
  \vspace{-2mm}
  \begin{equation}
    \label{ActionOnS4}
    \xymatrix@R=-2pt{
      \mathrm{Sp}(2)\cdot \mathrm{Sp}(1) \times S^4
      \ar[rr]
      &&
      S^4
      \\
      \big(
        [A, q'],
        \{v \cdot q \,\left\vert\, q \in \mathbb{H}^\times \right.\}
      \big)
      \ar@{}[rr]|-{\longmapsto}
      &&
      \{
        A \cdot v \cdot q
        \,\left\vert\, q \in \mathbb{H}^\times \right.
      \}
    }
  \end{equation}
  \end{itemize}

  \vspace{-3mm}
  \item[\bf (ii)]
  Its factorization
  $h_{\mathbb{H}} = t_{\mathbb{H}} \circ  h_{\mathbb{C}}$ through the
  combined Hopf/twistor fibrations
  retains equivariance under the subgroup
  $\xymatrix@C=12pt{\mathrm{Sp}(2)\, \ar@{^{(}->}[r] & \mathrm{Sp}(2)\cdot \mathrm{Sp}(1) }$
  \eqref{CanonicalInclusionIntoCentralProductGroup}
  with action on $\mathbb{C}P^3$ given by
  \vspace{-3mm}
  \begin{equation}
    \label{ActionOnCP3}
    \xymatrix@R=-2pt{
      \mathrm{Sp}(2) \times \mathbb{C}P^3
      \ar[rr]
      &&
      \mathbb{C}P^3
      \\
      \big(
        A
        ,
        \{x \cdot z \,\left\vert\, z \in \mathbb{C}^\times \right.\}
      \big)
      \ar@{}[rr]|-{\longmapsto}
      &&
      \{
        A \cdot x \cdot z
        \,\left\vert\, z \in \mathbb{C}\right.
      \}
    }
  \end{equation}
  \end{itemize}
 \end{prop}
 \noindent In summary:
\vspace{-2mm}
$$
  \xymatrix@C=4em{
    S^7
    \ar@(ul,ur)^-{ \scalebox{0.6}{$\mathrm{Sp}(2)\cdot \mathrm{Sp}(1)$} }
    \ar[rr]^-{ h_{\mathbb{H}} }
    &&
    S^4
    \ar@(ul,ur)^-{\scalebox{0.6}{$ \mathrm{Sp}(2)\cdot \mathrm{Sp}(1)$} }
  }
  \phantom{AAA}
  \mbox{and}
  \phantom{AAAA}
  \xymatrix@C=6em{
    S^7
    \ar@(ul,ur)^-{\scalebox{0.6}{$ \mathrm{Sp}(2)$} }
    \ar[r]^-{ h_{\mathbb{C}} }
    \ar@/_1pc/[rr]_-{ h_{\mathbb{H}} }
    &
    \mathbb{C}P^3
    \ar@(ul,ur)^-{\scalebox{0.6}{$ \mathrm{Sp}(2)$} }
    \ar[r]^{ t_{\mathbb{H}} }
    &
    S^4
    \ar@(ul,ur)^-{\scalebox{0.6}{$ \mathrm{Sp}(2)$} }
  }
$$
\begin{proof}
This is essentially immediate from
the presentation of the fibrations in \eqref{HopfTwistorFactorization}
\eqref{HopfTwistorFibrationCosetRealization}:

\noindent {\bf (ii)} Diagram \eqref{HopfTwistorFactorization} makes
manifest that all maps here are equivariant
with respect to left action by $\mathrm{GL}(8,\mathbb{R})$, hence
in particular under left action by $\mathrm{Sp}(2)$, which is
also manifest from \eqref{HopfTwistorFibrationCosetRealization}:
\vspace{-2mm}
\begin{equation}
  \label{Sp2EquivarianceOfHopfTwistorFibrations}
  \xymatrix{
    \{ A \cdot v \cdot t
      \,\left\vert\,
      t \in \mathbb{R}^\times_+
      \right.
    \}
   \; \ar@{|->}[r]^-{ h_{\mathbb{C}} }
    &
    \;
    \{ A \cdot v \cdot z
      \,\left\vert\,
      z \in \mathbb{C}^\times
      \right.
    \}
 \;   \ar@{|->}[r]^-{ t_{\mathbb{H}} }
    &
    \;
    \{ A \cdot v \cdot q
      \,\left\vert\,
      q \in \mathbb{H}^\times
      \right.
    \}.
  }
\end{equation}

\vspace{-2mm}
\noindent {\bf (i)} We see that the total quaternionic
Hopf fibration is also equivariant under the right $\mathrm{Sp}(1)$-action,
due to the fact that the reals commute with the quaternions:
\vspace{-2mm}
\begin{equation}
  \label{TowardsRightSp1EquivarianceOfQuaternionicHopf}
  \xymatrix{
    \{ v \cdot t \cdot q'
      \,\left\vert\,
      t \in \mathbb{R}^\times_+
      \right.
    \}
    \;=\;
    \{ v \cdot q' \cdot t
      \,\left\vert\,
      t \in \mathbb{R}^\times_+
      \right.
    \}
   \; \ar@{|->}[rr]^-{ h_{\mathbb{H}} }
    &&
    \;
    \{ A \cdot q' \cdot q
      \,\left\vert\,
      q \in \mathbb{H}^\times
      \right.
    \}
    \;=\;
    \{ A \cdot q
      \,\left\vert\,
      q \in \mathbb{H}^\times
      \right.
    \}.
  }
\end{equation}
Moreover, since the left multiplication action by $\mathrm{Sp}(2)$
evidently commutes with the right multiplication action by $\mathrm{Sp}(1)$
and since
$-1 \in \mathrm{Sp}(n)$ is central, this
generates the claimed $\mathrm{Sp}(2)\cdot \mathrm{Sp}(1)$-action.
(In fact, this is the maximal symmetry group of $h_{\mathbb{H}}$
\cite[4.1]{GWZ86}\cite[2.20]{FSS19b}.)
\hfill
\end{proof}

\begin{remark}[Twistor space breaks equivariance to $\mathrm{Sp}(2)$]
 \label{TwistorSpacesBreaksEquivarianceToSp2}
Notice that the
factorization of the quaternionic Hopf fibration through $\mathbb{C}P^3$
is \emph{not} equivariant under the
further right $\mathrm{Sp}(1)$-action from
\eqref{ActionOnS7} and \eqref{ActionOnS4}, Indeed,
the computation analogous to \eqref{TowardsRightSp1EquivarianceOfQuaternionicHopf}
now gives
\vspace{-2mm}
$$
  \xymatrix{
    \{ v \cdot t \cdot q'
      \,\left\vert\,
      t \in \mathbb{R}^\times_+
      \right.
    \}
    \;=\;
    \{ v \cdot q' \cdot t
      \,\left\vert\,
      t \in \mathbb{R}^\times_+
      \right.
    \}
   \; \ar@{|->}[rr]^-{ h_{\mathbb{C}} }
    &&
    \;
    \{ A \cdot q' \cdot z
      \,\left\vert\,
      z \in \mathbb{C}^\times
      \right.
    \}
    \;
    \underset{
      \mathclap{
        \mbox{
          \tiny
          in general
        }
      }
    }{\mbox{\color{red}$\neq$}}\;
    \{ A \cdot z \cdot q'
      \,\left\vert\,
      z \in \mathbb{C}^\times
      \right.
    \}
  }
$$

\vspace{-2mm}
\noindent since the complex numbers do not commute with the quaternions.
Therefore, factoring the quaternionic Hopf fibration through the
twistor fibration \eqref{HopfTwistorFactorization} breaks its symmetry from
$\mathrm{Sp}(2) \cdot \mathrm{Sp}(1)$ to $\mathrm{Sp}(2)$
\eqref{CanonicalInclusionIntoCentralProductGroup}.
\end{remark}

\begin{remark}[Homotopy quotients and Borel construction]
  \label{HomotopyQuotientAndBorelConstruction}
  For $X$ a topological space equipped with a continuous
  action by a topological group $G$, the
  \emph{homotopy quotient} $X \!\sslash\! G$ is the homotopy
  type which is represented by the
  \emph{Borel space} $\big(X \times E G\big)/_{\mathrm{diag}}G$,
  where $E G$ denotes the universal $G$-principal bundle:
  \vspace{-2mm}
  \begin{equation}
    \label{HomotopyQuotientByBorelConstruction}
    \overset{
      \mathclap{
      \raisebox{3pt}{
        \tiny
        \color{darkblue}
        \bf
        homotopy quotient
      }
      }
    }{
      X \!\sslash\! G
    }
    \;\;\;\;\;\;
    \simeq_{\mathrm{whe}}\;
    \overset{
      \mathclap{
      \raisebox{3pt}{
        \tiny
        \color{darkblue}
        \bf
        Borel construction
      }
      }
    }{
      \big(X \times E G\big)/_{\mathrm{diag}}G
    }\;.
  \end{equation}

  \vspace{-2mm}
  \noindent  {\bf (i)} This construction is clearly functorial: On the right
  this is a 1-functor on the category of topological spaces
  equipped with group actions, while on the left this is an
  $\infty$-functor on the $\infty$-category $\mathrm{Groupoids}_\infty$
  equipped with $\infty$-actions,
  see \cite[\S 4]{NSS12a}\cite[\S 2.2]{SS20b}.

  \noindent {\bf (ii)} In the special case when $X = \ast$ is the point,
  the Borel space is the classifying space $B G$. With (i), this means
  that topological group homomorphisms
  $G_1 \overset{\phi}{\longrightarrow} G_2$ induce maps of classifying spaces
   \vspace{-2mm}
  \begin{equation}
    \label{DeloopingOfMaps}
    \xymatrix{
      B G_1
      \ar[r]^-{ B \phi }
      &
      B G_2
    }.
  \end{equation}
\end{remark}

\begin{defn}[Borel-equivariant Hopf/twistor fibrations]
 \label{ParametrizedHopfTwistorFibration}
We say that the \emph{$\mathrm{Sp}(2)$-Borel-equivariant Hopf-twistor fibrations}
are the image (in homotopy types of topological spaces)
of the Hopf/twistor fibrations \eqref{HopfTwistorFactorization}
under taking the homotopy quotient \eqref{HomotopyQuotientByBorelConstruction}
by their compatible
$\mathrm{Sp}(2)$-action of Prop. \ref{EquivarianceOfCombinedHopfTwistorFibration}:

\vspace{-4mm}
\begin{equation}
  \label{ParametrizedHopfAndTwistorFibration}
  \hspace{-7mm}
  \raisebox{3pt}{
  \xymatrix@C=14em{
    S^7 \!\sslash\! \mathrm{Sp}(2)
    \ar[r]|-{\;
        h_{\mathbb{C}} \sslash \mathrm{Sp}(2)
    \;}_-{
        \raisebox{-3pt}{
          \tiny
          \color{darkblue}
          \bf
            parametrized complex Hopf fibration
        }
        \phantom{\mathclap{\vert}}
    }
    \ar@/^2pc/[rr]|-{\;
      h_{\mathbb{H}} \sslash \mathrm{Sp}(2)
   \; }^-{
      \raisebox{3pt}{
        \tiny
        \color{darkblue}
        \bf
        parametrized quaternionic Hopf fibration
      }
      \phantom{\mathclap{\vert}}
    }
    \ar@/_1pc/[dr]
    &
    \mathbb{C}P^3 \!\sslash\! \mathrm{Sp}(2)
    \ar[d]
    \ar[r]|-{\;
      t_{\mathbb{H}} \sslash \mathrm{Sp}(2)
    \;}_-{
      \mbox{
        \tiny
        \color{darkblue}
        \bf
        parametrized twistor fibration
      }
      \phantom{\mathclap{\vert^\vert}}
    }
    &
    S^4 \!\sslash\! \mathrm{Sp}(2)
    \ar@/^1pc/[dl]
    \\
    &
    B
    \mathrm{Sp}(2)
  }
  }
\end{equation}

\end{defn}

\begin{lemma}[Coset spaces as homotopy fibers {\cite[2.7]{FSS19b}\cite[2.79]{SS20b}}]
  \label{CosetSpaceAsHomotopyFiber}
  Let $\xymatrix@C=12pt{H\; \ar@{^{(}->}[r]^-{i} & G}$ be an inclusion of
  topological groups.

\noindent  {\bf (i)} The homotopy type of the corresponding
  coset space $G/H$ is, equivalently, the homotopy fiber of the
  induced morphism \eqref{DeloopingOfMaps} on classifying spaces.

    \noindent {\bf (ii)} The homotopy quotient of the coset space by $G$ is
  homotopy equivalent to the classifying space of $H$:
   \vspace{-2mm}
  $$
    \xymatrix@R=1.2em{
      G/H
      \ar[rr]^-{ \mathrm{hofib}(B\i) }
      &&
      B H
      \ar@{}[r]|-{\simeq}
      \ar[d]^-{ B i }
      &
      \big(G/H\big) \!\sslash\! G
      \\
      && B G
    }
  $$
\end{lemma}
The following Prop. \ref{BorelEquivariantTwistorFibrationAsMapsOfClassifyingSpaces} is the twistorial version of
\cite[Prop. 2.22]{FSS19b}.
\begin{prop}[Borel-equivariant twistor fibration as sequence of classifying spaces]
  \label{BorelEquivariantTwistorFibrationAsMapsOfClassifyingSpaces}
  The
  Borel-equivariant Hopf/twistor fibration
  (Def. \ref{ParametrizedHopfTwistorFibration})
  is homotopy equivalent to the following sequence of classifying
  spaces:
   \vspace{-2mm}
  $$
    \xymatrix@R=14pt{
      S^7 \!\sslash\! \mathrm{Sp}(2)
      \ar[rr]^-{ h_{\mathbb{C}} \sslash \mathrm{Sp}(2) }
      \ar@{}[d]|-{
        \rotatebox[origin=c]{-90}{
          $\simeq$
        }
      }
      &&
      \mathbb{C}P^3
      \!\sslash\! \mathrm{Sp}(2)
      \ar[rr]^-{ t_{\mathbb{H}} \sslash \mathrm{Sp}(2) }
      \ar@{}[d]|-{
        \rotatebox[origin=c]{-90}{
          $\simeq$
        }
      }
      &&
      S^4 \!\sslash\! \mathrm{Sp}(2)
      \ar@{}[d]|-{
        \rotatebox[origin=c]{-90}{
          $\simeq$
        }
      }
      \\
      B
      \big(
        \mathrm{Sp}(1)_L
      \big)
      \ar[rr]
      &&
      B
      \big(
        \mathrm{Sp}(1)_L \times \mathrm{U}(1)_R
      \big)
      \ar[rr]
      &&
      B
      \big(
        \mathrm{Sp}(1)_L \times \mathrm{Sp}(1)_R
      \big)
    }
  $$

   \vspace{-2mm}
\noindent  where the maps on the bottom are the deloopings \eqref{DeloopingOfMaps}
  of the canonical group inclusions (Example \ref{SubgroupsOfQuaternionLinearGroups}).
\end{prop}
\begin{proof}
  With Lemma \ref{CosetSpaceAsHomotopyFiber}
  this follows from the $\mathrm{Sp}(2)$-coset space realization of the
  Hopf/twistor fibration in \eqref{HopfTwistorFibrationCosetRealization}.
\end{proof}

\begin{lemma}[Borel-equivariant Hopf/twistor fibrations are spherical]
  \label{ParametrizedHopfTwistorFibrationsAreSpherical}
  The Borel-equivariant
  Hopf/twistor fibrations \eqref{ParametrizedHopfAndTwistorFibration}
  are still spherical fibrations:
\vspace{-2mm}
  $$
    \xymatrix@C=4em{
      S^1
      \ar[r]
      \ar[d]
      \ar@{}[dr]|-{ \mbox{\tiny\rm(pb)} }
      &
      S^3
      \ar[d]
      \ar[r]
      \ar@{}[dr]|-{ \mbox{\tiny\rm(pb)} }
      &
      S^7 \!\sslash\! \mathrm{Sp}(2)
      \ar[d]^-{ h_{\mathbb{C}} \sslash \mathrm{Sp}(2) }
      \\
      \ast
      \ar[r]
      &
      S^2
      \ar@{}[dr]|-{ \mbox{\rm\tiny\rm(pb)} }
      \ar[r]
      \ar[d]
      &
      \mathbb{C}P^3 \!\sslash\! \mathrm{Sp}(2)
      \ar[d]^-{ t_{\mathbb{H}} \sslash \mathrm{Sp}(2) }
      \\
      &
      \ast
      \ar[r]
      &
      S^4 \!\sslash\! \mathrm{Sp}(2)
    }
  $$
\end{lemma}
\begin{proof}
  This follows on general grounds, as in {\cite[Remark 3.17]{FSS19b}}.
  More concretely, by Prop. \ref{BorelEquivariantTwistorFibrationAsMapsOfClassifyingSpaces}
  and using again Lemma \ref{CosetSpaceAsHomotopyFiber} we have:
   \vspace{-2mm}
  $$
    \raisebox{20pt}{
    \xymatrix{
      \mathrm{fib}
      \big(
        h_{\mathbb{C}} \sslash \mathrm{Sp}(2)
      \big)
      \ar[r]
      &
      S^7 \!\sslash\! \mathrm{Sp}(2)
      \ar[d]^-{ h_{\mathbb{C}} \sslash \mathrm{Sp}(2) }
      \\
      &
      \mathbb{C}P^3 \!\sslash\! \mathrm{Sp}(2)
    }
    }
    \phantom{AAAA}
    \simeq
    \phantom{AAAA}
    \raisebox{20pt}{
    \xymatrix{
      \mathrm{U}(1)
      \ar[r]
      &
      B
      \big(
        \mathrm{Sp}(1)_L
      \big)
      \ar[d]
      \\
      &
      B
      \big(
        \mathrm{Sp}(1)_L \times \mathrm{U}(1)_R
      \big)
    }
    }
  $$
  and
  $$
    \raisebox{20pt}{
    \xymatrix{
      \mathrm{fib}
      \big(
        t_{\mathbb{H}} \sslash \mathrm{Sp}(2)
      \big)
      \ar[r]
      &
      \mathbb{C}P^3 \!\sslash\! \mathrm{Sp}(2)
      \ar[d]^-{ t_{\mathbb{H}} \sslash \mathrm{Sp}(2) }
      \\
      &
      S^4 \!\sslash\! \mathrm{Sp}(2)
    }
    }
    \phantom{AAAA}
    \simeq
    \phantom{AAAA}
    \raisebox{20pt}{
    \xymatrix{
      \mathrm{SU}(2)/\mathrm{U}(1)
      \ar[r]
      &
      B
      \big(
        \mathrm{Sp}(1)_L
        \times
        \mathrm{U}(1)_R
      \big)
      \ar[d]
      \\
      &
      B
      \big(
        \mathrm{Sp}(1)_L \times \mathrm{Sp}(1)_R
      \big)
    }
    }
  $$
\hfill\end{proof}

\begin{theorem}[Integral cohomology of Borel-equivariant Hopf/twistor-fibration]
  \label{CohomologyOfBorelEquivariantTwistorFibration}
  $\,$ \\
  \noindent {\bf (i)}
  The integral cohomology of the space $S^4 \!\sslash\! \mathrm{Sp}(2)$
  in \eqref{ParametrizedHopfAndTwistorFibration} is free on two generators
  in degree 4
   \vspace{-2mm}
  \begin{equation}
    \label{IntegralCohomologyOfS4ModSp2}
    H^\bullet
    \big(
      S^4 \!\sslash\! \mathrm{Sp}(2)
      ;\,
      \mathbb{Z}
    \big)
    \;\simeq\;
    \mathbb{Z}
    \big[
      \widetilde \Gamma_4,\,
      \widetilde \Gamma_4^{\mathrm{vac}}
    \big]
  \end{equation}

   \vspace{-3mm}
\noindent  with the property that their evaluation on the fundamental
  class of the 4-sphere fiber
  $\xymatrix@C=12pt{S^4 \ar[r] & S^4 \sslash \mathrm{Sp}(2)}$
  is unity and zero, respectively:
  \begin{equation}
    \label{EvaluationOfUniversalIntregral4ClassesOn4SphereFiber}
    \big\langle
      \widetilde \Gamma_4
      ,\,
      S^4
    \big\rangle
    \;=1\;
    ,
    \phantom{AAAAAA}
    \big\langle
      \widetilde \Gamma^{\mathrm{vac}}_4
      ,\,
      S^4
    \big\rangle
    \;=\;
    0
    \,.
  \end{equation}

  \noindent {\bf (ii)}
  The integral cohomology of the space
  $\mathbb{C}P^3 \!\sslash\! \mathrm{Sp}(2)$
  in \eqref{ParametrizedHopfAndTwistorFibration} is free on two generators
  in degrees 4 and 2, respectively:
   \vspace{-1mm}
  \begin{equation}
    \label{IntegralCohomologyOfS4ModSp2}
    H^\bullet
    \big(
      \mathbb{C}P^3 \!\sslash\! \mathrm{Sp}(2)
      ;\,
      \mathbb{Z}
    \big)
    \;=\;
    \mathbb{Z}
    \big[
      \mathrm{c}_2^L,\,
      c_1^R
    \big]
    \,.
  \end{equation}

  \noindent {\bf (iii)}  The two are related in that pullback in integral cohomology
  along the Borel-equivariant twistor fibration
  (Def. \ref{ParametrizedHopfTwistorFibration})
  takes the difference of the former generators to
  the cup-square of the latter:
   \vspace{-2mm}
  \begin{equation}
    \label{PullbackInIntegralCohomologyAlongEquivariantTwistorFibration}
    \xymatrix@R=-2pt{
      H^\bullet
      \big(
        S^4 \!\sslash\! \mathrm{Sp}(2)
        ;
        \,
        \mathbb{Z}
      \big)
      \ar[rr]^-{ \scalebox{0.7}{$
        \big(
          t_{\mathbb{H}} \sslash \mathrm{Sp}(2)
        \big)^\ast
        $}
      }
      &&
      H^\bullet
      \big(
        \mathbb{C}P^3 \!\sslash\! \mathrm{Sp}(2)
        ;
        \,
        \mathbb{Z}
      \big)
      \\
      \widetilde \Gamma_4
      -
      \widetilde \Gamma_4^{\mathrm{vac}}
      \ar@{}[rr]|-{\;\;\;\;\;\;\; \longmapsto}
      &&
      - a
        \;
          :=
        \;
        c_1^R \cup c_1^R
      \\
      \widetilde \Gamma_4
      \ar@{}[rr]|-{\longmapsto}
      &&
      c_2^L
    }
  \end{equation}
\end{theorem}
\begin{proof}
  Consider Diagram \eqref{TheComponents} in
  \hyperlink{FigureC}{Figure I}. The top part shows the
  equivalence of the Borel-equivariant Hopf/twistor fibration
  to a sequence of classifying spaces, according to
  Prop. \ref{BorelEquivariantTwistorFibrationAsMapsOfClassifyingSpaces}.
  On the top right we are making fully explicit the factor-wise
  nature of the corresponding maps, using Example \ref{SubgroupsOfQuaternionLinearGroups}
  and Remark \ref{HomotopyQuotientAndBorelConstruction}.

  The bottom part of the diagram shows the corresponding identification
  of the cohomology generators according to \cite[3.9]{FSS19b}.
  This involves the observation that:

  \noindent {\bf (a)}
  Half the universal Euler 4-class
  on $B \mathrm{Spin}(4)$
  is (e.g., \cite[\S 2]{BottCattaneo98})
  the class of the fiberwise unit volume form on the
  universal $S^4$-fibration,
  under the identification from
  Prop. \ref{BorelEquivariantTwistorFibrationAsMapsOfClassifyingSpaces}:
  \vspace{-.3cm}
  \begin{equation}
    \label{HalfEuler4Class}
    \raisebox{10pt}{
    \xymatrix@R=1pt{
    1\cdot [\mathrm{dvol}]
     \!\!\!\!\!\!\!  \ar@{}[r]|-{\longmapsfrom}
      &
    \!\!\!\!\!\!\!   \tfrac{1}{2}\rchi_4
      \mathrlap{
        \;\;\in\;
        H^4
        \big(
          B \mathrm{Spin}(4;
          \,
          \mathbb{R}
        \big)
      }
      \\
      S^4 \ar[r]
      &
      S^4 \sslash \mathrm{Spin}(5)
      \ar@{}[r]|-{\simeq}
      &
      B \mathrm{Spin}(4)
    }
    }
  \end{equation}

  \vspace{-.2cm}

  \noindent {\bf (b)} The fractional Euler class by itself is not integral,
  but its shift by $\tfrac{1}{4}p_1$ (which is also not integral by itself)
  is (the rational image of)
  an integral generator $\widetilde \Gamma_4 = \tfrac{1}{2}\rchi_4 +
  \tfrac{1}{4}p_1$ (e.g. \cite[Lemma 2.1]{CV98a}).

  Together, (a) and (b)  yield the claim \eqref{EvaluationOfUniversalIntregral4ClassesOn4SphereFiber},
  and make manifest that the pullback in question is equivalently that
  of the negative of the left Chern class $- c_2^L$
  along the map on classifying space $\!\!\xymatrix{B \mathrm{U}(1)
   \ar[rr]^{B(c\mapsto \mathrm{diag}(c, \overline c))}
  && B \mathrm{SU}(2)}\!\!$, induced from the inclusion
  $
   \xymatrix@C=14pt{
     {\rm U}(1)
     \ar[r]^-{\simeq}
     &
     \mathrm{S}
     \big(
       \mathrm{U}(1)^2
     \big)
     \;
     \ar@{^{(}->}[r]
     &
     {\rm SU}(2)\,,
    }
  $
  hence is $c_1 \cup (- c_1)$, which yields the last claim \eqref{PullbackInIntegralCohomologyAlongEquivariantTwistorFibration}.
\hfill\end{proof}

\medskip

\subsection{Rational homotopy type}

We compute (in Theorem \ref{SullivanModelOfParametrizedHopfTwistorFibrations})
the relative Sullivan model for the Borel-equivariant
Hopf/twistor fibration from Def. \ref{ParametrizedHopfTwistorFibration},
with generators normalized such as to match their integral pre-images
from Theorem \ref{CohomologyOfBorelEquivariantTwistorFibration}.

\medskip

\noindent
{\bf Notation.} We use the following notation for dg-algebraic rational
homotopy theory
(following \cite{FSS16a}, exposition in \cite[\S 3]{FSS19a}
full details in \cite[\S 3]{FSS20c}):

\noindent {\bf (i)} For $X$ a
(nilpotent, e.g. simply connected) topological space, we
write $\mathrm{CE}\big( \mathfrak{l}X\big)$
for its {\it Sullivan model}, namely for the
{\it minimal}
real differential graded-commutative (dgc) algebra (``FDA'') which
is quasi-isomorphic to the piecewise polynomial de Rham complex
of $X$ (which for $X$ a smooth manifold is itself quasi-isomorphic to the
ordinary de Rham complex).

\noindent {\bf (ii)} Our notation
is meant to be suggestive of the fact that
this is the Chevalley-Eilenberg algebra $\mathrm{CE}(-)$ (\cite[Def. 3.25]{FSS19a},
in generalization of classical CE-algebras computing Lie algebra cohomology
\cite[Ex. 3.24]{FSS19a})
of an $L_\infty$-algebra (\cite[Rem. 3.45]{FSS19a}),
namely
of the real {\it Whitehead $L_\infty$-algebra} $\mathfrak{l}X$ of $X$
(\cite[Prop. 3.67]{FSS19a}):\footnote{
This passage \eqref{RationalDgCorrespondence}
through Whitehead $L_\infty$-algebras makes transparent
how it is that dgc-algebras know about homotopy types and how
dgc-algebra homomorphisms between these encode $L_\infty$-algebra
valued higher gauge fields \cite[\S 3.3]{FSS19a}, but
for the purpose of the present article the reader may ignore $L_\infty$-algebra theory
and regard the notation $\mathrm{CE}(\mathfrak{l}(-))$ as a primitive for
Sullivan models.
}
\vspace{-.3cm}
\begin{equation}
  \label{RationalDgCorrespondence}
  \begin{tikzcd}[column sep=2pt]
   \overset{
     \mathclap{
       \raisebox{6pt}
       {
         \tiny
         \color{darkblue}
         \bf
         \def\arraystretch{.9}
         \begin{tabular}{c}
           rational
           \\
           topological
         \end{tabular}
       }
     }
   }{
     \underset{
       \mathclap{
         \raisebox{-2pt}{
           \tiny
           \color{darkblue}
           \bf
           space
         }
       }
     }{
       X
     }
   }
    &\longleftrightarrow&
    \overset{
      \mathclap{
      \raisebox{6pt}{
        \tiny
        \color{darkblue}
        \bf
        \def\arraystretch{.9}
        \begin{tabular}{c}
          higher
          \\
          Whitehead
        \end{tabular}
      }
      }
    }{
      \underset{
        \raisebox{-2pt}{
          \tiny
          \color{darkblue}
          \bf
          $L_\infty$-algebra
        }
      }{
        \mathfrak{l}X
      }
    }
    &\longleftrightarrow&
    \overset{
      \mathclap{
      \raisebox{6pt}{
        \tiny
        \color{darkblue}
        \bf
        Sullivan
      }
      }
    }{
      \underset{
       \mathclap{
       \raisebox{-2pt}{
         \tiny
         \color{darkblue}
         \bf
         dgc-algebra
       }
       }
      }{
        \mathrm{CE}(\mathfrak{l}X)
      }
    }
    &\underset{\mathrm{qi}}{\simeq}&
    \Omega^\bullet_{\mathrm{PL}}(X)
  \end{tikzcd}
\end{equation}
\vspace{-.3cm}

\noindent
Moreover, we give these minimal dgc-algebras by their polynomial
generators $\omega_n$ in some degree $n$, quotiented out by their differential relations
$d \omega_n \,=\, P(...)$ for $P$ some polynomial in generators of lower degree.

\noindent {\bf (iii)} For example (e.g. \cite{Menichi13}\cite[Ex. 3.71-2]{FSS19a}),
the Sullivan models of
Eilenberg-MacLane spaces and of spheres are:
$$
  \mathrm{CE}
  \big(
    B^n \mathbb{Z}
  \big)
  \;\simeq\;
  \mathbb{R}[n]\big/(d = 0)
  \quad
  \underset{
    \mathclap{
      n = 2k+1
    }
  }{
    \simeq
  }
  \quad
  \mathrm{CE}(\mathfrak{l}S^{2k+1})
  \,,
  \;\;\;\;\;\;
  \mathrm{CE}
  \big(
    S^{2k}
  \big)
  \;\simeq\;
  \mathbb{R}[\omega_{2k}, \omega_{4k-1}]\bigg/
  \left(
    \!\!\!
    \begin{array}{ll}
      d\, \omega_{4k-1} & = - \omega_{2k} \wedge \omega_{2k}
      \\
      d\, \omega_{2k} & = 0
    \end{array}
    \!\!\!
  \right)
  \,.
$$

\begin{lemma}[Normalized Sullivan model of spherical fibrations {\cite[p. 202]{FHT00}; see \cite[2.5]{FSS19b}}]
  \label{SullivanModelOfSphericalFibration}
  Let $X$ be a topological space with Sullivan model
  $\mathrm{CE}\big( \mathfrak{l} X \big) \in \mathrm{dgcAlgebras}_{\mathbb{R}}$
  \eqref{RationalDgCorrespondence}.
  Then the relative minimal Sullivan model for  a
  $S^{n}$-fibration $\xymatrix@C=11pt{Y \ar[r] & X}$
  is of the following form:

  \noindent {\bf (i)} for $n = 2k+1$ odd:
  \vspace{-4mm}
  \begin{equation}
    \xymatrix@R=1.5em{
      S^{2k+1}
      \ar[r]
      &
      Y
      \ar[d]
      &\;\;\;\;&
      \mathrm{CE}
      \big(
        \mathfrak{l} X
      \big)
      \!\!
      \big[
        \omega_{2k+1}
      \big]
      \Big/\!
      \left(
        {\begin{aligned}
          d\,\omega_{2k+1} & = \alpha_{2k+2}
        \end{aligned}}
      \!\!\!\!\!\! \right)
      \ar@{<-^{)}}[d]
      \\
      &
      X
      &&
      \mathrm{CE}
      \big(
        \mathfrak{l}X
      \big)
    }
  \end{equation}

  \noindent {\bf (ii)} for $n = 2k$ even:
  \vspace{-5mm}
  \begin{equation}
    \label{SullivanModelForEvenDimensionalSphericalFibration}
    \xymatrix@R=1.5em{
      S^{2k}
      \ar[r]
      &
      Y
      \ar[d]
      &\;\;\;\;&
      \mathrm{CE}
      \big(
        \mathfrak{l}
        X
      \big)
      \!
      \left[
      \!\!\!\!
      {\begin{array}{c}
        \omega_{2k},
        \\
        \omega_{4k-1}
      \end{array}}
      \!\!\!\!
      \right]
      \Big/
      \left(
      {\begin{aligned}
        \!\!\!\!d \,\omega_{2k} & = 0
        \\
        d \,\omega_{4k-1} & = - \omega_{2k} \wedge \omega_{2k} + \alpha_{4k}
      \end{aligned}}
      \right)
      \ar@{<-^{)}}[d]
      \\
      &
      X
      &&
      \mathrm{CE}
      \big(
        \mathfrak{l}
        X
      \big)
    }
  \end{equation}

  \vspace{-2mm}
\noindent  for some closed $\alpha \in \mathrm{CE}\big(\mathfrak{l}X\big)$
  (which can be characterized further, see \cite[p. 202]{FHT00}\cite[2.5]{FSS19b}).

  \noindent {\bf (iii)} The differential in \eqref{SullivanModelForEvenDimensionalSphericalFibration}
  is \emph{normalized} so that the generators $\omega_{d}$ restrict
  to the unit volume forms on the respective sphere fibers
  (see \cite[Lemma 3.19]{FSS19b}):
  \begin{equation}
    \label{NormalizationOfGeneratorsInSullivanModelForSpheres}
    \langle \omega_{2k},\, S^{2k} \rangle
    \;=\;
    1
    \,,
    \phantom{AAAA}
    \langle \omega_{4k-1},\, S^{4k-1} \rangle
    \;=\;
    1
    \,.
  \end{equation}
\end{lemma}

The action of triality group automorphisms
on $\mathrm{Spin}(8)$ famously relates
three distinct conjugacy classes of subgroup inclusions of
$\mathrm{Spin}(7)$. Less widely appreciated is another triple
of subgroups of $\mathrm{Spin}(8)$ that is permuted under triality:

\begin{lemma}[Triality on central product groups in $\mathrm{Spin}(8)$ {\cite[2.17]{FSS19b}}]
  \label{TrialityOnCentralProductGroups}
  Under the triality automorphisms of $\mathrm{Spin}(8)$
  the \emph{canonical} subgroup inclusions of the central product groups
  $\mathrm{Spin}(5)\cdot \mathrm{Spin}(3)$
  and $\mathrm{Sp}(2)\cdot \mathrm{Sp}(1)$ (Def. \ref{SpecialQuaternionLinearGroup})
  turn into each other:
  \vspace{-2mm}
  \begin{equation}
    \label{CentralProductTriality}
    \xymatrix@C=4em@R=1.5em{
      \mathrm{Sp}(2)\cdot \mathrm{Sp}(1)
      \ar[d]_-{\simeq}
      \; \ar@{^{(}->}[r]^-{ i_{\mathrm{Sp}} }
      &
      \mathrm{Spin}(8)
      \ar[d]^-{ \mathrm{tri} }_-\simeq
      \\
      \mathrm{Spin}(5)\cdot \mathrm{Spin}(3)
      \; \ar@{^{(}->}[r]^-{ i_{\mathrm{Spin}} }
      &
      \mathrm{Spin}(8)
    }
  \end{equation}
\end{lemma}

\begin{lemma}[Sullivan model for $B \mathrm{Spin}(5)$ and $B \mathrm{Sp}(2)$ {\cite[2.19]{FSS19b}}]
 \label{SullivanModelForBSp2}
 Minimal sullivan models for $B \mathrm{Spin}(5)$ and $B \mathrm{Sp}(2)$,
 and their relation under triality \eqref{SullivanModelForSp2} are given,
 up to isomorphism, as follows:
 \vspace{-1mm}
 \begin{equation}
   \label{SullivanModelForSp2}
   \hspace{-2cm}
   \raisebox{26pt}{
   \xymatrix@C=3em@R=1.3em{
     B
     \mathrm{Spin}(8)
     \ar[r]^-{ B i_{\mathrm{Sp}} }
     \ar[dd]_-{ B \mathrm{tri} }^-{ \simeq }
     &
     B \mathrm{Sp}(2)
     \ar[dd]^-{\simeq}
     &&
     \mathbb{R}
     \!
     \big[
       \!\!\!
       \begin{array}{c}
         \tfrac{1}{2}p_1,
         \rchi_8
       \end{array}
       \!\!\!
     \big]
     \mathrlap{
       \; = \;
       \mathrm{CE}
       \big(
         \mathfrak{l}
         B \mathrm{Sp}(2)
       \big)
     }
     \ar@{<-}[dd]^-{
       \scalebox{.7}{$
       \begin{array}{cc}
         \tfrac{1}{2}p_1
         &
         - \rchi_8
         \mathrlap{
           + \big( \tfrac{1}{4}p_1 \big)^2
         }
         \\
         \mapsup
         &
         \mapsup
         \\
         \tfrac{1}{2}p_1
         &
         \tfrac{1}{4}p_2
       \end{array}
       $}
     }
     \\
     \\
     B \mathrm{Spin}(8)
     \ar[r]^-{ B i_{\mathrm{Spin}} }
     &
     B \mathrm{Spin}(5)
     &&
     \mathbb{R}
     \!
     \big[
       \!\!\!
       \begin{array}{c}
         \tfrac{1}{2}p_1,
         p_2
       \end{array}
       \!\!\!
     \big]
     \mathrlap{
       \;
       =:
       \mathrm{CE}
       \big(
         \mathfrak{l}
         B \mathrm{Spin}(5)
       \big)
     }
   }
   }
 \end{equation}
where (by \cite[8.1, 8.2]{CV98b}\cite[3.7]{FSS19b}):
\vspace{0mm}
\begin{equation}
  \label{chi8RelationOnBSp2}
  \tfrac{1}{2}
  \rchi_8
  \;=\;
  \tfrac{1}{4}
  \Big(
    p_2
    -
    \big(
      \tfrac{1}{2}p_1
    \big)^2
  \Big)
  \;\;
  \in
  \;
  H^4
  (
    B \mathrm{Sp}(2);
    \,
    \mathbb{R}
  )\;.
\end{equation}
\end{lemma}

\begin{lemma}[Normalized Sullivan model for plain Hopf/twistor fibrations]
  \label{SullivanModelsForHopfTwistorFibrations}
  The minimal relative Sullivan model for the plain Hopf/twistor fibrations \eqref{HopfTwistorFactorization}
  is as follows:
  \vspace{-2mm}
\begin{equation}
  \label{SullivanModelForHopfTwistorFibrations}
  \hspace{-1cm}
  \raisebox{135pt}{
  \xymatrix@C=30pt{
    &
    S^7
    \ar[dd]_-{
      h_{\mathbb{C}}
      \phantom{\mathclap{\vert}}
    }
    &
    \mathbb{R}
    \!\!
    \left[
      \!\!\!
      {\begin{array}{c}
        h_1,
        \\
        f_2,
        \\
        h_3,
        \\
        \omega_4,
        \\
        \omega_7
      \end{array}}
      \!\!\!
    \right]
    \!\Big/\!\!
    \left(
    {\begin{aligned}
      d\,h_1 & = f_2
      \\
      d\,f_2 & = 0
      \\
      d\,h_3 & = \omega_4
      -
      f_2 \wedge f_2
      \\
      d\,\omega_4 & = 0
      \\
      d\,\omega_7
      & = - \omega_4 \wedge \omega_4
    \end{aligned}}
    \right)
    \ar[rr]^-{ \simeq }
    &\;\;\;\;\;&
    \mathbb{R}
    \left[
            \omega_7
          \right]
    \!\Big/\!
    \left(
    {\begin{aligned}
      d\,\omega_7
      & = 0
    \end{aligned}}
   \!\!\!\!\!\!\right)
    \\
    \\
    &
    \mathbb{C}P^3
    \ar[dd]_-{
      t_{\mathbb{H}}
      \phantom{\mathclap{\vert}}
    }
    &
    \mathbb{R}
    \!\!
    \left[
      \!\!\!
      {\begin{array}{c}
        f_2,
        \\
        h_3,
        \\
        \omega_4,
        \\
        \omega_7
      \end{array}}
      \!\!\!
    \right]
    \!\Big/\!
    \left(
    {\begin{aligned}
      d\,f_2 & = 0
      \\
      d\,h_3 & = \omega_4
      -
      f_2 \wedge f_2
      \\
      d\,\omega_4 & = 0
      \\
      d\,\omega_7
      & = - \omega_4 \wedge \omega_4
    \end{aligned}}
    \right)
    \ar[uurr]|-{
      \mathclap{
      \rotatebox[origin=c]{-55}{
      \scalebox{.7}{$
        \;\,
        \arraycolsep=2.2pt
        \begin{array}{cccc}
          0 & 0 & 0 & \omega_{\,7}
          \\
          \mapsup & \mapsup & \mapsup & \mapsup
          \\
          f_2 & h_3 & \omega_{\, 4} & \omega_{\, 7}
        \end{array}
      $}
      }
      }
    }
    \ar@{^{(}->}[uu]|-{
      \scalebox{.7}{$
        \arraycolsep=3pt
        \begin{array}{cccc}
          f_2 & h_3 & \omega_{\, 4} & \omega_{\, 7}
          \\
          \mapsup & \mapsup & \mapsup & \mapsup
          \\
          f_2 & h_3 & \omega_{\, 4} & \omega_{\, 7}
        \end{array}
      $}
    }
    \ar[rr]^-{ \simeq }
    &&
    \mathbb{R}
    \!\!
    \left[
      \!\!\!
      {\begin{array}{c}
        f_2,
        \\
        \omega_7
      \end{array}}
      \!\!\!
    \right]
    \!\Big/\!
    \left(
    {\begin{aligned}
      d\,f_2 & = 0
      \\
      d\,\omega_7
      & = - (f_2)^4
    \end{aligned}}
    \right)
    \\
    \\
    &
    S^4
    &
    \;\;\;\;
        \mathbb{R}
    \!\!
    \left[
      \!\!\!
      {\begin{array}{c}
        \omega_4,
        \\
        \omega_7
      \end{array}}
      \!\!\!
    \right]
    \!\Big/\!
    \left(
    {\begin{aligned}
      d\,\omega_4
      & = 0
      \\
      d\,\omega_7
      & = - \omega_4 \wedge \omega_4
    \end{aligned}}
    \right)
    \ar@{^{(}->}[uu]|-{
      \scalebox{.7}{$
        \arraycolsep=3pt
        \begin{array}{cc}
          \omega_{\, 4} & \omega_{\, 7}
          \\
          \mapsup & \mapsup
          \\
          \omega_{\, 4} & \omega_{\, 7}
        \end{array}
      $}
    }
    \ar[uurr]|-{
      \mathclap{
      \rotatebox[origin=c]{-60}{
      \scalebox{.7}{$
        \!\!\!\!\!\!\!\!\!\!\!\!\!\!\!\!\!\!
        \arraycolsep=2.2pt
        \begin{array}{cccc}
          -f_2 \!\wedge\! f_2 & \omega_{\, 7}
          \\
          \mapsup & \mapsup
          \\
          \omega_{\, 4} & \omega_{\, 7}
        \end{array}
      $}
      }
      }
    }
  }
  }
\end{equation}

\vspace{-1mm}
\noindent where the generators $\omega_4, \omega_7, f_2, h_3$ are
all normalized
according to \eqref{NormalizationOfGeneratorsInSullivanModelForSpheres},
in particular:
\vspace{-1mm}
\begin{equation}
  \label{NormalizationOfRationalGeneratorsOnPlainTwistorFibration}
  \big\langle
    \omega_4, S^4
  \big\rangle
  \;=\;
  1
  \phantom{AAA}
  \big\langle
    f_2, S^2
  \big\rangle
  \;=\;
  1
  .
\end{equation}
\end{lemma}

\noindent Note that on the right in \eqref{SullivanModelForHopfTwistorFibrations}
we are showing the minimal Sullivan models of
$S^7$ and of $\mathbb{C}P^3$ by themselves
(which is classical, e.g. \cite[p. 142, 203]{FHT00}\cite[1.2, 5.3]{Menichi13}),
while on the left we are showing their Sullivan models as
fiber spaces, i.e., the relative minimal Sullivan models.

\vspace{-1mm}

\begin{proof}
  {\bf (i)} It is classical that the Sullivan model for $S^4$ is as shown
  (it is also a special case of Lemma \ref{SullivanModelOfSphericalFibration}).

\noindent  {\bf (ii)} Since $\mathbb{C}P^3 \to S^4$ is an $S^2$-fibration
  \eqref{HopfTwistorFactorization},
  Lemma \ref{SullivanModelOfSphericalFibration}
  implies from {(i)} that
  $\mathbb{C}P^3$ fibered over $S^4$ is modeled by

   \vspace{-2mm}
  $$
    \mathrm{CE}
    \big(
      \mathfrak{l}
      \mathbb{C}P^3
    \big)_{S^4}
    \;=\;
    \mathbb{R}\big[
      \omega_4, \omega_7,
      {\color{darkblue}f_2},
      {\color{darkblue}h_3}
    \big]
    \Big/
    \left(
    \begin{aligned}
      d\,\omega_4 & = 0
      \\
      d\, \omega_7 & = - \omega_4 \wedge \omega_4
      \\
      d\,f_2 & = 0
      \\
      d\,h_3
      & =
      f_2 \wedge f_2
      +
      {\color{darkblue} \alpha_4 }
    \end{aligned}
    \right)
  $$

   \vspace{-1mm}
\noindent
  for \emph{some} closed element
  $
    \alpha_4 \;\in\;
    \mathrm{CE}
    \big(
      \mathfrak{l}
      S^4
    \big)
    \,.
  $
  But in the present case, due to {(i)}, there is a unique
  such element, up to a real factor, namely $\omega_4$.
  Below in \eqref{RelationsOfSullivanGeneratorsOnCP3ModSp2}
  we find this factor to be unity.
  This yields the middle part of \eqref{SullivanModelForHopfTwistorFibrations}.

\noindent  {\bf (iii)} Since $S^7 \to \mathbb{C}P^3$ is an $S^1$-fibration
  \eqref{HopfTwistorFactorization},
  Lemma \ref{SullivanModelOfSphericalFibration} implies, via (ii),
  that $S^7$ fibered over $\mathbb{C}P^3$ is modeled by
  \vspace{-1mm}
  $$
    \mathrm{CE}
    \big(
      \mathfrak{l}
      S^7
    \big)_{\mathbb{C}P^3}
    \;=\;
    \mathbb{R}\big[
      \omega_4, \omega_7, f_2, h_3,
      {\color{darkblue} h_1}
    \big]
    \Big/
    \left(
    \begin{aligned}
      d\,\omega_4 & = 0
      \\
      d\, \omega_7 & = - \omega_4 \wedge \omega_4
      \\
      d\,f_2 & = 0
      \\
      d\,h_3
      & =
      f_2 \wedge f_2
      +
      \omega_4
      \\
      d\,f_1 & =
      {\color{darkblue} \alpha_2}
    \end{aligned}
    \right)
  $$

  \vspace{-2mm}
\noindent  for \emph{some} closed degree-2 element
  $\alpha_2 \in \mathrm{CE}\big( \mathfrak{l} \mathbb{C}P^3 \big)_{S^2}$.
  But in the present case, due to (ii), there is a unique such element,
  up to a real factor, namely $f_2$.
  Thus, by suitably rescaling $f_1$,
  we obtain $\alpha_2 = f_2$ and the claim follows.
\hfill \end{proof}

\begin{theorem}[Normalized Sullivan model of
Borel-equivariant Hopf/twistor fibrations]
  \label{SullivanModelOfParametrizedHopfTwistorFibrations}
The iterative relative Sullivan models for the
parametrized Hopf/twistor fibrations \eqref{ParametrizedHopfAndTwistorFibration}
are as follows
(here $\tfrac{1}{2}p_1, \rchi_8 \in \mathrm{CE}\big( \mathfrak{l} B \mathrm{Sp}(2) \big)$,
  via Lemma \ref{SullivanModelForBSp2}):
  \vspace{-.7cm}
\begin{equation}
  \label{SullivanModelForParametrizedHopfTwistorFibrations}
  \hspace{-5mm}
  \raisebox{135pt}{
  \xymatrix@C=17pt{
    &
    S^7 \!\sslash\! \mathrm{Sp}(2)
    \ar[dd]|>>>>>>>>>>>>{
      h_{\mathbb{C}}
      \sslash
      \mathrm{Sp}(2)
    }
    \ar[dddl]
    &&
    \mathrm{CE}
    \big(
      \mathfrak{l}
      B \mathrm{Sp}(2)
    \big)
    \!\!\!
    \left[
      \!\!\!
      {\begin{array}{c}
        h_1,
        \\
        f_2,
        \\
        h_3,
        \\
        \omega_4,
        \\
        \omega_7
      \end{array}}
      \!\!\!
    \right]
    \!\Big/\!
    \left(
    {\begin{aligned}
      d\,h_1 & = f_2
      \\
      d\,f_2 & = 0
      \\
      d\,h_3 & = \omega_4
      -
      \tfrac{1}{4}p_1
      -
      f_2 \wedge f_2
      \\
      d\,\omega_4 & = 0
      \\
      d\,\omega_7
      & = - \omega_4 \wedge \omega_4
      +
      \big(\tfrac{1}{4}p_1\big)^2
      \\
      & \phantom{=}\,
      - \rchi_8
    \end{aligned}}
    \right)
    \\
    \\
    &
    \mathbb{C}P^3 \!\sslash\! \mathrm{Sp}(2)
    \ar[dd]|-{
      t_{\mathbb{H}}
      \sslash
      \mathrm{Sp}(2)
    }
    \ar[dl]
    &&
    {\phantom{AA}}
    \mathrm{CE}
    \big(
      \mathfrak{l}
      B \mathrm{Sp}(2)
    \big)
    \!\!\!
    \left[
      \!\!\!
      {\begin{array}{c}
        f_2,
        \\
        h_3,
        \\
        \omega_4,
        \\
        \omega_7
      \end{array}}
      \!\!\!
    \right]
    \!\Big/\!
    \left(
    {\begin{aligned}
      d\,f_2 & = 0
      \\
      d\,h_3
      & = \omega_4
      -
      \tfrac{1}{4}p_1
      -
      f_2 \wedge f_2
      \\
      d\,\omega_4 & = 0
      \\
      d\,\omega_7
      & = - \omega_4 \wedge \omega_4
      +
      \big(\tfrac{1}{4}p_1\big)^2
      \\
      & \phantom{=}\;\,
      -
      \rchi_8
    \end{aligned}}
    \right)
    \ar@{^{(}->}[uu]|-{
      \scalebox{.7}{$
        \arraycolsep=3pt
        \begin{array}{cccc}
          \omega_{\, 4} & \omega_{\, 7} & f_2 & h_3
          \\
          \mapsup & \mapsup & \mapsup & \mapsup
          \\
          \omega_{\, 4} & \omega_{\, 7} & f_2 & h_3
        \end{array}
      $}
    }
    \\
    B \mathrm{Sp}(2)
    &&
    \mathrm{CE}
    \big(
      \mathfrak{l}
      B \mathrm{Sp}(2)
    \big)
    \ar@{^{(}->}[ur]
    \ar@{^{(}->}[dr]
    \ar@{^{(}->}[uuur]
    \\
    &
    S^4 \!\sslash\! \mathrm{Sp}(2)
    \ar[ul]
    &
    \;\;\;\;
    &
    \mathrm{CE}
    \big(
      \mathfrak{l}
      B \mathrm{Sp}(2)
    \big)
    \!\!\!
    \left[
      \!\!\!
      {\begin{array}{c}
        \omega_4,
        \\
        \omega_7
      \end{array}}
      \!\!\!
    \right]
    \!\Big/\!
    \left(
    {\begin{aligned}
      d\,\omega_4
      & = 0
      \\
      d\,\omega_7
      & = - \omega_4 \wedge \omega_4
      +
      \big(\tfrac{1}{4}p_1\big)^2
      \\
      & \phantom{=}\;\,
      -
      \rchi_8
    \end{aligned}}
    \right)
    \,,
    \ar@{^{(}->}[uu]|-{
      \scalebox{.7}{$
        \arraycolsep=3pt
        \begin{array}{cc}
          \omega_{\, 4} & \omega_{\, 7}
          \\
          \mapsup & \mapsup
          \\
          \omega_{\, 4} & \omega_{\, 7}
        \end{array}
      $}
    }
  }
  }
\end{equation}

\vspace{-2mm}
\noindent
where   the generators $f_2$ and $\omega_4$
represent the classes $c_1^R$ and $\tfrac{1}{2}\rchi_4$ in  \eqref{TheComponents}, respectively:
\vspace{-1mm}
\begin{equation}
  \label{IdentiftingRationlGeneratorsOnSP3ModSp2}
  [\omega_4] \;=\; \tfrac{1}{2}\rchi_4
  \;\in\;
  H^4
  \big(
    \mathbb{C}P^3 \!\sslash\! \mathrm{Sp}(2)
    ;
    \mathbb{R}
  \big)
  \,,
  \phantom{AAAAA}
  [f_2] \;=\; c_1^R
  \;\in\;
  H^2
  \big(
    \mathbb{C}P^3 \!\sslash\! \mathrm{Sp}(2)
;       \mathbb{R}
  \big).
\end{equation}
\end{theorem}
\begin{proof}
  That the composite vertical morphism,
  upon discarding the generators $h_1, f_2$ and $h_3$,
  is the minimal relative Sullivan model for $h_{\mathbb{H}}\sslash \mathrm{Sp}(2)$
  with the identification
  $[\omega_4] = \tfrac{1}{2}\rchi_4$ \eqref{IdentiftingRationlGeneratorsOnSP3ModSp2}
  is the result of \cite[3.19]{FSS19b}.

  Its factorization through
  $\mathbb{C}P^3\!\sslash\!\mathrm{Sp}(2)$
  must have
  minimal Sullivan model given by
  adjoining generators $f_2$ and $f_3$,
  by Lemma \ref{ParametrizedHopfTwistorFibrationsAreSpherical}
  with Lemma \ref{SullivanModelOfSphericalFibration}.
  The fiberwise normalization
  \eqref{NormalizationOfRationalGeneratorsOnPlainTwistorFibration}
  implies the identification $[f_2] = c_1^R$ in
  \eqref{IdentiftingRationlGeneratorsOnSP3ModSp2}, using that $c_1^R$
  pulled back along
  $S^2 = \mathrm{SU}(2)/U(1)
 \to B \mathrm{U}(1)_R$ is the unit volume generator.

  For the factorization of $h_{\mathbb{H}}\sslash \mathrm{Sp}(2)$
  through $\mathbb{C}P^3 \!\sslash\! \mathrm{Sp}(2)$
  to reproduce on fibers over $B \mathrm{Sp}(2)$
  (hence upon discarding the generators
  $\tfrac{1}{2}p_1$ and $\rchi_8$)
  the minimal Sullivan model for the
  plain Hopf/twistor fibrations from Lemma \ref{SullivanModelsForHopfTwistorFibrations}
  \emph{at least} those monomials in $f_2$ shown in
  \eqref{SullivanModelForParametrizedHopfTwistorFibrations} have to appear.
  We just have to observe that the relative coefficients in
  the differential relations for $h_3$ are as shown.
  But under the identification \eqref{IdentiftingRationlGeneratorsOnSP3ModSp2}
  we have the second logical equivalence shown here:
    \vspace{-2mm}
  \begin{equation}
    \label{RelationsOfSullivanGeneratorsOnCP3ModSp2}
    d h_3
    =
    \omega_4 - \tfrac{1}{4}p_1 - f_2 \wedge f_2
    \phantom{A}
    \Leftrightarrow
    \phantom{A}
    \big[
      \omega_4 - \tfrac{1}{4}p_1
    \big]
    =
    f_2 \wedge f_2
    \phantom{A}
    \Leftrightarrow
    \phantom{A}
    \widetilde \Gamma_4 - \widetilde \Gamma_4^{\mathrm{vac}}
    =
    c_1^R \cup c_1^R
    \phantom{A}
    \in
    H^4
    \big(
      \mathbb{C}P^3 \!\sslash\! \mathrm{Sp}(2);
      \,
      \mathbb{R}
    \big)
    \,.
  \end{equation}

  \vspace{-1mm}

  \noindent
  That the relation on the right of \eqref{RelationsOfSullivanGeneratorsOnCP3ModSp2}
  does hold
  follows immediately from
  \eqref{PullbackInIntegralCohomologyAlongEquivariantTwistorFibration}
  in Theorem \ref{CohomologyOfBorelEquivariantTwistorFibration}.

  Hence to conclude it suffices now to show that no further monomials in $f_2$
  appear on the right of \eqref{SullivanModelForParametrizedHopfTwistorFibrations}:

  First, any further monomial in $f_2$ that does appear must contain as a factor a
  \emph{basic} generator, namely a generator from
  $\mathrm{CE}(\mathfrak{l} B \mathrm{Sp}(2))$, to guarantee that it
  vanishes on fibers (where we already have the right terms).
  Since, by Lemma \ref{SullivanModelForBSp2}, the generators of
  $\mathrm{CE}\big( \mathfrak{l} B \mathrm{Sp}(2) \big)$ are in degrees 4 and 8,
  the only further term that could possibly appear, by degree reasons, is the
  blue term in the following expression:
  \vspace{-2mm}
  \begin{equation}
    \label{SpuriousTermInParametrizedHopfTwistorSullivan}
    d\,\omega_7
    \;=\;
    - \omega_4 \wedge \omega_4
    + \big(\tfrac{1}{2}p_1\big)^2
    - \rchi_8
    +
    {\color{darkblue}
      a
      \cdot
      f_2 \wedge f_2 \wedge p_1
    }
  \end{equation}

  \vspace{-2mm}
\noindent  for some coefficient $a \in \mathbb{R}$.
  But we also know that $t_{\mathbb{H}}\sslash \mathrm{Sp}(2)$
  is an $S^2$-fibration (by Lemma \ref{ParametrizedHopfTwistorFibrationsAreSpherical}),
  so that Lemma \ref{SullivanModelOfSphericalFibration}
  rules out the appearance of the blue term in \eqref{SpuriousTermInParametrizedHopfTwistorSullivan}
  (i.e., implies $a = 0$).
\hfill \end{proof}

\section{Charge quantization in Twistorial Cohomotopy}
\label{GreenSchwarzMechanismFromJTwistedCohomotopy}

After recalling (in \cref{NonAbelianDoldChernCharacter})
general non-abelian cohomology and
highlighting the non-abelian Chern-Dold character map,
we introduce (in \cref{TwistorialCohomotopyTheory}) the
twisted non-abelian cohomology to be called \emph{Twistorial Cohomotopy}
and use the results from \cref{BorelEquivariantHopfTwistorFibration}
to show (Corollary \ref{CohomologicalRelationInTwistorialCohomotopy})
that charge quantization in Twistorial Cohomotopy
implies the heteoric shifted flux quantization condition \eqref{HeteroticFluxQuantizationFromTwistorialCohomotopy}.

\medskip

\subsection{Non-abelian character map}
\label{NonAbelianDoldChernCharacter}

We recall the Chern-Dold character \eqref{ChernDoldCharacter} in generalized
cohomology and then introduce its generalization,
to a non-abelian character map \eqref{ChernDoldCharacterNonabelian}
on non-abelian cohomology.
The full technical detail is laid out in \cite{FSS20c}.

\medskip

\noindent {\bf From generalized to non-abelian cohomology.}
It is well-known, though perhaps under-appreciated,
that cohomology theory is all about
homotopy groups of mapping spaces into a
given ``coefficient space'' or ``classifying space''.
We recall this briefly for ``bare'' cohomology theories,
with the domain spaces $X$ assumed to a sufficiently nice topological space;
but the statement remains true for structured cohomology theories
such as differential and/or equivariant cohomology,
when interpreted internal to suitable higher toposes, see \cite[p. 6]{SS20b}.

\medskip

For ordinary (e.g., singular) cohomology with coefficients in an
\emph{abelian} discrete group $A$, these classifying spaces are the
Eilenberg-MacLane spaces $K(A,n)$ (e.g. \cite[\S 7.1, Cor. 12.1.20]{AGP02}):

\vspace{-.3cm}
\begin{equation}
  \label{OrdinaryCohomologyByEilenbergMacLane}
  \overset{
    \mathclap{
    \raisebox{3pt}{
      \tiny
      \color{darkblue}
      \bf
      \begin{tabular}{c}
        ordinary
        \\
        cohomology
      \end{tabular}
      }
    }
  }{
    H^n\big(X; A)
  }
  \;\simeq\;
  \overset{
    \mathclap{
    \raisebox{3pt}{
      \tiny
      \color{darkblue}
      \bf
      \begin{tabular}{c}
        homotopy classes of maps to
        \\
        Eilenberg-MacLane space
      \end{tabular}
    }
    }
  }{
  \pi_{0}
  \,
  \mathrm{Maps}
  \big(
    X,
    \,
    K(A,n)
  \big)
  }
  \,.
\end{equation}
These happen to be based loop spaces of each other,
$K(A,n) \;\simeq_{{}_{\mathrm{whe}}}\; \Omega K(A,n+1)$
(e.g. \cite[7.1.1]{AGP02}), so that
each of them is an \emph{infinite loop space} (e.g. \cite{Adams78}).

\medskip
 More generally, consider a \emph{generalized cohomology theory}\footnote{
  The term is widely used but somewhat unfortunate, since various
  \emph{further} generalizations of Whitehead's
  generalization of ordinary cohomology theories
  are relevant, such as twisted-, sheaf-, differential-,
  equivariant- and nonabelian-cohomology theories,
  as well as all their joint combinations.
}
$E^\bullet$ in the sense of \cite{Whitehead62} (see \cite{Adams74}\cite{Adams78}),
such as K-theory,
elliptic cohomology, $\mathrm{tmf}$, stable Cobordism, stable Cohomotopy, etc.
These are classified by such sequences of (pointed) spaces
which are successively
equipped with weak homotopy equivalences
exhibiting them as based loop spaces of
each other, called a \emph{spectrum} of spaces:
\vspace{-2mm}
\begin{equation}
  \label{ASpectrum}
  \{E_n\}_{n \in \mathbb{N}}
  \,,
  \;\;\;
  \mbox{s.t.}\,
  \phantom{A}
  E_n \;\simeq\; \Omega E_{n+1}\;,
\end{equation}
in that
\vspace{-1mm}
\begin{equation}
  \label{GeneralizedCohomologyByMappingIntoSpectrum}
  \overset{
    \mathclap{
    \raisebox{3pt}{
      \tiny
      \color{darkblue}
      \bf
      \begin{tabular}{c}
        generalized
        \\
        cohomology
      \end{tabular}
    }
    }
  }{
  E^n
  (
    X
  )
  }
  \;\simeq\;
  \overset{
    \mathclap{
    \raisebox{3pt}{
      \tiny
      \color{darkblue}
      \bf
      \begin{tabular}{c}
        homotopy classes of maps to
        \\
        infinite loop space
      \end{tabular}
    }
    }
  }{
  \pi_0
  \,
  \mathrm{Maps}
  (
    X,\,
    E_{n}
  )
  }\,.
\end{equation}

\vspace{-1mm}
\noindent
This is the \emph{Brown representability theorem}, see e.g.
\cite[\S III.6]{Adams74}\cite[\S 3.4]{Kochman96}.
But the right hand side of \eqref{GeneralizedCohomologyByMappingIntoSpectrum}
makes sense for $E_n$ \emph{any} space,
not necessarily part of a spectrum \eqref{ASpectrum},
and not necessarily even being a loop space.
It is not the notion of cohomology itself, but
rather only some extra \emph{properties} enjoyed by these abelian
cohomology groups (such as existence of connecting homomorphisms)
which is what is reflected in the infinite loop space structure \eqref{ASpectrum}.

\medskip
Indeed, for $G$ a well-behaved topological group, \emph{not} necessarily
abelian
(such as $G = \mathrm{U}(1),\mathrm{SU}(n), \mathrm{Sp}(n),\cdots$)
the fundamental theorem of $G$-principal bundles
(\cite[\S 19.3]{Steenrod51}, review in \cite[\S 5]{Addington07})
says that
degree-1 \emph{non-abelian cohomology} with coefficients in $G$
is represented by the classifying space $B G$ of $G$:
\vspace{-2mm}
\begin{equation}
  \label{NonabelianCohomologyInFirstDegree}
  \overset{
    \raisebox{3pt}{
      \tiny
      \color{darkblue}
      \bf
      \begin{tabular}{c}
        non-abelian cohomology with
        \\
        coefficients in topological group
      \end{tabular}
    }
  }{
    H^1(X;\, G)
  }
  \;\simeq\;
  \overset{
    \mathclap{
    \raisebox{3pt}{
      \tiny
      \color{darkblue}
      \bf
      \begin{tabular}{c}
        homotopy classes of maps
        \\
        to classifying space of group
      \end{tabular}
    }
    }
  }{
  \pi_0
  \,
  \mathrm{Maps}
  (
    X,\,
    B G
  )
  }\;.
\end{equation}

\vspace{-2mm}
\noindent If $G = A$ is abelian and discrete, then
$B A \simeq K(A,1)$ and \eqref{NonabelianCohomologyInFirstDegree}
reduces to \eqref{OrdinaryCohomologyByEilenbergMacLane},
but not otherwise.
Moreover, the May recognition theorem implies that
\emph{any} connected space $A$ is weakly homotopy equivalent to
a classifying space
$B G$, namely for $G = \Omega A$ the based loop group of $A$
(which may be rectified, up to weak homotopy equivalence, to
an actual topological group). Thereby, the traditional
equivalence \eqref{NonabelianCohomologyInFirstDegree} is
re-interpreted as an elegant general notion of \emph{non-abelian cohomology}:
\vspace{-3mm}
\begin{equation}
  \label{NonAbelianCohomology}
  \overset{
    \mathclap{
    \raisebox{3pt}{
      \tiny
      \color{darkblue}
      \bf
      \begin{tabular}{c}
        non-abelian cohomology
        \\
        with coefficients in $A$
      \end{tabular}
    }
    }
  }{
  H
  \big(
    X;
    \,
    A
  \big)
  }
  \;\;\;\; :=\;
  \overset{
    \raisebox{3pt}{
    \tiny
    \color{darkblue}
    \bf
    \begin{tabular}{c}
      homotopy-classes of
      \\
      maps to $A$
    \end{tabular}
    }
  }{
  \pi_0
  \,
  \mathrm{Maps}
  \big(
    X,\, A
  \big)
  }
  \;\;\;
  =
  \;\;\;
  \left\{
  \!\!\!
  \xymatrix{
    X
    \ar@/^1.7pc/[rr]|-{\;c\;}^-{
      \tiny
      \color{darkblue}
      \bf
      \begin{tabular}{c}
    \bf    map/cocycle
      \end{tabular}
    }_-{\ }="s"
    \ar@/_1.8pc/[rr]|-{\; c' \;}_-{
      \tiny
      \color{darkblue}
      \bf
      \begin{tabular}{c}
   \bf     map/cocycle
      \end{tabular}
    }^-{\ }="t"
    &&
    A
    \ar@{=>}|-{
      \mbox{
        \tiny
        \color{orangeii}
        \bf
        \begin{tabular}{c}
          homotopy/
          \\
          coboundary
        \end{tabular}
      }
    }
      "s"; "t"
  }
  \!\!\!
  \right\}_{\!\!\!\!\!\!\big/\mathrm{homotopy}}
\end{equation}
\vspace{-2mm}

\noindent Non-abelian cohomology in this generality is discussed
in \cite{Toen02}\cite{Jardine09}\cite{RobertsStevenson12}\cite{NSS12a}\cite{NSS12b}\cite{SS20b}.
 For example, for $X = S^n$ an $n$-sphere, we have
$S^n \simeq B\big(\Omega S^n\big)$ and the corresponding
non-abelian cohomology theory \eqref{NonabelianCohomologyInFirstDegree}
is \emph{Cohomotopy theory}
\vspace{-1mm}
$$
  \pi^n
  (
    X
  )
  \;:=\;
  \pi_0
  \,
  \mathrm{Maps}
  \big(
    X,\, S^n
  \big)
  \;\simeq\;
  H^1
  \big(
    X;\, \Omega S^n
  \big).
$$

 This  perspective
on generalized/non-abelian cohomology via classifying spaces makes
many related concepts nicely transparent, for
example the notions of \emph{twisting in cohohomology}
and of \emph{generalized Chern characters}.

\medskip

\noindent {\bf Twisted non-abelian cohomology.}
A \emph{twist} of $A$-cohomology \eqref{NonabelianCohomologyInFirstDegree}
is what is classified by a twisted parametrization of $A$
over some base space $B$ \cite[\S 4]{NSS12a}\cite[\S 2.2]{SS20b}\cite[\S 2.2]{FSS20c}:
Instead of mapping into a fixed classifying spaces, a
\emph{twisted cocycle} maps into a varying classifying space that may twist and
turn as one moves in the domain space.
In other words, a \emph{twisting} $\tau$ of $A$-cohomology theory
on some $X$ is a bundle over $X$ with typical fiber $A$,
and a $\tau$-twisted cocycle is a \emph{section} of that bundle
\cite[\S 4]{NSS12a}\cite{ABGHR14}\cite[\S 2.2]{SS20a}:
\vspace{-2mm}
\begin{equation}
  \label{TwistedNonAbelianGeneralizedCohomology}
\hspace{-1cm}
  \begin{aligned}
  \mbox{
    \tiny
    \color{darkblue} \bf
    \begin{tabular}{c}
      $\tau$-twisted
      \\
      non-abelian generalized
      \\
      $A$-cohomology theory
    \end{tabular}
  }
  \;
  A^\tau(X)
  & \;:=\;
  \left\{\!\!\!\!\!\!\!
    \raisebox{22pt}{
    \xymatrix@C=3em@R=1.5em{
      &
      \overset{
        \mathclap{
        \mbox{
          \tiny
          \color{darkblue} \bf
          \begin{tabular}{c}
            $\phantom{a}$
            \\
            $A$-fiber bundle
          \end{tabular}
        }
        }
      }{
        P
      }
      \ar[d]^-p
      \ar[rr]
      &&
      \overset{
        \mathclap{
        \mbox{
          \tiny
          \color{darkblue} \bf
          \begin{tabular}{c}
            universal
            \\
            $A$-fiber bundle
          \end{tabular}
        }
        }
      }{
      A \!\!\sslash\! \mathrm{Aut}(A)
      }
      \ar[d]
      \\
      X
      \ar@/^1.04pc/@{-->}[ur]^{
        \mbox{
          \tiny
          \color{darkblue} \bf
          \begin{tabular}{c}
            continuous
            section
            \\
            =
            twisted cocycle
          \end{tabular}
        }
      }
      \ar@{=}[r]
      &
      X
      \ar[rr]^-{\tau}_-{
               \mathclap{
          \mbox{
            \tiny
            \color{darkblue} \bf
            \begin{tabular}{c}
              classifying map
                            for $P$
            \end{tabular}
          }
          }
        }
      &&
      B \mathrm{Aut}(A)
    }
    }
  \right\}_{\!\!\!\!
       \Big/
    \sim_{
      {}_{
        \frac{
          \mathrm{homotopy}
        }
        {
          B \mathrm{Aut}(A)
        }
      }
    }
  }
  \\
  & \;\;\simeq\;
  \left\{
    \raisebox{22pt}{
    \xymatrix@C=4em{
      X
      \ar[dr]_-{
        \mathllap{
          \mbox{
            \tiny
            \color{darkblue} \bf
            twist
          }
        }
        \;
        \tau
      }^>>>>{\ }="t"
      \ar@{-->}[rr]^-{
        \mbox{
          \tiny
          \color{darkblue} \bf
          continuous function
        }
      }_-{\ }="s"
      &&
      A \!\!\sslash\! \mathrm{Aut}(A)
      \ar[dl]
      \\
      &
      B \mathrm{Aut}(A)
      \ar@{=>} "s"; "t"
      \ar@<-16pt>@{}^-{
        \rotatebox[origin=c]{48}{
          \tiny
          \color{orangeii}
          \bf
          homotopy
        }
      } "s"; "t"
    }
    }
  \right\}_{
    \!\!\!\!
    \Big/
    \sim_{
      {}_{
        \frac{
          \mathrm{homotopy}
        }
        {
          B \mathrm{Aut}(A)
        }
      }
    }
  }
  \end{aligned}
  \end{equation}
Here the equivalent formulation shown in the second line
follows because $A$-fiber bundles are themselves classified by
nonabelian $\mathrm{Aut}(A)$-cohomology
(see \cite[4.11]{NSS12a}\cite[2.92]{SS20b}),
as shown on the right of the first line.

\medskip

With a general concept of twisted non-abelian cohomology theories in hand, we turn to
discussion of their character maps. At their core, these
come from the rationalization approximation on coefficient/classifying spaces:

\medskip

\noindent {\bf Rationalization.}
For $X$ a connected nilpotent space, we write
\vspace{-3mm}
\begin{equation}
  \label{Rationalization}
  \xymatrix{
    X
    \ar[rr]^-{ \eta^{\mathbb{R}}_X }_-{
      \mbox{
        \tiny
        \begin{tabular}{c}
          \color{darkblue}
          \bf
          rationalization
          \\
          (over the real numbers)
        \end{tabular}
      }
    }
    &&
    L_{\mathbb{R}}X
  }
\end{equation}

\vspace{-3mm}
\noindent for its rationalization (e.g. \cite[1.4]{Hess07}) over the real numbers.
And we write $\mathfrak{l}X$
for the \emph{Whitehead  $L_\infty$-algebra} that is the formal dual of the
minimal Sullivan model for
$X$ \cite{BFM06}\cite[\S 2.1]{BMSS19}\cite[Prop. 3.67]{FSS20c}
(or in their rectified incarnation \cite[\S 1.0.2]{FRS13}: dg-Lie algebras
as in the original \cite{Quillen69}).

\medskip

\noindent {\bf Chern-Dold character in abelian cohomology.}
Given an abelian generalized cohomology theory $E^\bullet$
\eqref{GeneralizedCohomologyByMappingIntoSpectrum},
rationalization \eqref{Rationalization}
of its classifying spaces \eqref{ASpectrum}
induces a cohomology operation from $E$-cohomology theory to
ordinary cohomology with coefficients in the rationalized
stable homotopy groups of $E$:
\vspace{-.3cm}
\begin{equation}
  \label{ChernDoldCharacter}
  \hspace{-.2cm}
  \overset{
    \mathclap{
    \raisebox{3pt}{
      \tiny
      \color{darkblue}
      \bf
      \begin{tabular}{c}
        Chern-Dold
        \\
        character
      \end{tabular}
    }
    }
  }{
    \mathrm{ch}_E
  }
  \;:
  \xymatrix@C=20pt{
    \underset{
      \mathclap{
      \raisebox{-3pt}{
        \tiny
        \color{darkblue}
        \bf
        \begin{tabular}{c}
          generalized
          \\
          cohomology
        \end{tabular}
      }
      }
    }{
    E^n
    (
      X
    )
    }
   \;\;
   \overset{
   \mathclap{
    \raisebox{+8pt}{
      \tiny
      \color{darkblue}
      \bf
      \begin{tabular}{c}
        Brown
        \\
        represen-
        \\
        tability
      \end{tabular}
    }
    }
   }{
     \simeq
   }
   \;\;
   \pi_0
   \,
    \mathrm{Maps}
   \big(
      X,\,
      E_n
  \big)
  \ar[rr]^-{
    \mbox{
      \tiny
      \color{darkblue}
      \bf
      \begin{tabular}{c}
        rationalization
        \\
        \phantom{A}
      \end{tabular}
    }
  }_-{
    \pi_0
    \,
    \mathrm{Maps}
    (
      X,\,
      L_{\mathbb{R}}
    )
  }
  &&
  \underset{
    \mathclap{
    \;\;\;\;\;\;\;\;\;\;\;\;
    \raisebox{-3pt}{
      \tiny
      \color{darkblue}
      \bf
      rational $E$-cohomology
    }
    }
  }{
  \pi_0
  \,
  \mathrm{Maps}
  \big(
    X,\,
    L_{\mathbb{R}}E_n
  \big)
  \;=:\;
  E_{\mathbb{R}}^n
  \big(
    X
  \big)
  }
  \overset{
    \mathclap{
    \raisebox{+7pt}{
      \tiny
      \color{darkblue}
      \bf
      \begin{tabular}{c}
        Dold's
        \\
        equiv-
        \\
        alence
      \end{tabular}
    }
    }
  }{
    \;\simeq\;
  }
  \underset{k}{\bigoplus}
  \;
  \underset{
    \mathclap{
    \raisebox{0pt}{
      \tiny
      \color{darkblue}
      \bf
      \begin{tabular}{c}
        ordinary cohomology
        \\
        with coefficients in
        \\
        rationalized homotopy groups of $E$
      \end{tabular}
    }
    }
  }{
    H^{n+k}
    \big(
      X;
      \,
      \pi_k(E)\!\otimes_{{}_{\mathbb{Z}}}\! \mathbb{R}
    \big)
    \,,
    }
  }
\end{equation}

\vspace{-2mm}
\noindent where Dold's equivalence on the far right is due to
\cite[Cor. 4]{Dold65}, reviewed in \cite[\S II.3.17]{Rudyak98}.
This map \eqref{ChernDoldCharacter}
is called the \emph{Chern-Dold character map}, due to
\cite{Buchstaber70}. The modern formulation
above is made fully explicit in \cite[\S 2.1]{LindSatiWesterland16};
see also \cite[\S 4.8]{HopkinsSinger05}\cite[p. 17]{BunkeNikolaus14}\cite{GS-AHSS}.
For example:
\begin{enumerate}[\bf (i)]
\vspace{-.2cm}
\item
When $E^\bullet = H^\bullet(-;\mathbb{Z})$ is ordinary
integral cohomology, its rationalization is
$E^\bullet_{\mathbb{R}} = H^\bullet(-; \mathbb{R})$ and the
Chern-Dold character \eqref{ChernDoldCharacter} reduces to
extension of scalars from integral to real cohomology,
as in \eqref{BareShiftedFluxQuantization}, \eqref{FluxQuantizationOnMO9}.
\vspace{-.2cm}
\item
When $E^\bullet = \mathrm{KU}^\bullet$
is complex topological K-theory, its rationalization is
$\mathrm{KU}^{0,1}_{\mathbb{R}} \simeq H^{\mathrm{even},\mathrm{odd}}
(-;\mathbb{R})$ and the Chern-Dold character \eqref{ChernDoldCharacter}
reduces to the ordinary Chern character
(see \cite{GS-KO}\cite{GS-RR} for extensive discussions).
\end{enumerate}

\noindent {\bf Character map in twisted non-abelian cohomology.}
\begin{defn}[Non-abelian de Rham cohomology {\cite[\S 6.5]{SatiSchreiberStasheff08}\cite[\S 4.1]{FiorenzaSchreiberStasheff10}\cite[Def. 3.82]{FSS20c}}]
 \label[Def. 3.66]{NonAbelianDeRhamCohomology}
The non-abelian de Rham cohomology of a smooth manifold $X$
with coefficients in an $L_\infty$-algebra $\mathfrak{g}$
of finite type
is the quotient of the set of dg-algebra homomorphism from
the Chevalley-Eilenberg algebra
$\mathrm{CE}(\mathfrak{g})$ of $\mathfrak{g}$ (which is the Sullivan
model of a rational space) to the de Rham dg-algebra
$\Omega^\bullet_{\mathrm{dR}}$ of differential
forms on $X$,
quotiented by dg-algebra homotopies:
\vspace{-2mm}
\begin{equation}
  \label{NonabelianDeRhamCohomology}
  \overset{
    \mathclap{
    \raisebox{3pt}{
      \tiny
      \color{darkblue}
      \bf
      \begin{tabular}{c}
        non-abelian
        \\
        de Rham cohomology
        \\
        with coefficients in $\mathfrak{l}A$
      \end{tabular}
    }
    }
  }{
  H_{\mathrm{dR}}
  \big(
    X;\,
    \mathfrak{l}A
  \big)
  }
  \;\;:=\;\;
  \underset{
  }{
    \mathrm{Hom}
    \big(
      \overset{
        \mathclap{
        \!\!\!
        \mbox{
          \tiny
          \color{darkblue}
          \bf
          \begin{tabular}{c}
            Sullivan model/
            \\
            Chevalley-Eilenberg
            \\
            dg-algebra
          \end{tabular}
        }
        }
      }{
        \mathrm{CE}
        (\mathfrak{l}A)
      }
      \;\;\,,\;\;
      \overset{
        \mathclap{
        \,\,\,
        \mbox{
          \tiny
          \color{darkblue}
          \bf
          \begin{tabular}{c}
            de-Rham
            \\
            dg-algebra
          \end{tabular}
        }
        }
      }{
        \Omega_{\mathrm{dR}}^\bullet(X)
      }
    \big)_{\big/\!\sim}
  }
  \;\;
  =
  \;\;
  \left\{
  \!\!\!
  \xymatrix{
    \Omega^\bullet_{\mathrm{dR}}
    (
      X
    )
    \;\;\;\;\;\;\;
    \ar@{<-}@/^1.7pc/[rr]|-{\;A\;}^-{
      \mathclap{
      \mbox{
        \tiny
        \color{darkblue}
        \bf
        \begin{tabular}{c}
          dg-algebra homomorphism/
          \\
          flat $\mathfrak{l}A$-valued differential form
        \end{tabular}
      }
      }
    }_-{\ }="s"
    \ar@{<-}@/_1.7pc/[rr]|-{\;A'\;}_-{
      \mathclap{
      \mbox{
        \tiny
        \color{darkblue}
        \bf
        \begin{tabular}{c}
          dg-algebra homomorphism/
          \\
          flat $\mathfrak{l}A$-valued differential form
        \end{tabular}
      }
      }
    }^-{\ }="t"
    &&
\;\;\;\;\;    \mathrm{CE}
    \big(
      \mathfrak{l}A
    \big)
    \ar@{=>}|-{
      \mbox{
        \tiny
        \color{orangeii}
        \bf
        \begin{tabular}{c}
          dga-homotopy/
          \\
          coboundary
        \end{tabular}
      }
    }
      "s"; "t"
  }
  \!\!\!
  \right\}_{\!\!\!\big/_{\mathrm{homotopy}}}
\end{equation}
\end{defn}
\begin{example}[Recovering ordinary de Rham cohomology {\cite[Prop. 3.94]{FSS20c}}]
  \label{RecoveringOrdinarydeRhamCohomology}
  In the case that $\mathfrak{g} = \mathbb{R}[n]$ is the
  \emph{line} $L_\infty$-algebra concentrated in degree $n$,
  its Chevalley-Eilenberg algebra is the free graded-commutative
  algebra on a single generator in degree $n+1$ with vanishing
  differential; which is also the Sullivan model of
  the Eilenberg-MacLane space \eqref{OrdinaryCohomologyByEilenbergMacLane}
  in that degree:
  \begin{equation}
    \mathrm{CE}
    \big(
      \mathbb{R}[n]
    \big)
    \;=\;
    \mathbb{R}\big[c_{n_1}\big]
    \!\big/\!
    \left(
    {\begin{aligned}
      d\,c_{n+1} & = 0
    \end{aligned}}
    \!\!\!
    \right)
    \;\simeq\;
    \mathrm{CE}
    \big(
      \mathfrak{l}
      K(n+1, \mathbb{Z})
    \big)
    \,.
  \end{equation}
  Hence dg-algebra homomorphisms out of this into a de Rham algebra
  are equivalently closed differential $(n+1)$-forms:
  \begin{equation}
    \mathrm{Hom}
    \big(
      \mathrm{CE}(\mathbb{R}[1])
      ,\,
      \Omega^\bullet_{\mathrm{dR}}(X)
    \big)
    \;\simeq\;
    \Omega^n(X)_{\mathrm{cl}}
    \,,
  \end{equation}
  and dg-algebra homotopies between these are equivalently de Rham
  coboundaries. Therefore, the non-abelian de Rham cohomology \eqref{NonabelianDeRhamCohomology}
  with these coefficients reduces to ordinary de Rham cohomology
  in that degree:
  \begin{equation}
    H_{\mathrm{dR}}
    \big(
      X;
      \,
      \mathbb{R}[n]
    \big)
    \;\simeq\;
    H^{n+1}_{\mathrm{dR}}(X)
    \,.
  \end{equation}
\end{example}
\begin{prop}[Non-abelian de Rham theorem {\cite[Thm. 3.95]{FSS20c}}]
  \label{NonAbelianDeRhamTheorem}
  Let $X$ be a smooth manifold and $A$ a nilpotent
  topological space of finite rational homotopy type, hence
  with a minimal Sullivan model $\mathrm{CE}\big(\mathfrak{l}A\big)$
  for its rationalization $\mathfrak{l}A$ \eqref{Rationalization}.
    Then the non-abelian cohomology \eqref{NonAbelianCohomology}
  of $X$ with real coefficients $\mathfrak{l}A$ is equivalent to the
  non-abelian de Rham cohomology \eqref{NonabelianDeRhamCohomology}
  with coefficient in $\mathfrak{l}A$:
  \vspace{-2mm}
  \begin{equation}
  \label{EquivalenceOfNonAbelianDeRhamTheorem}
  \overset{
    \mathclap{
    \raisebox{3pt}{
      \tiny
      \color{darkblue}
      \bf
      \begin{tabular}{c}
        non-abelian
        \\
        real cohomology
      \end{tabular}
    }
    }
  }{
  H
  (
    X;
    \,
    L_{\mathbb{R}} A
  )
  }
  \;\;\simeq\;\;
  \overset{
    \raisebox{3pt}{
      \tiny
      \color{darkblue}
      \bf
      \begin{tabular}{c}
        non-abelian
        \\
        de Rham cohomology
      \end{tabular}
    }
  }{
  H_{\mathrm{dR}}
  (
    X;
    \,
    \mathfrak{l}A
  )\;.
  }
\end{equation}
\end{prop}
\begin{proof}
  Unwinding the definitions, the equivalence \eqref{EquivalenceOfNonAbelianDeRhamTheorem}
  reduces to the fundamental theorem of
  rational homotopy theory \cite[\S 9.4]{BousfieldGugenheim76}
  (reviewed as \cite[Prop. 2.11]{BMSS19};
  see also \cite[Cor. 1.26]{Hess07}) which identifies
  the hom-sets in the homotopy categories of {\bf a)}
  nilpotent and finite-type rational topological spaces,
  and  {\bf b)} the opposite of dgc-algebras.
  %
  %
\hfill \end{proof}

\begin{prop}[Non-abelian de Rham theorem for stable coefficients {\cite[Ex. 3.75]{FSS19a}}]
  \label{NonabelianDeRhamForStableCoefficients}
  Let $X$ be a smooth manifold, and
  $E$ an infinite-loop space \eqref{ASpectrum}.
  Then non-abelian de Rham cohomology \eqref{NonabelianDeRhamCohomology}
  of $X$ with coefficients in $\mathfrak{l}E$ is equivalent to
  the real cohomology of $X$
  with coefficients in the rationalized homotopy groups of $E$:
  \vspace{-2mm}
  \begin{equation}
    \label{StableVersionOfNonabelianDeRham}
    H_{\mathrm{dR}}
    \big(
      X;
      \,
      \mathfrak{l}E
    \big)
    \;\;
    \simeq
    \;\;
    \underset{
      k
    }{\bigoplus}
    \,
    H^{k}
    \big(
      X;
      \,
      \pi_k(E)\otimes_{{}_{\mathbb{Z}}}\mathbb{R}
    \big)
    \,.
  \end{equation}
\end{prop}
\begin{proof}
The minimal Sullivan model of an infinite loop space is the free graded
algebra generated by its rationalized homotopy groups,
with vanishing differential
(see \cite[p. 143]{FHT00}, or, from a broader perspective of rational
spectra, \cite[Lemma 2.25, Prop. 2.30]{BMSS19}). This implies the claim
by Example \ref{RecoveringOrdinarydeRhamCohomology},
via the ordinary de Rham theorem (e.g. \cite[10.15]{FHT00}).
\hfill \end{proof}
In conclusion:

\begin{prop}[Non-abelian de Rham theorem recovers Dold's equivalence {\cite[Prop. 4.6]{FSS20c}}]
  \label{NonAbelianDeRhamTheoremRecoversDoldEquivalence}
  Let $X$ be a smooth manifold, and
  $E$ (the connective spectrum of) an infinite-loop space \eqref{ASpectrum}.
  Then Dold's equivalence is equivalent to the restriction of the
  non-abelian de Rham theorem (Prop. \ref{NonAbelianDeRhamTheorem})
  to stable coefficients (Prop. \ref{NonabelianDeRhamForStableCoefficients}):
   \vspace{-2mm}
  \begin{equation}
    \xymatrix@C=5em{
      E_{\mathbb{R}}^n(X)
      \ar[r]^-{
          \mbox{
            \tiny
            \color{darkblue}
            \bf
            \begin{tabular}{c}
              Dold's equivalence
            \end{tabular}
          }}_-{
          \simeq
      }
      \ar@{=}[d]
      &
      \underset{k}{\bigoplus}
      \,
      H^{n+k}
      \big(
        X;
        \,
        \pi_{n+k}(E)\otimes_{{}_{\mathbb{Z}}}\mathbb{R}
      \big)
      \ar@{<-}[d]^-{
        \rotatebox[origin=c]{-90}{
          \scalebox{.6}{$\simeq$}
        }
        \mathrlap{
          \;
          \scalebox{.7}{
            \rm
            \eqref{StableVersionOfNonabelianDeRham}
          }
        }
      }
      \\
      H(X;\,L_{\mathbb{R}} E_n)
      \ar[r]_-{
                \mbox{
            \tiny
            \begin{tabular}{c}
              \color{darkblue}
              \bf
              non-abelian
              \\
              \color{darkblue}
              \bf
              de Rham theorem
              \\
              \rm
              Prop. \ref{NonAbelianDeRhamTheorem}
            \end{tabular}
          }
          }^-{
          \simeq
        }
      &
      H_{\mathrm{dR}}
      \big(
        X;
        \,
        \mathfrak{l}E_n
      \big)
    }
  \end{equation}
\end{prop}

Therefore, we obtain the following generalization of the
Chern-Dold character \eqref{ChernDoldCharacter}:

\begin{defn}[Character map in non-abelian cohomology {\cite[Def. 4.3]{FSS20c}}]
 \label{ChernDoldCharacterInNonabelianCohomology}
Let $X$ be a smooth manifold and $A$ a nilpotent space
of finite rational type.
Then the \emph{non-abelian Chern-Dold character} on
non-abelian cohomology theory \eqref{NonAbelianCohomology}
represented by $A$
is the composite of
\\
{\bf (a)} the rationalization map \eqref{Rationalization}
on coefficients
\\
{\bf (b)} the non-abelian de Rham theorem \ref{NonAbelianDeRhamTheorem}:

\vspace{-8mm}
\begin{equation}
  \label{ChernDoldCharacterNonabelian}
\phantom{AA}
\overset{
    \mathclap{
    \raisebox{3pt}{
      \tiny
      \color{darkblue}
      \bf
      \begin{tabular}{c}
        non-abelian
        \\
        Chern-Dold
        character
      \end{tabular}
    }
    }
  }{
    \mathrm{ch}_A
  }
  \;\; \;\;\;\;\;\;  :\;\;
  \xymatrix@C=25pt{
    \underset{
      \mathclap{
      \raisebox{0pt}{
        \tiny
        \color{darkblue}
        \bf
        \begin{tabular}{c}
          non-abelian
          \\
          cohomology
          \\
          with coefficients in $A$
        \end{tabular}
      }
      }
    }{
      H(X;A)
    }
   \;
   :=
   \;
   \pi_0
   \,
    \mathrm{Maps}
   (
      X,\,
      A
  )
  \ar[rr]^-{
    \mbox{
      \tiny
      \color{darkblue}
      \bf
      \begin{tabular}{c}
        rationalization
        \\
        \phantom{A}
      \end{tabular}
    }
  }_-{
    \pi_0
    \,
    \mathrm{Maps}
    (
      X,\,
      \eta^{\mathbb{R}}_A
    )
  }
  &&
  \underset{
    \mathclap{
    \;\;\;\;\;\;\;
    \raisebox{0pt}{
      \tiny
      \color{darkblue}
      \bf
      \begin{tabular}{c}
        non-abelian cohomology
        \\
        with coefficients in $\mathfrak{l}A$
      \end{tabular}
    }
    }
  }{
  \pi_0
  \,
  \mathrm{Maps}
  (
    X,\,
    L_{\mathbb{R}}A
  )
  \;=:\;
  H
  (
    X;\,
    L_{\mathbb{R}}A
  )
  }
  \overset{
    \mathclap{
    \raisebox{+8pt}{
      \tiny
      \color{darkblue}
      \bf
      \begin{tabular}{c}
        non-abelian
        \\
        de Rham
        \\
        theorem
      \end{tabular}
    }
    }
  }{
    \;\;\simeq\;\;\;\;
  }
  \underset{
    \mathclap{
    \raisebox{-3pt}{
      \tiny
      \color{darkblue}
      \bf
      \begin{tabular}{c}
        non-abelian
        \\
        de Rham cohomology
        \\
        with coefficients in $\mathfrak{l}A$
      \end{tabular}
    }
    }
  }{
  H_{\mathrm{dR}}
  (
    X;\,
    \mathfrak{l}A
  )\;.
  }
  }
\end{equation}
\end{defn}

\noindent {\bf Character map in twisted non-abelian cohomology.}
The above constructions immediately generalize to
twisted
nonabelian cohomology \eqref{TwistedNonAbelianGeneralizedCohomology}
to yield the
twisted non-abelian Chern character cohomology operation:

\begin{defn}[Character in twisted non-abelian cohomology {\cite[Def. 5.4]{FSS20c}}]
The twisted non-abelian character map is the
non-abelian character (Def. \ref{ChernDoldCharacterInNonabelianCohomology})
in the slice over $B \mathrm{Aut}(A)$:
\vspace{-3mm}
\begin{equation}
  \label{ChernCharacterCohomologyOperation}
  \phantom{AA}
 \overset{
    \mathclap{
    \raisebox{3pt}{
      \tiny
      \color{greenii}
      \bf
      \begin{tabular}{c}
        twisted non-abelian
        \\
        Chern-Dold character
      \end{tabular}
    }
    }
  }{
    \mathrm{ch}^\tau_A
  }
  \;\;\; : \;
  \xymatrix@C=10pt{
    \underset{
      \mathclap{
      \raisebox{0pt}{
        \tiny
        \color{darkblue}
        \bf
        \begin{tabular}{c}
          $\tau$-twisted non-abelian
          \\
          cohomology
          \\
          with coefficients in $A$
        \end{tabular}
      }
      }
    }{
      H^\tau(X;A)
    }
      :=
   \pi_0
   \,
    \mathrm{Maps}_{{}_{B \mathrm{Aut}(A)}}
    (
      X,
      \,
      A
    )
  \ar[rr]^-{
    \mbox{
      \tiny
      \color{greenii}
      \bf
      \begin{tabular}{c}
        rationalization
        \\
        \phantom{A}
      \end{tabular}
    }
  }
  &&
  \underset{
    \mathclap{
    \;\;\;\;\;\;\;
    \raisebox{0pt}{
      \tiny
      \color{darkblue}
      \bf
      \begin{tabular}{c}
        $L_{\mathbb{R}}\tau$-twisted non-abelian cohomology
        \\
        with coefficients in $L_{\mathbb{R}}A$
      \end{tabular}
    }
    }
  }{
  \pi_0
  \,
  \mathrm{Maps}_{{}_{/L_{\mathbb{R}}B \mathrm{Aut}(A)}}
  (
    X,\,
    L_{\mathbb{R}}A
  )
  =:
  H^{L_{\mathbb{R}}\tau}
  \big(
    X;\,
    L_{\mathbb{R}}A
  \big)
  }
  \overset{
    \mathclap{
    \raisebox{+8pt}{
      \tiny
      \color{greenii}
      \bf
      \begin{tabular}{c}
        twisted non-abelian
        \\
        de Rham
        \\
        theorem
      \end{tabular}
    }
    }
  }{
    \;\;\simeq\;\;\;\;
  }
  \underset{
    \mathclap{
    \raisebox{-3pt}{
      \tiny
      \color{darkblue}
      \bf
      \begin{tabular}{c}
        $\mathfrak{l}\tau$-twisted non-abelian
        \\
        de Rham cohomology
        \\
        with coefficients in $\mathfrak{l}A$
      \end{tabular}
    }
    }
  }{
  H^{\tau_{\mathrm{dR}}}_{\mathrm{dR}}
  (
    X;\,
    \mathfrak{l}A
  )
  }
  }
\end{equation}
 \phantom{AAAA}
  $
  \left\{\!\!
  \raisebox{24pt}{
  \xymatrix@C=2pt{
    X
    \ar[dr]_-{\tau}
    \ar[rr]^-{c}
    &&
    A \sslash \mathrm{Aut}(A)
    \ar[dl]
    \\
    & B \mathrm{Aut}(A)
  }
  }
  \!\right\}_{\!\!\!\!\!\!\Big/\sim}
  $
  \phantom{AA}
  \xymatrix{
    \ar@{|->}[r]
    &
  }
  \phantom{AA}
  $
  \left\{\!\!
  \raisebox{28pt}{
  \xymatrix@C=-14pt{
    \Omega^\bullet_{\mathrm{dR}}(X)
    \ar@{<-}[dr]_-{\tau^\ast}
    \ar@{<-}[rr]^-{ A }
    &&
    \mathrm{CE}
    \Big(
      \mathfrak{l}_{{}_{\mathrm{B \mathrm{Aut}(A)}}}
      \big(A\sslash \mathrm{Aut}(A)\big)
    \Big)
    \\
    &
    \mathrm{CE}
    \Big(
      \mathfrak{l}
      \big(B \mathrm{Aut}(A)\big)
    \Big)
    \ar@{_{(}->}[ur]
  }
  }
  \! \right\}_{\!\!\!\!\!\!\Big/\sim}
  $
\end{defn}

\noindent
This means that the twisted character on $A$-cohomology is the
plain character on $A \sslash \mathrm{Aut}(A)$-cohomology fibered over
$B \mathrm{Aut}(A)$, hence is the fiberwise $A$-character on an
$A$-fiber $\infty$-bundle.
(The notation $\mathfrak{l}_{{}_{B}}(-)$ in \eqref{ChernCharacterCohomologyOperation}
denotes the {\it relative} Whitehead $L_\infty$-algebra
over a base $B$ \cite[Prop. 3.80]{FSS20c}, such that
$\mathrm{CE}\big( \mathfrak{l}_{B} (-) \big)$ denotes the
Sullivant minimal model {\it relative} to that of the base $B$
(\cite[Prop. 3.49]{FSS20c}),
thus ensuring that the domain on the right is still cofibrant in the co-sliced model structure,
as in \cite[proof of Prop. 3.115]{FSS20c}.)

\begin{remark}[Charge quantization by lift through character map]
  \label{ChargeQuantizationByLiftThroughChernCharacter}
  Just as for the traditional Chern character
  on K-theory (see \cite{GS-KO} for a detailed account),
  the Chern-Dold character \eqref{ChernDoldCharacter}
  is generally far from being surjective, and the
  same is true for its non-abelian \eqref{ChernDoldCharacterNonabelian}
  and its twisted non-abelian generalization
  \eqref{ChernCharacterCohomologyOperation}.

\noindent {\bf (i)}  The obstruction to lifting
  de Rham form data through the
  Chern-Dold character maps are
  \emph{integrality} conditions that disappear upon rationalization,
  hence are ``quantization'' conditions (in the original sense of
  Bohr-Sommerfeld quantization).

\noindent {\bf (ii)}
  Therefore,
  if any given differential form data lifts through the
  Chern-Dold character
  of some twisted non-abelian $A$-cohomology theory, we say that that it is
  \emph{quantized in A-theory}.

\noindent {\bf (iii)}
  In typical examples the
  differential forms in question are flux densities, encoding charges
  of physical fields, and hence we speak of
  \emph{charge-quantization in A-theory}.
  (For abelian cohomology this is discussed in \cite{Freed00}\cite{GS-RR}.)
\end{remark}

\medskip

\subsection{Twistorial Cohomotopy theory}
  \label{TwistorialCohomotopyTheory}

We now identify and study the twisted non-abelian cohomology theory whose
classifying space is the Borel-equivariant twistor fibration (Def. \ref{ParametrizedHopfTwistorFibration}).
The main result of this section is Theorem \ref{ShiftedIntegralityUnderTTheoreticChernCharacter},
which shows that charge-quantization (Remark \ref{ChargeQuantizationByLiftThroughChernCharacter})
in this \emph{Twistorial Cohomotopy} (Prop. \ref{ChernCharacterInJTwistedCohomotopy})
imposes a shifted integrality condition
\eqref{HeteroticShiftedFluxQuantizationInTTheory}
on Chern-Dold character forms (Corollary \ref{ShiftedIntegralityUnderTTheoreticChernCharacter})
matching that of \eqref{HeteroticFluxQuantizationFromTwistorialCohomotopy}.

\medskip

\noindent {\bf Tangential $\mathrm{Sp}(2)$-structure.}
Consider smooth spin 8-manifolds $X$ that are equipped with
tangential $\mathrm{Sp}(2)$-structure (e.g. \cite[4.48]{SS20b}),
hence with a homotopy-lift\footnote{
  All diagrams in the following are filled with such
  homotopies, but for ease of presentation we mostly suppress them, notationally.
} of
the classifying map of their tangent bundle to
the classifying space $B \mathrm{Sp}(2)$  of the
quaternionic unitary group (Def. \ref{QuaternionicGroups})
along its canonical inclusion $i_{\mathrm{Sp}}$ \eqref{CentralProductTriality}:

\vspace{-2mm}
\begin{equation}
  \label{The8Manifold}
  \xymatrix{
    \overset{
      \mathclap{
      \raisebox{3pt}{
        \tiny
        \color{darkblue}
        \bf
        8-manifold
      }
      }
    }{
      X
    }
    \ar[dr]
    \ar[rr]^-{
      \overset{
        \mathclap{
        \raisebox{3pt}{
          \tiny
          \color{greenii}
          \bf
          \begin{tabular}{c}
            tangential
            \\
            $\mathrm{Sp}(2)$-structure
          \end{tabular}
        }
        }
      }{
        \tau
      }
    }_-{\ }="s"
    \ar[dr]_-{
      \mathllap{
        \mbox{
          \tiny
          \color{greenii}
          \bf
          \begin{tabular}{c}
            classifying map
            \\
            of tangent bundle
          \end{tabular}
        }
        \;\;
      }
      T X
    }^>>>>>>>{\ }="t"
    &&
    B \mathrm{Sp}(2)
    \ar[dl]^-{ B i_{\mathrm{Sp}} }
    \\
    & B \mathrm{Spin}(8)
    \ar@{=>} "s"; "t"
  }
\end{equation}

  \vspace{-2mm}
\noindent In the intended applications, this spin 8-manifold \eqref{The8Manifold}
is one factor in an 11-dimensional spacetime
of the form $\mathbb{R}^{2,1} \times X$ (see \cite[\S 3]{FSS19b}).
We write $\omega$ for any affine connection on $T X$ (``spin connection'')
and write
\begin{equation}
  p_i(\omega)
  \;\in\;
  H^{2i}_{\mathrm{dR}}(X)
  \;\simeq\;
  H^{2i}(X; \mathbb{R})\;.
\end{equation}
for the induced Pontrjagin forms (e.g. \cite[p. 10]{GS-KO}).

\medskip

\noindent {\bf Associated twistor-space fibration.}
By Prop. \ref{EquivarianceOfCombinedHopfTwistorFibration},
a tangential $\mathrm{Sp}(2)$-structure \eqref{The8Manifold} induces,
via pullback of the parametrized Hopf/twistor fibration
from Def. \ref{ParametrizedHopfTwistorFibration},
an $S^4$-fibration $E$ and a $\mathbb{C}P^3$-fibration $\widetilde E$
over $X$, connected by a morphism of fibrations over $X$
which is fiberwise the plain twistor fibration $t_{\mathbb{H}}$ \eqref{HopfTwistorFactorization}:
  \vspace{-2mm}
\begin{equation}
  \label{AssociatedTwistorSpaceFibration}
  \hspace{-3cm}
  \raisebox{20pt}{
  \xymatrix@C=4.2em{
    \mathbb{C}P^3
    \ar[dr]^{ \!\!\!\!\!\!
      t_{\mathbb{H}}
      \!\!\!\!\!
      \mathrlap{
        \mbox{
          \tiny
          \color{darkblue}
          \bf
          \begin{tabular}{c}
            twistor
            \\
            fibration
          \end{tabular}
        }
      }
    }
    \ar[ddr]
   \; \ar@{^{(}->}[rr]
    &&
    \overset{
      \mathclap{
      \raisebox{3pt}{
        \tiny
        \color{darkblue}
        \bf
        \begin{tabular}{c}
          $\mathbb{C}P^3$-fibration
          \\
          over spacetime
        \end{tabular}
      }
      }
    }{
      \widetilde E
    }
   \; \ar@{^{(}->}[rr]
    \ar[ddr]|>>>>>>>>>>>>>>>>>>>{ \phantom{AA} }
    \ar[dr]
    &
    &
    \mathbb{C}P^3 \!\sslash\! \mathrm{Sp}(2)
    \mathrlap{
    \mbox{
      \tiny
      \color{darkblue}
      \bf
      \begin{tabular}{c}
        $\mathbb{C}P^3$-fibration over
        \\
        classifying space
      \end{tabular}
    }
    }
    \ar[ddr]|>>>>>>>>>>>>>>>>>>>>>>{ \phantom{AAAA} }
    \ar[dr]|-{
      \;\;\;\;\;\;\;\;
      t_{\mathbb{H}} \sslash \mathrm{Sp}(2)
      \mathrlap{
        \mbox{
          \tiny
          \color{darkblue}
          \bf
          \begin{tabular}{c}
            universal fiberwise
            \\
            twistor fibration
          \end{tabular}
        }
      }
    }
    \\
    &
    S^4
    \ar[d]
    \ar[rr]
    \ar@{}[drr]|-{\mbox{\tiny(pb)}}
    &
    &
    \overset{
      \mathrlap{
      \;\;\;\;\;
      \mbox{
        \tiny
        \color{darkblue}
        \bf
        \begin{tabular}{c}
          $S^4$-fibration
          \\
          over spacetime
        \end{tabular}
      }
      }
    }{
      E
    }
    \ar[rr]
    \ar[d]
    \ar@{}[drr]|-{ \mbox{\tiny(pb)} }
    &&
    S^4 \!\sslash\! \mathrm{Sp}(2)
    \mathrlap{
    \mbox{
      \tiny
      \color{darkblue}
      \bf
      \begin{tabular}{c}
        $S^4$-fibration over
        \\
        classifying space
      \end{tabular}
    }
    }
    \ar[d]
    \\
    &
    \{x\}
   \; \ar@{^{(}->}[rr]
    &
    &
    \underset{
      \mathllap{
        \mbox{
          \tiny
          \color{darkblue}
          \bf
          spacetime
        }
        \;\;\;\;\;
      }
    }{
      X
    }
    \ar[dr]_-{T X}
    \ar[rr]|-{\; \tau \;}_-{
      \mathclap{
      \mbox{
        \tiny
        \color{darkblue}
        \bf
        $\mathrm{Sp}(2)$-structure
      }
      }
    }
    &&
    \underset{
      \mathrlap{
        \;\;\;\;
        \mbox{
          \tiny
          \color{darkblue}
          \bf
          classifying space
        }
      }
    }{
      B \mathrm{Sp}(2)
    }
    \ar[dl]
    \\
    &&
    &
    &
    B \mathrm{Spin}(8)
  }
  }
\end{equation}

\noindent {\bf Twistorial Cohomotopy theory.}
A section $(c,a)$ of
the $\mathbb{C}P^3$-fibration $\widetilde E$ is
a cocycle in a twisted non-abelian cohomology theory
\eqref{TwistedNonAbelianGeneralizedCohomology}, which we call
\emph{Twistorial Cohomotopy theory} \footnote{
  Not to be confused with
  \emph{twistor cohomology} (see, e.g., \cite{EPW81}).
  The latter is abelian cohomology
  \emph{of} twistor space, while Twistorial Cohomotopy is
  non-abelian cohomology with
  coefficients \emph{in} (Borel-equivariantized) twistor space,
  hence with cocycles being
  maps \emph{into} twistor space.
} of $X$.
Notice that, as in \eqref{TwistedNonAbelianGeneralizedCohomology},
such a section is equivalently a lift of
the classifying map $\tau$ to the parametrized twistor space:

\vspace{-2mm}
$$
\hspace{1cm}
  \raisebox{20pt}{
  \xymatrix{
    \widetilde E
    \ar[rr]
    \ar[d]
    \ar@{}[drr]|-{\mbox{\tiny(pb)}}
    &&
    \mathbb{C}P^3 \!\sslash\! \mathrm{Sp}(2)
    \ar[d]
    \\
    X
    \ar@{-->}@/^1.3pc/[u]^-{
      \mathllap{
        \mbox{
          \tiny
          \color{greenii}
          \bf
          \begin{tabular}{c}
            section of
            \\
            $\tau$-associated
            \\
            twistor fibration
          \end{tabular}
        }
      }
      (c,a)
    }
    \ar[rr]|-{ \; \tau  \;}
    &&
    B \mathrm{Sp}(2)
  }
  }
  \phantom{AAA}
    \Leftrightarrow
  \phantom{AAAAAAAAA}
  \raisebox{20pt}{
  \xymatrix{
    &&
    \mathbb{C}P^3 \!\sslash\! \mathrm{Sp}(2)
    \ar[d]
    \\
    X
    \ar@{-->}[urr]^-{
      \mathllap{
        \mbox{
          \tiny
          \color{greenii}
          \bf
          \begin{tabular}{c}
            lift of $\tau$ to
            \\
            universally parametrized
            \\
            twistor space
          \end{tabular}
        }
      }
      (c,a)
    }
    \ar[rr]|-{\; \tau \;}
    &&
    B \mathrm{Sp}(2)
  }
  }
$$
We write
\begin{equation}
  \label{JTwistedTCohomology}
    \mathclap{
    \raisebox{0pt}{
      \tiny
      \color{darkblue}
      \bf
      {\begin{tabular}{c}
            Twistorial Cohomotopy
            \\
            of tangentially
            $\mathrm{Sp}(2)$-structured
            \\
            manifold
          \end{tabular}}
     }
     \phantom{AAAAAAAAAAAAAAA}
     }
          \mathcal{T}^\tau( X)
   \;\;:=\;\;
  \left\{\!\!\!\!
  \raisebox{20pt}{
  \xymatrix@C=4em@R=1.5em{
    &&
    \overset{
      \mathllap{
      \;\;\;
      \mbox{
        \tiny
        \color{darkblue}
        \bf
        \begin{tabular}{c}
        universally parametrized
        \\
        twistor space
        \end{tabular}
      }
      }
    }{
      \mathbb{C}P^3 \!\sslash\! \mathrm{Sp}(2)
    }
    \ar[d]
    \\
    X
    \ar@{-->}[urr]^-{
      \overset{
        \raisebox{2pt}{
          \tiny
          \color{orangeii}
          \bf
          cocycle
        }
      }{
        (c, \,a)
      }
    }
    \ar[rr]|-{\;\tau \;}_-{
      \underset{
        \mathclap{
        \raisebox{-3pt}{
          \tiny
          \color{greenii}
          \bf
          $\mathrm{Sp}(2)$-structure
        }
        }
      }{
      }
    }
    &&
    B \mathrm{Sp}(2)
  }
  }
 \;  \right\}_{\!\!\!\Big/\sim}
\end{equation}
for the set of homotopy classes (relative $X$) of such sections,
and call this the \emph{cohomology set of Twistorial Cohomotopy},
when evaluated on spin-8 manifolds with tangential
$\mathrm{Sp}(2)$-structure $\tau$ \eqref{The8Manifold}.

\medskip

\noindent {\bf Twistor fibration as cohomology operation.}
Notice the direct analogy of Twistorial Cohomotopy theory \eqref{JTwistedTCohomology}
to J-twisted Cohomotopy theory \cite[2.1]{FSS19b}:
\vspace{-2mm}
\begin{equation}
  \label{JTwistedCohomotopy}
      \mathclap{
    \raisebox{0pt}{
      \tiny
      \color{darkblue}
      \bf
      {\begin{tabular}{c}
        J-twisted
                4-Cohomotopy
          \\
                of tangentially
            \\
            $\mathrm{Sp}(2)$-structured
            \\
            manifold
          \end{tabular}}
}
      \phantom{AAAAAAAAAAAAAAA}
      }
          \pi^\tau(X)
   \;\;:=\;\;
  \left\{\!\!\!\!
  \raisebox{20pt}{
  \xymatrix@C=4em@R=1.5em{
    &&
    \overset{
      \mathllap{
      \;\;\;
      \mbox{
        \tiny
        \color{darkblue}
        \bf
        \begin{tabular}{c}
        universally parametrized
        \\
        4-sphere
        \end{tabular}
      }
      }
    }{
      S^4 \!\sslash\! \mathrm{Sp}(2)
    }
    \ar[d]
    \\
    X
    \ar@{-->}[urr]^-{
      \overset{
        \raisebox{2pt}{
          \tiny
          \color{orangeii}
          \bf
          cocycle
        }
      }{
        c
      }
    }
    \ar[rr]|-{\;\tau \;}_-{
      \underset{
        \mathclap{
        \raisebox{-3pt}{
          \tiny
          \color{greenii}
          \bf
          $\mathrm{Sp}(2)$-structure
        }
        }
      }{
      }
    }
    &&
    B \mathrm{Sp}(2)
  }
  }
 \;\; \right\}_{\!\!\!\Big/\sim}
\end{equation}

\vspace{-2mm}
\noindent
and the fact that postcomposition with the parametrized twistor fibration
(Def. \ref{ParametrizedHopfTwistorFibration})
constitutes a cohomology operation (a natural transformation of
cohomology sets) between the two:
\vspace{-2mm}
\begin{equation}
  \label{CohomologyOperationInducedByTwistorFibation}
  \xymatrix@C=3em{
    \overset{
      \mathclap{
      \raisebox{3pt}{
        \tiny
        \color{darkblue}
        \bf
        \begin{tabular}{c}
          Twistorial
          \\
          Cohomotopy
        \end{tabular}
      }
      }
    }{
      \mathcal{T}^\tau
    }
  \; \;\; \ar[rrrr]_{\scalebox{.7}{$
      \big(
        t_{\mathbb{G}} \sslash \mathrm{Sp}(2)
      \big)_\ast
      $}
      }^{
         \raisebox{3pt}{
          \tiny
          \color{greenii}
          \bf
          \begin{tabular}{c}
            cohomology operation by
            \\
            parametrized
            twistor fibration
          \end{tabular}
        }
      }
       &&&&
 \;\;\;\;\;   \overset{
      \mathclap{
      \raisebox{3pt}{
        \tiny
        \color{darkblue}
        \bf
        \begin{tabular}{c}
          J-twisted
          \\
          Cohomotopy
        \end{tabular}
      }
      }
    }{
      \pi^\tau
    }
  }
\end{equation}

\noindent {\bf Chern-Dold character in Twistorial Cohomotopy.}
The Chern-Dold character \eqref{ChernCharacterCohomologyOperation}
in J-twisted 4-Cohomotopy \eqref{JTwistedCohomotopy} is discussed in some detail
\cite{FSS19b}. The following Prop. \ref{ChernCharacterInJTwistedCohomotopy}
is its generalization to Twistorial Cohomotopy:

\begin{prop}[Character map in Twistorial Cohomotopy theory]
  \label{ChernCharacterInJTwistedCohomotopy}
The twisted non-abelian character \eqref{ChernCharacterCohomologyOperation}
in Twistorial Cohomotopy \eqref{JTwistedTCohomology} is of the following form:
\vspace{-2mm}
$$
\hspace{3mm}
  \xymatrix{
    \overset{
      \mathclap{
      \raisebox{3pt}{
        \tiny
        \color{darkblue}
        \bf
        \begin{tabular}{c}
          Twistorial
          \\
          Cohomology
        \end{tabular}
      }
      }
    }{
      \mathcal{T}^\tau(X)
    }
    \;\;\;\;\;
    \ar[rr]^-{
      \mathrm{ch}_{\mathcal T}
    }_-{
      \mbox{
        \tiny
        \color{darkblue}
        \bf
        \begin{tabular}{c}
          twisted non-abelian
          \\
          Chern-Dold character
        \end{tabular}
      }
    }
    &&
        \left\{
    \!\!
    {\begin{array}{l}
      F_2,
      \\
      H_3,
      \\
      G_4,
      \\
      G_7
    \end{array}}
    \!\!
    \in
    \Omega^\bullet(X)
    \left\vert\;
    {\begin{aligned}
      d\,F_2 & = 0
      \\
      d\, H_3 & = G_4 - \tfrac{1}{4}p_1(\omega) - F_2 \wedge F_2
      \\
      d\, G_4 & = 0
      \\
      d \, G_7 & =
      -
      \tfrac{1}{2}
      \big(
        G_4 - \tfrac{1}{4}p_1(\omega)
      \big)
      \wedge
      \big(
        G_4 + \tfrac{1}{4}p_1(\omega)
      \big)
      -
      \tfrac{1}{4}
      \big(
        p_2 - (\tfrac{1}{2}p_1(\omega))^2
      \big)
    \end{aligned}}
    \right.
    \!\!
    \right\}_{\!\!\!\!\!\!\Big/\sim}
  }
$$
\vspace{-.6cm}
\begin{equation}
  \label{TwistorialChernCharacter}
  \hspace{-12.2cm}
  \xymatrix{
    {\phantom{AAAAAAAA}} (c,a) \;\;\;\;\;\;\;
    \ar@{|->}[rr]
    &&
    {\phantom{A}}
    (c,a)^\ast
    \left(
    \!\!\!
    {\begin{array}{c}
      f_2,
      \\
      h_3,
      \\
      \omega_4,
      \\
      \omega_7
    \end{array}}
    \!\!\!
    \right)
  }
\end{equation}
\end{prop}
\begin{proof}
By Theorem \ref{SullivanModelOfParametrizedHopfTwistorFibrations}
the class of a section of the parametrized twistor fibration
in rational homotopy theory is given
equivalently by a dg-algebra homomorphism
shown as the dashed arrow in the following diagram:
\vspace{-2mm}
\begin{equation}
  \label{SullivanModelOfSectionOfParametrizedTwistorSpaceFibration}
  \hspace{-1mm}
  \xymatrix@C=3pt{
    &&
    \mathbb{C}P^3 \!\sslash\! \mathrm{Sp}(2)
    \ar[dd]^-{ t_{\mathbb{H}} \sslash \mathrm{Sp}(2) }
    &&
    &&
    \mathrm{CE}\big(
      \mathfrak{l}
      B \mathrm{Sp}(2)
    \big)
    \!\!
    \left[
      \!\!\!
      {\begin{array}{c}
        f_2,
        \\
        h_3,
        \\
        \omega_4,
        \\
        \omega_7
      \end{array}}
      \!\!\!
    \right]
    \!\!\Big/\!\!
    \left(
    \!\!
    {\begin{aligned}
      d\, f_2 & = 0
      \\
      d \, h_3 & =
      \omega_4  - \tfrac{1}{4}p_1 - f_2 \wedge f_2
      \\
      d\,\omega_4 & = 0
      \\
      d\,\omega_7
        &
        =
        -\big(
          \omega_4 - \tfrac{1}{4}p_1
        \big)
        \wedge
        \big(
          \omega_4 + \tfrac{1}{4}p_1
        \big)
        \\
        &
        \phantom{=} \; -
        \tfrac{1}{2}
        \big(
          p_2 - (\tfrac{1}{2}p_1)^2
        \big)
    \end{aligned}}
    \!\!\!
    \right)
    \ar@{<-}[dd]|-{
      \scalebox{.7}{$
        \arraycolsep=3pt
        \begin{array}{cc}
          \omega_4 & \omega_7
          \\
          \mapsup & \mapsup
          \\
          \omega_4 & \omega_7
        \end{array}
      $}
    }
    \ar@{-->}[ddll]|-{
      \overset{
        \raisebox{3pt}{
          \tiny
          \color{orangeii}
          \bf
          \begin{tabular}{c}
            rational dg-model
            \\
            for cocycle in
            \\
            twistorial Cohomotopy
          \end{tabular}
        }
      }{
      \scalebox{.7}{$
        \color{darkblue}
        \begin{aligned}
          F_2 & \mapsfrom f_2
          \\
          H_3 & \mapsfrom h_3
          \\
          G_4 & \mapsfrom \omega_4
          \\
          2G_7 & \mapsfrom \omega_7
          \\
          \phantom{A}
              \\
          \phantom{A}
              \\
          \phantom{A}
              \\
          \phantom{A}
        \end{aligned}
      $}
      }
    }
    \\
    \\
    X
    \ar@{-->}[rr]^-{
             \raisebox{3pt}{
          \tiny
          \color{greenii}
          \bf
          \begin{tabular}{c}
            cocycle in
            \\
            twisted Cohomotopy
          \end{tabular}
        }
       }_-{
       c
      }
    \ar@{-->}[uurr]^-{
      \overset{
        \mathllap{
          \raisebox{3pt}{
            \tiny
            \color{orangeii}
            \bf
            \begin{tabular}{c}
              cocycle in
              \\
              twistorial Cohomotopy
            \end{tabular}
          }
          \;\;\;\;\;\;
        }
      }{
        (c,a)
      }
    }
    \ar[dr]^-{ \; \tau \; }
    \ar[d]_-{ T X }
    &&
    S^4 \!\sslash\! \mathrm{Sp}(2)
    \ar[dl]
    &&
    \Omega^\bullet(X)
    \ar@{<-}[rr]
    \ar@{<-}[dr]_-{
         \scalebox{.6}{$
        \arraycolsep=3pt
        \begin{array}{cc}
             \tfrac{1}{2}p_1(\omega)
          &
          \rchi_8(\omega)
          \\
          \mapsup & \mapsup
          \\
          \tfrac{1}{2}p_1
          &
          \rchi_8
             \end{array}
      $}
    }
        &&
   \mathrm{CE}\big(
      \mathfrak{l}
      B \mathrm{Sp}(2)
    \big)
    \!\!
    \left[
      \!\!\!
      {\begin{array}{c}
        \omega_4,
        \\
        \omega_7
      \end{array}}
      \!\!\!
    \right]
    \!\!\Big/\!\!
    \left(
    \!
    {\begin{aligned}
      d\,\omega_4 & = 0
      \\
      d\,\omega_7
        &
        =
        -\big(
          \omega_4 - \tfrac{1}{4}p_1
        \big)
        \wedge
        \big(
          \omega_4 + \tfrac{1}{4}p_1
        \big)
        \\
        &
        \phantom{=} \; -
        \tfrac{1}{2}
        \big(
          p_2 - (\tfrac{1}{2}p_1)^2
        \big)
    \end{aligned}}
    \!\!\!
    \right)
    \ar@{<-_{)}}[dl]
    \\
    B \mathrm{Spin}(8)
    \ar@{}[r]|{\;\;\leftarrow}
    &
  B \mathrm{Sp}(2)
    &
    &&
    &
   \mathclap{
    \mathrm{CE}
    \big(
      \mathfrak{l}
      B \mathrm{Sp}(2)
    \big)
    }
    {\phantom{\vert^{\vert}\vert^{\vert}\vert^{\vert}\vert^{\vert}\vert^{\vert}}}
  }
\end{equation}
Here the dg-algebras on the right are the Sullivan model for the
Borel-equivariant twistor fibration \eqref{SullivanModelForParametrizedHopfTwistorFibrations}
from Theorem \ref{SullivanModelOfParametrizedHopfTwistorFibrations}.
These being Sullivan models means that they are cofibrant as
dg-algebras, which implies that all homotopy classes of rational
sections are indeed represented this way. Therefore, a rational section
is specified by the differential forms on $X$ to which it pulls back
the generators on the right. The condition for any such set of
differential forms to arise this way is that it satisfies
the same differential relations as the generators, now in the de Rham
dg-algebra $\Omega^\bullet(X)$. This way the relation $d f_2 = 0$
in the Sullivan model pulls back to the relation
$d F_2 = 0$ in $\Omega^\bullet(X)$ in \eqref{SullivanModelOfSectionOfParametrizedTwistorSpaceFibration},
etc.
\hfill \end{proof}

For use below, we record the de Rham-cohomological relations
implied by the differential relations \eqref{TwistorialChernCharacter}:
\begin{cor}[Cohomological relations in Twistorial Cohomotopy]
\label{CohomologicalRelationInTwistorialCohomotopy}
For $X$ an 8-manifold with tangential $\mathrm{Sp}(2)$-structure \eqref{The8Manifold},
let $F_2, H_3, G_4, G_7 \in \Omega^\bullet(X)$ be differential form components
in the image of the Chern-Dold character in Twistorial Cohomotopy on $X$
(Def. \ref{ChernCharacterInJTwistedCohomotopy}). Then the real/de Rham
cohomology classes these represent satisfy the following relations:
\vspace{-3mm}
\begin{align}
  \label{CohomologyRelationInDegree4}
  [G_4]
  -
  \tfrac{1}{4}p_1
  \;=\;
  [F_2 \wedge F_2]
  \;\;
  &
  \in \;\;
  H^4
  (
    X, \mathbb{R}
  )\;, \phantom{AAAA}
  \\
  \label{CohomologyRelationInDegree8}
  0
  \;=\;
  \big(
    [F_2 \wedge F_2]
    +
    \tfrac{1}{2}p_1
  \big)
  \cup
  [F_2 \wedge F_2]
  \;+\;
  \tfrac{1}{2}
  \big(
    p_2  - \tfrac{1}{4} p_1 \cup p_1
  \big)
  \;\;
  &
  \in \;\;
  H^8
  (
    X, \mathbb{R}
  )\;. \phantom{AAAA}
\end{align}
\end{cor}
\begin{proof}
Equation \eqref{CohomologyRelationInDegree4} is the direct consequence of the second line in
\eqref{TwistorialChernCharacter}.
From the fourth line of \eqref{TwistorialChernCharacter} we similarly get the
relation
\begin{equation}
  \label{Degree8CohomologyRelation}
  - [G_4 \wedge G_4]
  + \tfrac{1}{16} p_1 \cup p_1
  - \rchi_8
  \;= 0\;
\end{equation}
Plugging \eqref{CohomologyRelationInDegree4}
and \eqref{chi8RelationOnBSp2} into \eqref{Degree8CohomologyRelation}
yields \eqref{CohomologyRelationInDegree8}.
\hfill \end{proof}

\noindent {\bf Charge quantization in Twistorial Cohomotopy.}
Finally we obtain the claimed result \eqref{HeteroticFluxQuantizationFromTwistorialCohomotopy}:
\begin{cor}[Shifted integrality of $G_4$, $F_2$ in Twistorial Cohomotopy]
  \label{ShiftedIntegralityUnderTTheoreticChernCharacter}
  Let $X$ be a spin 8-manifold with tangential
  $\mathrm{Sp}(2)$-structure $\tau$ \eqref{The8Manifold}.
  Then differential form data $(F_2, H_3, G_4, G_7) \in \Omega^\bullet(X)$
  which is in the image \eqref{TwistorialChernCharacter}
  of the Chern-Dold character
  from Prop. \ref{ChernCharacterInJTwistedCohomotopy}, hence which is
  charge-quantized (Remark \ref{ChargeQuantizationByLiftThroughChernCharacter})
  in Twistorial Cohomotopy \eqref{JTwistedTCohomology},
  satisfies the following integrality conditions:

  \noindent
  {\bf (i)} The class of $G_4$ shifted by $\tfrac{1}{4}p_1(\omega)$ is integral,
  hence is the image in real cohomology of a class in integral cohomology:
  \vspace{-2mm}
  \begin{equation}
    \label{ShiftedQuantizationConditionOnG4InJTwistedTTheory}
    [G_4 + \tfrac{1}{4}p_1(\omega)]
    \;\in\;
    \xymatrix{
      H^4
      \big(
        X, \mathbb{Z}
      \big)
      \ar[r]
      &
      H^4
      (
        X, \mathbb{R}
      )
    }.
  \end{equation}

  \noindent {\bf (ii)} The class of $F_2$ is integral:
  \begin{equation}
    [F_2]
    \;\in\;
    \xymatrix{
      H^2
      (
        X, \mathbb{Z}
      )
      \ar[r]
      &
      H^2
      (
        X, \mathbb{R}
      )\;.
    }
  \end{equation}

  \noindent {\bf (iii)} Hence the relation \eqref{CohomologyRelationInDegree4}
  is the image
  of such a relation in integral cohomology:
  \begin{equation}
    \label{HeteroticShiftedFluxQuantizationInTTheory}
    [
      G_4
      -
      \tfrac{1}{4} p_1(\omega)
    ]
    \;=\;
    [F_2 \wedge F_2]
    \;
    \in
    \xymatrix{
      H^4
      \big(
        X, \mathbb{Z}
      \big)
      \ar[r]
      &
      H^4
      (
        X, \mathbb{R}
      )\;.
    }
  \end{equation}
\end{cor}
\begin{proof}
  By Prop. \ref{ChernCharacterInJTwistedCohomotopy}
  these de Rham classes are pullbacks of the generators
  in the Sullivan model from Theorem \ref{SullivanModelOfParametrizedHopfTwistorFibrations}.
  By the normalization \eqref{IdentiftingRationlGeneratorsOnSP3ModSp2}
  there, the statement hence follows with Theorem \ref{CohomologyOfBorelEquivariantTwistorFibration}.
\hfill \end{proof}

In fact, we have a stronger statement:
\begin{remark}[Cochain-level model of the C-field]
While Corollary \ref{ShiftedIntegralityUnderTTheoreticChernCharacter}
produces Ho{\v r}ava-Witten's identity \eqref{HeteroticFluxQuantizationFromTwistorialCohomotopy}
between the cohomology classes related to the
C-field in heterotic M-theory,
the twistorial character map from Prop. \ref{ChernCharacterInJTwistedCohomotopy}
gives a little more
information, namely an explicit differential form (cochain) model
for these cohomology classes. Incidentally,
this cochain expression for the C-field,
$$
  G_4 \;=\;  \tfrac{1}{4}p_1(\omega) - c_2(A) + d H_3
$$
as obtained from twistorial Cohomotopy in the second line of
\eqref{TwistorialChernCharacter}
(and from differential twistorial Cohomotopy in \cite[(296)]{FSS20c}), coincides with the
proposed model for the C-field in \cite[(3.9)]{DFM03}
(under identifying our $H_3$ with minus their $c$ and
our $G_4$ with minus their $G$).
\end{remark}

\medskip

\appendix

\section{Quaternion-linear groups}

For reference, we record some basics of
quaternion-linear groups:

\begin{defn}[Special quaternion-linear group]
  \label{SpecialQuaternionLinearGroup}
  The special quaternion-linear group
  \vspace{-2mm}
  \begin{equation}
    \label{DefSL2H}
    \mathrm{SL}(2,\mathbb{H})
    \;:=\;
    \big\{
      \left.
      \!\!
      G
      \;\in\;
      \mathrm{Mat}
      (2 \times 2, \mathbb{H})
    \;\right\vert\;
    \mathrm{det}_{\mathrm{Di}}(G)
    \;=\;
    1
    \big\}
  \end{equation}

    \vspace{-2mm}
\noindent  is the group of $2 \times 2$ quaternionic matrices
  with unit \emph{Dieudonn{\'e} determinant}
  \cite{Dieudonne43} (review in \cite{Aslaken96}\cite[\S 1]{VenanciaBatista20}).
\end{defn}
\begin{remark}[Size of $\mathrm{SL}(2,\mathbb{H})$]
  When restricted along the inclusion of complex matrices into quaternionic matrices
    \vspace{-2mm}
  $$
    \xymatrix{
      \mathrm{Mat}(2 \times 2, \mathbb{C})
     \; \ar@{^{(}->}[r]^-{i_{\mathbb{C}}}
      &
      \mathrm{Mat}(2 \times 2, \mathbb{H})
    }
  $$
  the Dieudonn{\'e} determinant does \emph{not} reduce to the
  ordinary determinant, but to its absolute value:
    \vspace{-2mm}
  \begin{equation}
    \label{DieudonneDeterminantOnComplexMatrices}
    \mathrm{det}_{\mathrm{Di}}
    \big(
      i_{\mathbb{C}}(A)
    \big)
    \;=\;
    \left\Vert
      \mathrm{det}(A)
    \right\Vert
    \,.
  \end{equation}

    \vspace{-2mm}
\noindent  Accordingly, $\mathrm{SL}(2,\mathbb{H})$ (Def. \ref{SpecialQuaternionLinearGroup})
  is larger than the notation might suggest:
  For instance, it follows immediately
  from \eqref{DieudonneDeterminantOnComplexMatrices} that
  all complex unitary matrices have unit
  Dieudonn{\'e} determinant. In fact,
  Example \ref{SubgroupsOfQuaternionLinearGroups} says that
  the full quaternion-unitary group (Def. \ref{QuaternionicGroups})
  is contained in $\mathrm{SL}(2,\mathbb{H})$ \eqref{Sp2InSL2H}
  (and hence coincides with what would otherwise be called
  $\mathrm{SU}(2,\mathbb{H})$).
\end{remark}
\begin{defn}[Unitary quaternion-linear groups]
  \label{QuaternionicGroups}
  Let $n \in \mathbb{N}$.
$\,$
\vspace{-2mm}
\item  {\bf (i) } The $n \times n$ \emph{quaternionic unitary group} is
\vspace{-2mm}
  \begin{equation}
    \label{Spn}
    \mathrm{Sp}(n)
    \;:=\;
    \mathrm{U}(n,\mathbb{H})
    \;:=\;
    \big\{
      \left.
      G \in \mathrm{GL}(n, \mathbb{H})
      \,\right\vert\,
      G \cdot G^\dagger \;=\; 1
    \big\}
    \,,
  \end{equation}

  \vspace{-2mm}
\noindent  where $(-)^\dagger$ denotes matrix transpose combined with quaternionic
  conjugation.

\item  {\bf (ii)} The \emph{central product group} of $\mathrm{Sp}(n_1)$
  with $\mathrm{Sp}(n_2)$ is
  \vspace{-2mm}
  \begin{equation}
    \label{CentralProductOfQuaternionUnitaryGroups}
    \mathrm{Sp}(n_1)
    \cdot
    \mathrm{Sp}(n_2)
    \;:=\;
    \big( \mathrm{Sp}(n_1) \times \mathrm{Sp}(n_2) \big)/
    \underset{
      \simeq \, \mathbb{Z}_2
    }{
      \underbrace{
        \{(1,1), (-1,-1)\}
      }
    }
  \end{equation}

\vspace{-4mm}
\end{defn}
\begin{example}[Subgroups of quaternion-linear groups]
  \label{SubgroupsOfQuaternionLinearGroups}
  We have the following canonical subgroup inclusions into
  special quaternion-linear (Def. \ref{SpecialQuaternionLinearGroup})
  and unitary quaternion-linear groups (Def. \ref{QuaternionicGroups}):

  \noindent
  {\bf (i)} The algebra inclusion of the complex numbers into the
  quaternions induces:
  \vspace{-2mm}
\begin{equation}
  \xymatrix@R=2pt{
    {\phantom{\mathrm{U}(n,}}
    \mathbb{C}
    \;\;\;
        \ar@{^{(}->}[r]
    &
    \!\!{\phantom{\mathrm{U}(n,}}
    \mathbb{H}
    \\
        \mathrm{U}(n,\mathbb{C})
    \;
    \ar@{}[dd]|-{
      \rotatebox[origin=c]{-90}{
        $=$
      }
    }
    \ar@{^{(}->}[r]
    &
    \mathrm{U}(n, \mathbb{H})
     \ar@{}[dd]|-{
      \rotatebox[origin=c]{-90}{
        $=$
      }
    }
   \\
   \\
    \mathrm{U}(n)
    \;
    \ar@{^{(}->}[r]
    &
    \mathrm{Sp}(n)
  }
\end{equation}

\vspace{-2mm}
\noindent {\bf (ii)}
We write
\begin{equation}
  \label{LeftRightSp1Inclusion}
  \xymatrix@R=-2pt{
    \mathrm{Sp}(1)_L
    \times
    \mathrm{Sp}(1)_R
    \;
    \ar@{^{(}->}[r]
    &
    \mathrm{Sp}(2)
    \\
    (q_L, q_R)
    \ar@{|->}[r]
    &
    \mathrm{diag}(q_L, q_R)
  }
\end{equation}
for the subgroup of $\mathrm{Sp}(2)$ given by the diagonal
matrices with coefficients in unit-norm quaternions $q$, hence the
direct product group of two copies of $\mathrm{Sp}(1)$,
equipped with their left and right factor embedding, as indicated.

  \noindent
  {\bf (iii)}
  The unitary quaternion-linear $2 \times 2$-matrices
  (Def. \ref{QuaternionicGroups}) have Dieudonn{\'e}-determinant
  (Def. \ref{SpecialQuaternionLinearGroup}) equal to 1 {\cite[6.4]{CohenDeLeo00}}
  and hence include into the special quaternion-linear group:
    \vspace{-2mm}
  \begin{equation}
    \label{Sp2InSL2H}
    \mathrm{Sp}(2)
    \;=\;
    \mathrm{U}(2,\mathbb{H})
    \;\subset\;
    \mathrm{SL}(2, \mathbb{H})\;.
  \end{equation}

\vspace{-2mm}
\noindent  {\bf (iv)} There is the canonical subgroup inclusion
of symplectic-unitary groups into their central product groups
\eqref{CentralProductOfQuaternionUnitaryGroups}
\vspace{-2mm}
  \begin{equation}
    \label{CanonicalInclusionIntoCentralProductGroup}
    \xymatrix@R=-2pt{
      \mathrm{Sp}(n_1)
      \; \ar@{^{(}->}[rr]
      &&
      \mathrm{Sp}(n_1) \cdot \mathrm{Sp}(n_2)
      \\
      A
        \ar@{|->}[rr]
      &&
      [A,1]
    }
  \end{equation}

\end{example}

\end{document}